\documentclass[journal]{IEEEtran}

\usepackage[utf8]{inputenc}
\usepackage[english]{babel}
\usepackage{cite}
\usepackage{amsmath,amssymb,amsfonts}
\usepackage{subfig}    
\usepackage{textcomp}
\usepackage{xcolor}
\usepackage{mathrsfs}
\usepackage[utf8]{inputenc}
\usepackage{comment}
\usepackage[ruled,vlined]{algorithm2e}
\usepackage{mathtools}
\usepackage[font=small,labelfont=bf]{caption}
\usepackage{blkarray, bigstrut}
\usepackage{subfig}
\usepackage{enumerate}

\ifCLASSINFOpdf
  \usepackage[pdftex]{graphicx}
\else
\fi

\usepackage{hyperref}
\hypersetup{
    colorlinks=true,
    linkcolor=blue,
    filecolor=blue,      
    urlcolor=blue,
    citecolor=blue
}
\newcommand{\floor}[1]{\lfloor #1 \rfloor}

\SetCommentSty{mycommfont}
\SetKwInput{KwInput}{Input}                
\SetKwInput{KwOutput}{Output}              

\usepackage{enumitem}
\usepackage{float}

\usepackage{pgf, tikz}
\usetikzlibrary{arrows, automata}
\usepackage{tikz}
\usetikzlibrary{automata,positioning}

\newcommand*\circled[2][1.6]{\tikz[baseline=(char.base)]{
    \node[shape=circle, draw, inner sep=1pt, 
        minimum height={\f@size*#1},] (char) {\vphantom{WAH1g}#2};}}

\newlist{legal}{enumerate}{10}
\setlist[legal]{label*=\arabic*.}
\usepackage{ifpdf}

\usepackage{float}

\usepackage{amsthm}

\newtheoremstyle{mystyle1}
  {.5pt}
  {.5pt}
  {}
  {}
  {\bfseries}
  {.}
  { }
  {\thmname{#1}\thmnumber{ #2}\thmnote{ (#3)}}

\theoremstyle{mystyle1}

\newtheorem{defn}{Definition}
\newtheorem{rem}{Remark}
\newtheorem{exmp}{Example}

\newtheoremstyle{mystyle}
  {.5pt}
  {.5pt}
  {\itshape}
  {}
  {\bfseries}
  {.}
  { }
  {\thmname{#1}\thmnumber{ #2}\thmnote{ (#3)}}

 \theoremstyle{mystyle}

\newtheorem{thm}{Theorem}
\newtheorem{lem}{Lemma}
\newtheorem{prop}{Proposition}
\newtheorem{conj}{Conjecture}

\usepackage[utf8]{inputenc}
\usepackage{cite}
\usepackage{amsmath,amssymb,amsfonts}
\usepackage{algorithmic}
\usepackage{graphicx}
\usepackage{textcomp}
\usepackage{xcolor}
\usepackage{mathrsfs} 
\usepackage{enumitem}
\usepackage{amsthm}
\usepackage{cite}
\usepackage{array}
\usepackage{amsmath}

\setlist[legal]{label*=\arabic*.}

\theoremstyle{mystyle1}

\usepackage{enumitem}
\usepackage{geometry}
\begin{document}

\title{An Update-based Maximum Column Distance Coding Scheme for Index Coding}

\author{\IEEEauthorblockN{Arman Sharififar, Neda Aboutorab, Parastoo Sadeghi}
\\
\IEEEauthorblockA{\textit{\normalsize School of Engineering and Information Technology}, \\
\textit{\normalsize	University of New South Wales, Australia} \\
\normalsize	Email:\{a.sharififar, n.aboutorab,  p.sadeghi\}@unsw.edu.au
}
}

\maketitle
\vspace{1mm}

\begin{abstract}

In this paper, we propose a new scalar linear coding scheme for the index coding problem called update-based maximum column distance (UMCD) coding scheme. The central idea in each transmission is to code messages such that one of the receivers with the minimum size of side information is instantaneously eliminated from unsatisfied receivers. One main contribution of the paper is to prove that the other satisfied receivers can be identified after each transmission, using a polynomial-time algorithm solving the well-known maximum cardinality matching problem in graph theory. This leads to determining the total number of transmissions without knowing the coding coefficients. Once this number and what messages to transmit in each round are found, we then propose a method to determine all coding coefficients from a sufficiently large finite field. We provide concrete instances where the proposed UMCD coding scheme has a better broadcast performance compared to the most efficient existing coding schemes, including the recursive scheme (Arbabjolfaei and Kim, 2014) and the interlinked-cycle cover (ICC) scheme (Thapa  \textit{et al.}, 2017). We prove that the proposed UMCD coding scheme performs at least as well as the MDS coding scheme in terms of broadcast rate.
By characterizing two classes of index coding instances, we show that the gap between the broadcast rates of the recursive and ICC schemes and the UMCD scheme grows linearly with the number of messages. Then, we extend the UMCD coding scheme to its vector version by applying it as a basic coding block to solve the subinstances.\footnote{Preliminary results of this paper are presented, in part, in \cite{Sharififar2021UMCDconference}.}
\end{abstract}


\section{introduction}
\IEEEPARstart{I}{ndex} coding problem (introduced by Birk and Kol \cite{Birk1998}) models an efficient communication system where a single server broadcasts a set of $m$ messages via a noiseless channel to multiple receivers, each demanding a specific message while they may know some other messages a priori as their side information. Exploiting the side information of the receivers, the server can reduce the number of transmissions to satisfy all the receivers by sending coded messages rather than uncoded transmissions. 
Take a simple instance of index coding problem, depicted in Figure \ref{fig:toy-exm} as an example, where the server wishes to satisfy the three receivers. While a trivial solution is to send each message one-by-one (uncoded scheme) in a total of three transmissions, by taking advantage of the receivers' side information and sending two coded messages $x_1+x_2$ and $x_1+x_3$, each receiver is able to decode its desired message. For this instance, one transmission is saved thanks to both the receivers' side information and encoding the messages at the server. 
The main objective of index coding problem is to design efficient coding schemes for any arbitrary index-coding instance so as to minimize the overall number of transmissions, which is still an open problem.

\begin{figure}
    \centering
    \includegraphics [scale=.45]{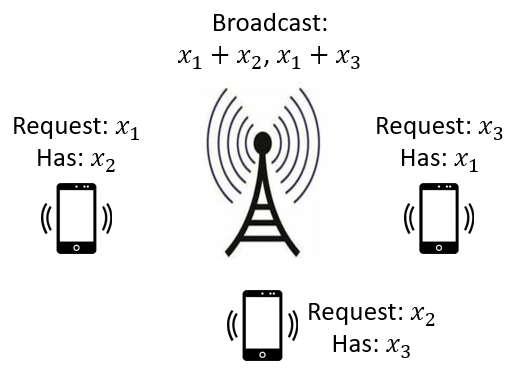}
    \caption{A simple index coding instance which shows that encoding the messages can take advantage of the receivers' side information to save one transmission.}
    \label{fig:toy-exm}
\end{figure}

Index coding problem has so far been extensively studied in the literature and various coding schemes have been proposed. However, to the best of our knowledge, the proposed coding scheme in this paper is the only scheme that considers updating the problem at each step of transmission for the index coding problem. This brings about several advantages. First, it can lead to a lower broadcast rate as will be illustrated through several instances in this paper. Second, this could reduce the average decoding delay due to both the lower broadcast rate and satisfying at least one receiver instantaneously at each step of transmission. Third, in a dynamic system where some receivers may be added to or removed from the system, the update-based scheme can be adapted so that it can deal with the updated system.

The index codes can broadly be categorized into linear and nonlinear codes. Although it has been shown that optimal linear coding can be outperformed by nonlinear codes \cite{Dougherty2005,sharififar2021broadcast, Arman-4-2, Arman-4-3}, owing to its simple and straightforward encoding and decoding processes, it has attracted considerable attention in the literature. While in scalar linear coding each specific message is considered as one variable, so that it is encoded using only one function, in vector linear coding each message can be decomposed into submessages (multi-variables), where each submessage can be encoded by different functions. This can lead to a lower broadcast rate for many index coding instances, but at the cost of increasing the computational complexity.

The existing structured coding schemes can be classified into the following five categories.
\begin{itemize}[leftmargin=*]
    \item Minrank coding scheme \cite{Maleki2014}: The minrank scheme is a combinatorial optimization problem, where its solution can give the optimal linear code over a predetermined finite field size for any instance of index coding problem. However, this comes with two main drawbacks. First, this predetermined field size can lead to a linear code where its rate is far from the optimal linear rate, which is achievable over another field size. In fact, it has been shown in \cite{Lubetzky2009} that for any field size, there is an explicit way of constructing index coding instances, where the minrank can have a much better performance over another selected field size. Second, the computational complexity of the minrank scheme, especially for the vector linear coding is considerably high and intractable. This is why various other methods have so far been introduced for designing lower-complexity index coding schemes.
    \item Composite coding scheme \cite{Arbabjolfaeicomposite,liu2018three,Liu2020}: Inspired by the random coding idea in information theory, in the composite coding scheme, every subset of the message set is randomly mapped to a new composite index set. While the scheme can be used to characterize an achievable rate region, it cannot be used to construct any code due to its random coding nature.
    \item Interference alignment coding scheme \cite{Maleki2014}: Inspired by the methods for solving the interference problem in wireless channels, the interference alignment coding scheme was proposed in the context of index coding. This scheme aims to compress the linear space spanned by the interference messages of each receiver as far as all the desired messages are still decodable. This resulted in two main techniques, namely one-to-one alignment and subspace alignment only for index coding instances with certain structures. However, a systematic method is not available for solving any arbitrary instance of the index coding problem.
    \item Graph-based coding schemes: Since any instance of the index coding problem can be represented as a directed graph, well-known techniques in graph theory have been employed to provide different coding schemes, including the cycle cover \cite{Chaudhry} and the clique cover \cite{Birk2006} schemes. The interlinked-cycle cover (ICC) scheme  was proposed in \cite{Thapa2017}, which includes the cycle and clique cover schemes as a special case. 
    \item MDS-based coding schemes: In the partial clique cover (PCC) scheme \cite{Birk1998}, the index coding instance is first partitioned into subinstances, then each subinstance is solved using the maximum distance separable (MDS) scheme. The vector version of the (PCC) scheme, namely the fractional partial clique cover (FPCC) scheme \cite{Yu2014}, can be achieved by time-sharing over the solution of each subinstance, which results in a lower broadcast rate for many instances. The recursive scheme \cite{Arbabjolfaei2014} is an extension of the FPCC algorithm, in which the local rate of the MDS code is recursively calculated for each subinstance at each stage, which strictly improves upon the FPCC scheme.
\end{itemize} 
In this paper, the index coding problem is approached from a different perspective. In the beginning, the receivers are sorted based on the size of their side information. Then, in each transmission, a linear combination of the messages is designed to instantaneously satisfy one of the receivers with the minimum size of side information. Then, the problem is updated by eliminating all receivers who are able to decode their requested message from the coded messages received so far along with the messages in their side information. This process is repeated until all receivers can successfully decode their requested message.

To design an update-based scheme, the following two questions must be addressed at each step of the transmission. How to design the coding coefficients of the messages? And how to determine whether a receiver can decode its requested message from the information available to it?
\\
Similar to the MDS-based index coding schemes, the proposed UMCD  code can be employed as a modular code for solving subinstances.
We know that in the encoding matrix of the MDS coding scheme with size $k\times m$ where $k\leq m$, each square submatrix of size $k\times k$ is full-rank. Inspired by this idea, 
for the former question, the coefficients of the messages in the UMCD coding scheme are designed so that in the encoding matrix, each column is linearly independent from the space spanned by any other columns to the extent possible. 
For the latter question, this linear independence property is used to prove that the problem of identifying the receivers who are able to decode their requested message at each stage of transmission is equivalent to a well-known problem in graph theory called the maximum cardinality matching (MCM) problem, which can be solved in polynomial-time using the Hopcroft-Karp algorithm \cite{Hopcroft1971}. This leads to determining the broadcast rate of the proposed UMCD coding scheme independent of knowing the exact coefficients of the encoding matrix. Once the broadcast rate of the UMCD coding scheme is found by the MCM algorithm in the first phase, in the next phase, the maximum column distance (MCD) algorithm is proposed to design the coefficients of the encoding matrix from a sufficiently large finite field, such that each subset of its columns will be linearly independent as much as possible (the complexity of the MCD algorithm is in general exponential). That is why our proposed coding scheme is called update-based maximum column distance (UMCD) scheme.

The MCM algorithm is significantly more efficient than the MCD algorithm in terms of computational complexity. Thus, separating the UMCD coding scheme into two phases, (i) finding its broadcast rate using the MCM algorithm and (ii) generating its encoding matrix using the MCD algorithm, reduces the complexity of the UMCD coding scheme. This complexity reduction will be more notable for the vector version of the UMCD coding scheme. This is because, first, the optimal solution of the vector version of the UMCD scheme  will be achieved only by the broadcast rate (not the encoding matrix) of the UMCD coding scheme for subinstances, and second, the broadcast rate of the UMCD coding scheme is obtained independently of using the MCD algorithm (see also Remark \ref{rem:complexity-extention}).

\subsection{Our Contributions}
\begin{enumerate} [leftmargin=*]
    \item We propose a new linear coding scheme, namely the UMCD  coding scheme in Algorithm \ref{alg:UMCD}. The UMCD coding scheme consists of two parts:
    \begin{itemize}
        \item  First, its (achievable) broadcast rate is determined using the MCM algorithm in polynomial-time.
        \item Second, its encoding matrix (code) is constructed using the proposed MCD algorithm, where its complexity is in general exponential.
    \end{itemize}
     We provide concrete instances where the proposed UMCD coding scheme outperforms the recursive and ICC coding schemes.
    \item We show the satisfied receivers in each transmission can be identified using a polynomial-time algorithm solving maximum cardinality matching (MCM) problem without any knowledge of the coding coefficients. This requires each column of the encoding matrix to be linearly independent of the space spanned by other columns as much as possible.
    \item In Algorithm \ref{alg:MCD}, which we call the maximum column distance (MCD) algorithm, we propose a new deterministic method to generate the elements of the encoding matrix so that it meets the aforementioned linear independence requirement.
    \item We prove that the broadcast rate of the proposed UMCD coding scheme is no larger than the MDS coding scheme.
    \item We characterize a class of index coding instances for which the gap between the broadcast rates of the recursive coding scheme and the proposed UMCD coding scheme grows linearly with the number of messages.
    \item We characterize a class of index coding instances for which the gap between the broadcast rates of the ICC coding scheme and the proposed UMCD coding scheme grows linearly with the number of messages.
    \item We extend the UMCD coding scheme to its partial and fractional versions by applying time sharing over the subinstances, where each subinstance is solved using the UMCD coding scheme. This brings about the partial UMCD and fractional partial UMCD coding schemes, which strictly improve upon the PCC and FPCC coding schemes, respectively. The fractional partial UMCD coding scheme is optimal for all index coding instances with up to and including five messages.
\end{enumerate}

\subsection{Organization of the Paper}
The rest of this paper is organized as follows. Section \ref{sec:02} provides a brief overview of the system model, relevant background and definitions. In Section \ref{sec:03}, three index coding instances are provided to describe the motivation of this paper. In Section \ref{sec:04}, the UMCD coding scheme is proposed. Section \ref{sec:MCM-MCD} first establishes the relation between finding the satisfied receivers and the MCM problem, and then the MCD algorithm is proposed. In Section \ref{sec:07}, first we prove that the UMCD performs at least as well as the MDS code in terms of the broadcast rate. Then, we provide two classes of index coding instances for which the gap between the broadcast rates of the recursive and ICC coding schemes and the UMCD coding schemes grows linearly with the number of messages. In Section \ref{sec:09}, the UMCD coding scheme is extended to its vector version. Finally, Section \ref{sec:10} concludes the paper.

\section{System Model and Background} \label{sec:02}

\subsection{Notation}
Small letters such as $n$  denote an integer number where  $[n]\triangleq\{1,...,n\}$ and $[n:m]\triangleq\{n,n+1,\dots m\}$ for $n<m$. Capital letters such as $L$ denote a set whose cardinality is denoted by $|L|$ and power set is denoted by $\mathcal{P}(A)$. Symbols in bold face such as $\boldsymbol{l}$ and $\boldsymbol{L}$ denote a vector and a matrix, respectively, with $\boldsymbol{L}^T$ denoting the transpose of matrix $\boldsymbol{L}$. A calligraphic symbol such as $\mathcal{L}$ is used to denote a set whose elements are sets.\\
We use $\mathbb{F}_q$ to denote a finite field of size $q$ and write $\mathbb{F}_{q}^{n\times m}$ to denote the vector space of all $n\times m$ matrices over the field $\mathbb{F}_{q}$.
Given a matrix $\boldsymbol{L}\in \mathbb{F}_{q}^{n\times m}$ with elements $l_{i,j}\in \mathbb{F}_q$, we use $\boldsymbol{L}_{S}, S\subseteq[n]$ to represent the $|S|\times m$ submatrix of $\boldsymbol{L}$ comprised of the rows of $\boldsymbol{L}$ indexed by $S$.
\subsection{System Model}
Consider a broadcast communication system in which a server transmits a set of $mt$ messages $X=\{x_{i}^{j},\ i\in[m],\ j\in [t]\},\ x_{i}^{j}\in \mathbb{F}_q$, to a number of receivers $U=\{u_i,\ i\in[m]\}$ through a noiseless broadcast channel. Each receiver $u_i$ wishes to receive a message of length $t$, $X_i=\{x_{i}^{j},\  j\in[t]\}$ and may have a priori knowledge of a subset of the messages $S_i:=\{x_{l}^{j},\ l\in A_{i},\ j\in[t]\},\ A_{i}\subseteq[m]\backslash \{i\}$, which is referred to as its side information set. The main objective is to minimize the number of coded messages which is required to be broadcast so as to enable each receiver to decode its requested message.
An instance of the index coding problem $\mathcal{I}$ is characterized by the side information set of all receivers and can be represented as  $\mathcal{I}=\{(i|A_i), i\in[m]\}$.

\subsection{General Index Code}
\begin{defn}[Index Code]
Given an instance of the index coding problem $\mathcal{I}=\{(i|A_i), i\in[m]\}$, a $(t,r)$ index code is defined as $\mathcal{C}_{\mathcal{I}}=(\phi_{\mathcal{I}},\{\psi_{\mathcal{I}}^{i}\})$, where
 \begin{itemize}[leftmargin=*]
     \item $\phi_{\mathcal{I}}: \mathbb{F}_{q}^{mt}\rightarrow \mathbb{F}_{q}^{r}$ is the encoding function which maps the $mt$ message symbols $x_{i}^{j}\in \mathbb{F}_{q}$ to the $r$ coded messages as $Y=\{y_1,\dots,y_r\}$, where $y_k\in \mathbb{F}_{q}, k\in [r]$.
     \item $\psi_{\mathcal{I}}^{i}:$ represents the decoder function, where for each receiver $u_i, i\in[m]$, the decoder $\psi_{\mathcal{I}}^{i}: \mathbb{F}_{q}^{r}\times \mathbb{F}_{q}^{|A_i|t}\rightarrow \mathbb{F}_{q}^{t}$ maps the received $r$ coded messages $y_k\in Y, k\in[r]$ and the $|A_i|t$ messages $x_{l}^{j}\in S_i$ in the side information to the $t$ decoded symbols $\psi_{\mathcal{I}}^{i}(Y,S_i)=\{\hat{x}_{i}^{j}, j\in [t]\}$, where $\hat{x}_{i}^{j}$ is an estimate of $x_{i}^{j}$.
 \end{itemize}
\end{defn}

\begin{defn}[$\beta(\mathcal{C}_{\mathcal{I}})$: Broadcast Rate of $\mathcal{C}_{\mathcal{I}}$]
Given an instance of the index coding problem $\mathcal{I}$, the broadcast rate of a $(t,r)$ index code $\mathcal{C}_{\mathcal{I}}$ is defined as $\beta(\mathcal{C}_{\mathcal{I}})=\frac{r}{t}$.
\end{defn}

\begin{defn}[$\beta(\mathcal{I})$: Broadcast Rate of $\mathcal{I}$]
Given an instance of the index coding problem $\mathcal{I}$, the broadcast rate $\beta(\mathcal{I})$ is defined as
\begin{equation} \label{eq:opt-rate}
    \beta(\mathcal{I})=\inf_{t} \inf_{\mathcal{C}_{\mathcal{I}}} \beta(\mathcal{C}_{\mathcal{I}}).
\end{equation}
Thus, the broadcast rate of any index code $\mathcal{C}_{\mathcal{I}}$ provides an upper bound on the broadcast rate of $\mathcal{I}$, i.e., $\beta(\mathcal{I}) \leq \beta(\mathcal{C}_{\mathcal{I}})$.
\end{defn}

\begin{defn}[Index Coding Subinstance]
Given an index coding instance $\mathcal{I}=\{(i|A_i), i\in[m]\}$, for any subset $M\subseteq[m]$, we define a subinstance as $\mathcal{I}_M=\{(i|A_i\cap M), i\in M\}$. Since each subinstance $\mathcal{I}_M$ is completely characterized by $M$, for brevity, we call it subinstance $M$ in the rest of this paper.
\end{defn}

\begin{defn}[Partitioning an index coding instance into subinstances]
We say that the index coding instance $\mathcal{I}$ is partitioned into subinstances $M_1, \dots, M_n$ for some $n\leq m$, if the following conditions are met.
\begin{equation} \label{eq:partition-con}
   \left\{\begin{array}{lc}
        M_1 \cup ... \cup M_n=[m],   \\ \\
        M_{j_1} \cap M_{j_2}=\emptyset, \ \ \text{if} \ j_1 \neq j_2 \ \text{for} \ j_1,j_2 \in [n]. 
    \end{array}
    \right.
\end{equation}
\end{defn}
Note, unless we explicitly use the word `partition', the subsets $M_1, \dots, M_n$ may overlap.

\subsection{Linear Index Code}
Let $\boldsymbol{x}=[\boldsymbol{x}_{1}^{T},\dots,\boldsymbol{x}_{m}^{T}]^{{T}}\in \mathbb{F}_{q}^{mt\times 1}$ denote the message vector, where $\boldsymbol{x}_{i}=[x_{i}^{1},\dots,x_{i}^{t}]^{T}\in \mathbb{F}_{q}^{t\times 1}$ is the requested message vector by receiver $u_i, i\in[m]$.

\begin{defn}[Linear Index Code]
Given an instance of the index coding problem $\mathcal{I}=\{(i|A_i), i\in[m]\}$, a $(t,r)$ linear index code is defined as $\mathcal{C}_{\mathcal{I}}=(\boldsymbol{H},\{\psi_{\mathcal{I}}^{i}\})$, where
  \begin{itemize}[leftmargin=*]
      \item $\boldsymbol{H}: \mathbb{F}_{q}^{mt\times 1}\rightarrow \mathbb{F}_{q}^{r\times 1}$ is the $r\times mt$ encoding matrix which maps the message vector $\boldsymbol{x}\in \mathbb{F}_{q}^{mt\times 1}$  to a coded message vector $\boldsymbol{y}=[y_1,\dots,y_r]^{T}\in \mathbb{F}_{q}^{{r}\times 1}$ as follows
      \begin{equation} \nonumber
      \boldsymbol{y}=\boldsymbol{H}\boldsymbol{x}=\sum_{i\in [m]} \boldsymbol{H}^{\{i\}}\boldsymbol{x}_i.
      \end{equation}
      Here $\boldsymbol{H}^{\{i\}}\in \mathbb{F}_{q}^{r\times t}$ is the local encoding matrix of the $i$-th message vector $\boldsymbol{x}_i$ such that
      $\boldsymbol{H}=
      \left [\begin{array}{c|c|c}
        \boldsymbol{H}^{\{1\}} & \dots & \boldsymbol{H}^{\{m\}}
      \end{array}
     \right ]\in \mathbb{F}_{q}^{r\times mt}$.
     \item $\psi_{\mathcal{I}}^{i}$ represents the linear decoder function for receiver $u_{i}, i\in[m]$, where $\psi_{\mathcal{I}}^{i}(\boldsymbol{y}, S_i)$ maps the received coded message $\boldsymbol{y}$ and its side information messages $S_i$ to $\hat{\boldsymbol{x}}_{i}$, which is an estimate of the requested message vector $\boldsymbol{x}_i$.
  \end{itemize}
\end{defn}

\begin{prop} \label{prop-lineardecoding-condition}
The necessary and sufficient condition for linear decoder $\psi_{\mathcal{I}}^{i}, \forall i\in[m]$ to correctly decode the requested message vector $\boldsymbol{x}_i$ is
   \begin{equation} \label{eq:dec-cond}
       \mathrm{rank} \ \boldsymbol{H}^{\{i\}\cup B_i}= \mathrm{rank} \ \boldsymbol{H}^{B_i} + t,
   \end{equation}
where $B_i=[m]\backslash (A_i\cup \{i\})$ represents the interfering message set of receiver $u_i, i\in [m]$, and $\boldsymbol{H}^L$ denotes the matrix $\left [\begin{array}{c|c|c}
       \boldsymbol{H}^{\{l_1\}} & \dots & \boldsymbol{H}^{\{l_{|L|}\}}
     \end{array}
    \right ]$ for the given set $L=\{l_1,\dots,l_{|L|}\}$. 
\end{prop}

\begin{proof}
Refer to Appendix \ref{proof:prop-lineardecoding-condition}.  
\end{proof}

\begin{defn}[Scalar and Vector Linear Index Code \cite{Arbabjolfaei2018}]
The linear index code $\mathcal{C}_{\mathcal{I}}$ is said to be scalar if $t=1$. Otherwise, it is called a vector (or fractional) code. For scalar codes, we use $x_i=x_{i}^{1}, i\in [m]$, for simplicity.
\end{defn}

Since the proposed UMCD coding scheme is a scalar scheme, throughout the paper until Section \ref{sec:09}, where we extend the UMCD to its vector version, we assume that $t=1$. Thus, for any matrix such as $\boldsymbol{H}\in \mathbb{F}_{q}^{r\times m}, r\leq m$, we use $\boldsymbol{H}^{L}, L\subseteq [m]$ to denote the $r\times |L|$ submatrix of $\boldsymbol{H}$ comprised of the columns indexed by $L$. 
\\
Note, by setting $t=1$, we have $\beta(\mathcal{C}_{\mathcal{I}})=r$. Since sending the messages uncoded with rate $r=m$ is a linear index code for any index coding instance, we always have $r\leq m$.

\begin{defn}[$\beta_{\text{MDS}}(\mathcal{I})$ MDS Broadcast Rate for $\mathcal{I}$]
Given an index coding instance $\mathcal{I}=\{(i|A_i), i\in[m]\}$, let $|A|_{\text{min}}=\min_{i\in[m]} A_i$ denote the minimum size of side information. Then, the broadcast rate of the MDS coding scheme is $\beta_{\text{MDS}}(\mathcal{I})=m-|A|_{\text{min}}$.
\end{defn}

A brief overview of the PCC, FPCC, recursive and ICC coding schemes is provided in Appendix \ref{FirstAppendix}. Throughout the paper, $\beta_{\text{PCC}}(\mathcal{I})$, $\beta_{\text{FPCC}}(\mathcal{I})$, $\beta_{\text{R}}(\mathcal{I})$ and $\beta_{\text{ICC}}(\mathcal{I})$, respectively, denote the broadcast rate of the PCC, FPCC, recursive and ICC schemes. The broadcast rate of the proposed UMCD coding scheme in this paper is denoted by $\beta_{\text{UMCD}}(\mathcal{I})$.

\subsection{Graph Definitions}
\begin{defn}[$\mathcal{G}_{\mathcal{I}}$: Graph Representation of $\mathcal{I}$]
The index coding instance $\mathcal{I}=\{(i|A_i), i\in[m]\}$ can be represented as a directed graph $\mathcal{G}_{\mathcal{I}}=(V,E)$, where $V=[m]$ and $E\subseteq[m]\times [m]$, respectively, denote the vertex and edge sets such that $(i,j)\in E$ if and only if (iff) $i\in A_j$, for all $(i,j)\in [m]\times [m]$. Moreover, any subinstance $M\subseteq[m]$ can be represented by the induced subgraph
$\mathcal{G}_{M}=(M, E_M)$, where $E_M\subseteq M\times M$ such that $(i,j)\in E_M$ iff $i\in A_j\cap M$. Since, subgraph $G_{M}$ is completely characterized by subset $M$, for brevity, we call it subgraph $M$ in the rest of this paper.
\end{defn}

Throughout this paper, if there is a pairwise clique between two vertices $i$ and $j$ (i.e., $i\in A_j$ and $j\in A_i$), then they will be connected by a solid bidirectional arrow. Otherwise, we use a dashed unidirectional arrow.

\begin{defn}[Maximum Acyclic Induced Subgraph (MAIS) of $\mathcal{I}$]
Given an index coding instance $\mathcal{I}$, let $\mathcal{M}$ be the set of all acyclic vertex-induced subgraphs of $\mathcal{G}_{\mathcal{I}}$. A subgraph $M\in \mathcal{M}$ with maximum size $|M|$ is said to be the MAIS set of $\mathcal{I}$, and $\beta_{\text{MAIS}}(\mathcal{I})=|M|$ is called the MAIS bound on $\mathcal{I}$.
\end{defn}

\begin{prop}[\mdseries Bar-Yossef \textit{et all}. \cite{blasiak2011lexicographic}]
Given the index coding instance $\mathcal{I}$, we have $\beta(\mathcal{I})\geq \beta_{\text{MAIS}}(\mathcal{I})$. Thus, the MAIS bound imposes a lower bound on the broadcast rate.
\end{prop}

\section{Motivating Index coding Instances} \label{sec:03}

In this section, three index coding instances are provided to illustrate the motivation behind the proposed UMCD coding scheme for the index coding problem. Note that since the MCM and the MCD algorithms have not yet been discussed, for these instances, we employ neither the MCM algorithm for finding the satisfied receivers, nor the MCD algorithm for determining the elements of the encoding matrix. However, the elements of the encoding matrix are designed to meet the UMCD coding scheme requirement (where each column is designed to be linearly independent of the space spanned by other columns to the extent possible).

\begin{exmp}[The UMCD versus the MDS Coding Scheme]
Consider the index coding instance $\mathcal{I}_1=\{(1|-), (2|3,4), (3|2,4), (4,|2,3)\}$, shown in Figure \ref{fig:motiv:1}. For this instance, the broadcast rate of the MDS code is $\beta_{\text{MDS}}(\mathcal{I}_1)=m-|A|_{\text{min}}=4$, which cannot perform better than the uncoded transmission. The main demerit of the MDS code is that its broadcast rate is determined only by the minimum size of side information. This means the MDS code cannot benefit from the side information set of other receivers. 
Now, consider the following transmission technique. First, the server aims to satisfy receiver $u_1$ which has the minimum size of side information. So, it transmits $x_1$, and then removes this receiver from the instance. Now, the server observes three receivers $u_1, u_2$ and $u_3$, which form a clique with each other. Thus, by sending $x_2+x_3+x_4$ all the receivers will be satisfied. This shows how updating the problem by removing the satisfied receivers after each transmission can provide a more efficient broadcast rate compared to the MDS code. 
\end{exmp}

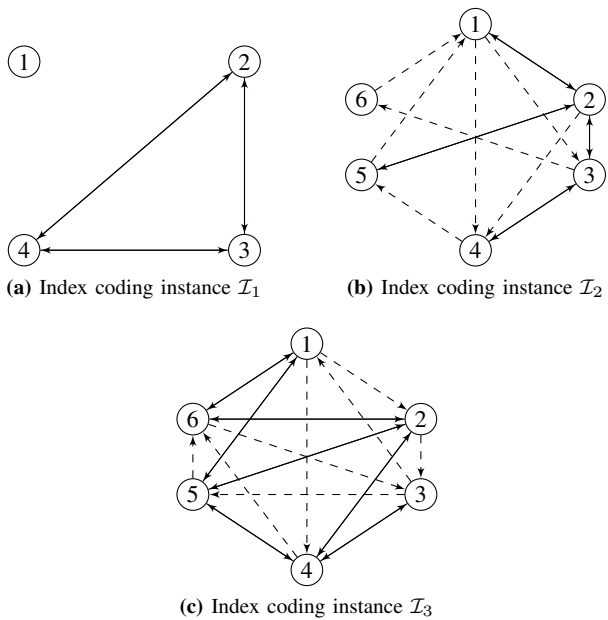
\begin{figure}[t]
\centering
\subfloat[Index coding instance $\mathcal{I}_1$]
{
       \begin{tikzpicture}
                 \tikzset{vertex/.style = {shape=circle,draw,minimum size=1em}}
                 \tikzset{edge/.style = {->,> = latex'}}
                 \node[vertex,inner sep=1.5pt, minimum size=0.5pt] (a1) at  (0,2.5) {\small 1};
                 \node[vertex,inner sep=1.5pt, minimum size=0.5pt] (a2) at  (2.9,2.5) {\small 2};
                 \node[vertex,inner sep=1.5pt, minimum size=0.5pt] (a3) at  (2.9,0) {\small 3};
                 \node[vertex,inner sep=1.5pt, minimum size=0.5pt] (a4) at  (0,0) {\small 4};
                 \draw[edge] (a2) to             (a3);
                 \draw[edge] (a2) to             (a4);
                 \draw[edge] (a3) to             (a2);
                 \draw[edge] (a3) to             (a4);
                 \draw[edge] (a4) to             (a2);
                 \draw[edge] (a4) to             (a3);
            \end{tikzpicture}  \label{fig:motiv:1}
}
\ \ \ \ \ \ \
\subfloat[Index coding instance $\mathcal{I}_2$]
{
           \begin{tikzpicture}
                 \tikzset{vertex/.style = {shape=circle,draw,minimum size=1em}}
                 \tikzset{edge/.style = {->,> = latex'}}
                 \node[vertex,inner sep=1.5pt, minimum size=0.5pt] (a1) at  (0,0)         {\small 1};
                 \node[vertex,inner sep=1.5pt, minimum size=0.5pt] (a2) at  (1.5,-1)       {\small 2};
                 \node[vertex,inner sep=1.5pt, minimum size=0.5pt] (a3) at  (1.5,-2)     {\small 3};
                 \node[vertex,inner sep=1.5pt, minimum size=0.5pt] (a4) at  (0,-3)     {\small 4};
                 \node[vertex,inner sep=1.5pt, minimum size=0.5pt] (a5) at  (-1.5,-2)      {\small 5};
                 \node[vertex,inner sep=1.5pt, minimum size=0.5pt] (a6) at  (-1.5,-1)        {\small 6};
                 
                 \draw[edge] (a1) to                (a2);
                 \draw[edge] (a2) to                (a1);
                 
                 \draw[edge] (a2) to                (a3);
                 \draw[edge] (a3) to                (a2);
                 
                 \draw[edge] (a2) to                (a5);
                 \draw[edge] (a5) to                (a2);
                 
                 \draw[edge] (a3) to                (a4);
                 \draw[edge] (a4) to                (a3);

                 \draw[edge] (a5) [dashed] to      (a1);
                 \draw[edge] (a6) [dashed] to      (a1);

                 \draw[edge] (a1) [dashed] to      (a3);
               
                 \draw[edge] (a1) [dashed] to      (a4);
                 \draw[edge] (a2) [dashed] to      (a4);
                 
                 \draw[edge] (a4) [dashed] to      (a5);
                 
                 \draw[edge] (a3) [dashed] to      (a6);
                 
            \end{tikzpicture} \label{fig:motiv:2}
}
\ \ \ \ \ \ \
\subfloat[Index coding instance $\mathcal{I}_{3}$]
{
           \begin{tikzpicture}
                 \tikzset{vertex/.style = {shape=circle,draw,minimum size=1em}}
                 \tikzset{edge/.style = {->,> = latex'}}
                 \node[vertex,inner sep=1.5pt, minimum size=0.5pt] (a1) at  (0,0)         {\small 1};
                 \node[vertex,inner sep=1.5pt, minimum size=0.5pt] (a2) at  (1.5,-1)       {\small 2};
                 \node[vertex,inner sep=1.5pt, minimum size=0.5pt] (a3) at  (1.5,-2)     {\small 3};
                 \node[vertex,inner sep=1.5pt, minimum size=0.5pt] (a4) at  (0,-3)     {\small 4};
                 \node[vertex,inner sep=1.5pt, minimum size=0.5pt] (a5) at  (-1.5,-2)      {\small 5};
                 \node[vertex,inner sep=1.5pt, minimum size=0.5pt] (a6) at  (-1.5,-1)        {\small 6};
                 
                 \draw[edge] (a6) to                (a1);
                 \draw[edge] (a5) to                (a1);
                 
                 \draw[edge] (a4) to                (a2);
                 \draw[edge] (a5) to                (a2);
                 \draw[edge] (a6) to                (a2);
                 
                 \draw[edge] (a4) to                (a3);
                 
                 \draw[edge] (a2) to                (a4);
                 \draw[edge] (a3) to                (a4);
                 \draw[edge] (a5) to                (a4);
                 
                 \draw[edge] (a1) to                (a5);
                 \draw[edge] (a2) to                (a5);
                 \draw[edge] (a4) to                (a5);
                 
                 \draw[edge] (a1) to                (a6);
                 \draw[edge] (a2) to                (a6);

                 \draw[edge] (a3) [dashed] to      (a1);
                 
                 \draw[edge] (a1) [dashed] to      (a2);
                 
                 \draw[edge] (a2) [dashed] to      (a3);
                 \draw[edge] (a6) [dashed] to      (a3);
                 
                 \draw[edge] (a1) [dashed] to      (a4);
                 
                 \draw[edge] (a3) [dashed] to      (a5);
                 
                 \draw[edge] (a4) [dashed] to      (a6);
                 \draw[edge] (a5) [dashed] to      (a6);
                 
            \end{tikzpicture} \label{fig:motiv:3}
}
\caption{Motivating index coding instances.}
\label{fig:motivation}
\end{figure}

\begin{table}[h!]
\small
\centering
\captionsetup{format=hang}
\caption{\textsc{Comparison Between the Broadcast Rate of Different Coding Schemes for $\mathcal{I}_2$}}
\label{Table1}
\begin{tabular}{  c | c  }
Schemes                                          &    Broadcast rate  \\
\hline 
Maximum distance separable (MDS) \cite{Arbabjolfaei2018}                &       5            \\
Scalar clique cover \cite{Birk2006}                             &       4            \\
Scalar cycle cover \cite{Chaudhry}                              &       4            \\
Partial clique cover (PCC) \cite{Birk1998}                      &       4            \\
Scalar Interlinked-cycle cover (ICC) \cite{Thapa2017}            &       4            \\
Scalar recursive \cite{Arbabjolfaei2014}                                &       4            \\
\hline
Fractional clique cover \cite{Birk1998}                         &       4            \\
Fractional cycle cover \cite{Chaudhry}                          &       3.5          \\
Fractional partial clique cover (FPCC) \cite{Yu2014}          &       3.5          \\
Interlinked-cycle cover (ICC) \cite{Thapa2017}                   &       3.5          \\
Recursive \cite{Arbabjolfaei2014}                                       &       3.5          \\
\hline
MAIS bound                                       &       3            \\
Proposed UMCD                                    &       3            
\end{tabular}
\end{table}

\begin{exmp}[The UMCD versus the MDS-based and Graph-based Coding Schemes]
Consider the index coding instance $\mathcal{I}_2=\{(1|2,5,6), (2|1,3,5), (3|1,2,4),(4|1,2,3), (5|2,4)$, $(6|3)\}$, depicted in Figure \ref{fig:motiv:2}. For the broadcast rate of the MDS code, we have $\beta_{\text{MDS}}(\mathcal{I}_2)=m-|A|_{\text{min}}=5$. It can be verified that for the MDS-based and graph-based schemes, partitioning the instance $\mathcal{I}_2$ can improve the broadcast rate by saving two transmissions. The broadcast rate can even be reduced to 3.5 by the vector version of the schemes, which comes at the expense of higher computational complexity. However, none can still achieve the MAIS bound with $\beta_{\text{MAIS}}(\mathcal{I}_2)=3$. Now, assume that the server first targets receiver $u_6$, which has the minimum size of side information by transmitting a linear combination of messages $x_6$ (requested message) and $x_i, i\in A_6$ as follows
\begin{equation} \nonumber
    y_1=[0\ 0\ h_{1,3}\ 0\ 0\ h_{1,6}] \ \boldsymbol{x},
\end{equation}
where $h_{1,i}\in \mathbb{F}_{q}$ and $\boldsymbol{x}=[x_1\ x_2\ x_3\ x_4\ x_5\ x_6]^T$.
This transmission satisfies only receiver $u_6$. Then, the server targets receiver $u_5$ (which has the minimum size of side information among the unsatisfied receivers) by sending a linear combination of messages $x_5$ (requested message) and $x_i, i\in A_5$ as follows
\begin{equation} \nonumber
    y_2=[0\ h_{2,2}\ 0\ h_{2,4}\ h_{2,5}\ 0]\ \boldsymbol{x}.
\end{equation}
It can be checked that only receiver $u_5$ is satisfied by receiving the coded messages $y_1$ and $y_2$. Now, the remaining unsatisfied receivers $u_i, i=[4]$ have the same size of side information. Assume receiver $u_4$ is chosen randomly. To satisfy receiver $u_4$, a linear combination of messages $x_4$ (requested message) and $x_i, i\in A_4$ is transmitted as follows
\begin{equation} \nonumber
    y_3=[h_{3,1}\ h_{3,2}\ h_{3,3}\ h_{3,4}\ 0\ 0]\ \boldsymbol{x}.
\end{equation}
By setting the field size $q=3$ and fixing the nonzero coefficients $h_{1,3}=h_{1,6}=$$h_{2,2}=h_{2,4}=h_{2,5}=$ $h_{3,1}=h_{3,2}=h_{3,3}=1$ and $h_{3,4}=2$, it can be seen that each subset of columns are linearly independent as much as possible, at each step of transmission. One can verify that all the remaining receivers $u_1$, $u_2$, $u_3$ and $u_4$ are able to decode their requested message upon receiving the coded messages $y_1$, $y_2$ and $y_3$. This index code achieves the MAIS bound and, so it is optimal for $\mathcal{I}_2$. This simple instance illustrates how the proposed UMCD coding scheme works, which can outperform the existing linear coding schemes, including the recursive and ICC coding schemes, as presented in Table \ref{Table1}.
\end{exmp}

\begin{exmp}[The UMCD versus the Scalar Binary Minrank Coding Scheme] \label{exm:motive:3}
Please refer to Appendix \ref{app:minrank~UMCD} to see how the proposed UMCD coding scheme outperforms the scalar binary minrank coding scheme for the index coding instance $\mathcal{I}_{3}$, depicted in Figure \ref{fig:motiv:3}.
\end{exmp}

\section{Main Results} \label{sec:04}

\subsection{A Brief Discussion of the MCM Problem and the MCD Algorithm}
Consider a binary matrix $\boldsymbol{G}\in \mathbb{F}_{2}^{r\times m}, r\leq m$, which can be characterized by either set $G=\{(k,i)\in [r]\times [m], g_{k,i}=1\}$ or its support sets as $G_{k}=\{i\in [m], g_{k,i}=1\}, k\in [r]$. Now, we say that matrix $\boldsymbol{H}\in \mathbb{F}_{q}^{r\times m}$ fits the binary matrix $\boldsymbol{G}$ if for any $q\geq 2$, we have $h_{k,i}=0, \forall (k,i)\in ([r]\times [m])\backslash G$.

Now, we consider the following optimization problem
\begin{align} \label{eq:opt}
     \max_{\substack{h_{k,i}\in \mathbb{F}_q\\  (k,i)\in G}} \ \ \ \ \ &\mathrm{rank} \ \boldsymbol{H}, \ \ \ \ \ \ \
     \nonumber \\
    \text{subject to}\ \ \ \ \ \
    &\boldsymbol{H} \ \text{fits}\ \boldsymbol{G},
\end{align}
which gives the maximum rank of $\boldsymbol{H}$ over all possible values from $\mathbb{F}_{q}$ for its elements $h_{k,i}, (k,i)\in G$ such that $\boldsymbol{H}$ fits $\boldsymbol{G}$.
\\
In Subsection \ref{sec:05}, we prove that for any field size $q\geq 2$, the solution of \eqref{eq:opt} is equal to the MCM value of the bipartite graph associated with the binary matrix $\boldsymbol{G}$, denoted by $\mathrm{mcm}(\boldsymbol{G})$. 
In other words, subject to the condition that $\boldsymbol{H}$ fits $\boldsymbol{G}$, we have
\begin{equation} \nonumber
\mathrm{mcm}(\boldsymbol{G}) = \max_{\substack{h_{k,i}\in \mathbb{F}_q\\  (k,i)\in G}} \ \mathrm{rank} \ \boldsymbol{H}.
\end{equation}
The MCM problem can be solved by the polynomial-time Hopcroft-Karp algorithm where its complexity for $\boldsymbol{G}\in \mathbb{F}_{2}^{r\times m}$ is $\mathcal{O}(r^2m)$.\footnote{It is worth noting that although the minrank problem is NP-hard, the maxrank problem in \eqref{eq:opt} can be solved in polynomial-time.}
\\
In Subsection \ref{sec:06}, given a binary matrix $\boldsymbol{G}\in \mathbb{F}_{2}^{r\times m}$, we propose the MCD algorithm to design a matrix $\boldsymbol{H}\in \mathbb{F}_{q}^{r\times m}$ over a sufficiently large field such that $\boldsymbol{H}$ fits $\boldsymbol{G}$ and $\boldsymbol{H}_{[k]}^{L}$ will reach its maximum rank in $\eqref{eq:opt}$ for all $k\in [r], L\subseteq [m]$. The complexity of the MCD algorithm is in general exponential.
\\
Assuming that the encoding matrix $\boldsymbol{H}$ is designed by the MCD algorithm, for any $k\in [r]$ and $L\subseteq[m]$, we have
\begin{equation} \label{eq:rank(H)=mcm(G)}
    \mathrm{rank}(\boldsymbol{H}_{[k]}^{L})=\mathrm{mcm}(\boldsymbol{G}_{[k]}^{L}).
\end{equation}
Hence, the decoding condition in \eqref{eq:dec-cond} will be equivalent to
\begin{equation} \label{eq:dec-mcm}
    \mathrm{mcm}(\boldsymbol{G}_{[k]}^{\{i\}\cup B_{i}})=\mathrm{mcm}(\boldsymbol{G}_{[k]}^{B_{i}})+1,
\end{equation}
which means that, the satisfied receivers in each transmission can be determined just by having access to the binary matrix $\boldsymbol{G}$. This leads to achieving the broadcast rate of the UMCD scheme without knowing the exact value of the elements in $\boldsymbol{H}$. This results in reducing the computational complexity, especially when the UMCD is extended as to its vector version and used as a basic code for solving index coding subinstances.

\subsection{The Proposed UMCD Coding Scheme} 

In this section, we describe the proposed UMCD coding scheme, which is a scalar linear code and is based on updating the problem instance after each transmission. In this scheme, first, in transmission $k$, the UMCD coding scheme characterizes the support set $G_k$ such that $g_{k,i}=1$ if message $x_i$ is included in the linear coded message $y_k$, and $g_{k,i}=0$ otherwise. So, for the first $k$ transmissions, we have a binary matrix $\boldsymbol{G}_{[k]}$. Then, using this binary matrix, the other satisfied receivers are identified and removed from the instance. After all the receivers are satisfied at the transmission $k=r=\beta_{\text{UMCD}}(\mathcal{I})$, the encoding matrix $\boldsymbol{H}$ which fits $\boldsymbol{G}_{[r]}$ will be designed. More specifically, the UMCD coding scheme is described as follows.

\begin{itemize}[leftmargin=*]
    \item In transmission $k$, the UMCD coding scheme aims at satisfying one of the receivers with the minimum size of side information by sending a linear combination of its desired message and the messages in its side information.
    This will give the support set $G_k$ which will, in turn, determine the $k$-th row of the binary matrix $\boldsymbol{G}_{[k]}$.
    \item Assuming that the encoding matrix $\boldsymbol{H}$ will be designed by the MCD algorithm, the satisfied receivers can be determined by checking the condition in \eqref{eq:dec-mcm} using the polynomial-time Hopcroft-Karp algorithm which solves the MCM problem. This will eventually determine the broadcast rate $\beta_{\text{UMCD}}(\mathcal{I})$ as soon as all receivers are satisfied.
    \item Finally, given the binary matrix $\boldsymbol{G}_{[r]}$, the encoding matrix $\boldsymbol{H}$ will be determined by the proposed MCD algorithm from a sufficiently large field.
\end{itemize}

\begin{rem}
In the UMCD coding scheme, selecting the receiver with the minimum size of side information rather than the receivers with the non-minimum size of side information, in each transmission, can intuitively bring about two main advantages. First, it can lead to a lower broadcast rate as it will be discussed in Example \ref{rem:satisfying-minimum}. Second, since the number of ones at each row of the binary matrix $\boldsymbol{G}$ is determined by the side information of such a receiver (with the minimum size of side information), it will lead to a sparser binary matrix $\boldsymbol{G}$. This, in turn, can result in a lower complexity and smaller field size for the MCD algorithm, which will be discussed in Subsection \ref{sec:06}.
\end{rem}

\subsection{Description of the UMCD Algorithm}

In the UMCD algorithm, which is provided in Algorithm \ref{alg:UMCD}, $N$ denotes the set of unsatisfied receivers which is updated after each transmission. In the beginning, all the receivers are considered to be unsatisfied and $N=[m]$. In transmission $k$, let $W=\mathrm{argmin}_{i\in N} |A_i|$ represent the indices of the unsatisfied receivers with the minimum size of side information. Then, one element, denoted by $w$ is chosen randomly from $W$. Now, to satisfy receiver $u_w$, UMCD designs a linear combination of the messages indexed by $G_k=\{w\}\cup A_w$, which characterizes the $k$-th row of the binary matrix $\boldsymbol{G}$ such that $g_{k,i}=1$ if $i\in G_{k}$, and $g_{k,i}=0$ otherwise. Now, receiver $u_w$ is guaranteed to be removed from $N$, since $G_k$ only contains what $u_w$ wants to decode and what it knows. The instant decodability for one of the remaining receivers with the minimum size of side information is a key to reducing the overall broadcast rate in an adaptive manner.
Then, using the Hopcroft-Karp algorithm, the other receivers satisfying the condition in \eqref{eq:dec-mcm} are identified and removed from $N$. This process continues until all the receivers are satisfied, which gives the broadcast rate $\beta_{\text{UMCD}}(\mathcal{I})=r$. Finally, having obtained the binary matrix $\boldsymbol{G}_{[r]}$, the elements of the encoding matrix $\boldsymbol{H}$ are determined using the deterministic MCD algorithm from a field of size 
$q\geq q_{min}$, where 
\begin{equation} \label{eq:q_{min}}
    q_{min}= {m \choose p}, \ p=\min \{\floor{\frac{m}{2}}, r\}.
\end{equation}
It is worth noting that the broadcast rate $\beta_{\text{UMCD}}(\mathcal{I})$ is obtained independently of knowing the exact value of the elements of the encoding matrix $\boldsymbol{H}$. This helps to reduce the computational complexity of the UMCD scheme by not needing to invoke the MCD algorithm to find satisfied receivers at each transmission.

\begin{algorithm}
\SetAlgoLined
\KwInput{$\mathcal{I}=\{(i|A_i), i\in[m]\}$}
  \KwOutput{$\beta_{\text{UMCD}}(\mathcal{I})$ and $\boldsymbol{H}$}
 initialization\;
 $k=0$\;
 ${N}=[m]$\;
 set $\boldsymbol{G}$ as a full-zero matrix of size $m\times m$\;
 \While{${N}\neq\emptyset$}{
  $k \leftarrow k+1$\;
  $W=\arg \min_{i\in N} |A_i|$\; 
  choose an element $w\in {W}$ at random\;
  set $G_{k}=\{w\}\cup A_{w}$\;
  set $g_{k, i}=1, \forall i\in G_k$\;
  $N \leftarrow N\backslash\{w\}$\;
  \For{$i\in N$}{
   \If{$\mathrm{mcm}\ ({\boldsymbol{G}_{[k]}^{\{i\}\cup B_{i}}})=\mathrm{mcm}\ ({\boldsymbol{G}_{[k]}^{B_{i}}})+1$}{
   \texttt{\\}
   $N \leftarrow N\backslash\{i\}$\;
   }
   }
  }
  $\beta_{\text{UMCD}}(\mathcal{I})=r=k$: Broadcast rate\;
  set the field size $q$ such that $q\geq q_{min}$\;
  $\boldsymbol{H}$= $\mathrm{MCD}\ (\boldsymbol{G}_{[r]}, q)$: Encoding matrix\;
 \caption{UMCD Coding Scheme}
 \label{alg:UMCD}
\end{algorithm}

\subsection{The UMCD Scheme Can Outperform the Recursive and ICC Coding Schemes for Some Index Coding Instances with Five Messages}
\begin{exmp} \label{exmp:Recursive-5}
Consider the instance of index coding problem $\mathcal{I}_{4}=\{(1|2,5)$, $(2|1,4)$, $(3|2)$, $(4|5), (5|1,3)\}$, depicted in Figure \ref{fig:I_3}. In the first round $k=1$, UMCD begins with one of the receivers indexed by $W=\{3,4\}$ which have the minimum size of side information. Let $w=3\in W$ be chosen at random. Then, the first transmission is a linear combination of the messages indexed by $G_1=\{3\}\cup A_3=\{2,3\}$, satisfying receiver $u_3$. So, $N=[5]\backslash \{3\}$.
Using the Hopcroft-Karp algorithm, \eqref{eq:dec-mcm} does not hold for the remaining receivers in $N$.
\\
Now, for the second round $k=2$, we have $W=\{4\}$, $w=4$, $G_2=\{4\}\cup A_4=\{4,5\}$, and $N=[5]\backslash \{3,4\}$.
Using the Hopcroft-Karp algorithm, \eqref{eq:dec-mcm} does not hold for the remaining receivers in $N$. 
\\
For the third round $k=3$, we have $W=\{1,2,5\}$. Let $w=1\in W$, $G_3=\{1\}\cup A_1=\{1,2,5\}$, and  $N=[5]\backslash \{1,3,4\}$. Thus, the binary matrix $\boldsymbol{G}_{[3]}$ will be
\vspace{.1mm}
 \begin{equation} \label{exmp:G-recu}
      \boldsymbol{G}_{[3]}=
      \begin{bmatrix}
      0 & 1 & 1 & 0 & 0\\
      0 & 0 & 0 & 1 & 1\\
      1 & 1 & 0 & 0 & 1
      \end{bmatrix}.
  \end{equation}
Using the Hopcroft-Karp algorithm, it can be verified that \eqref{eq:dec-mcm} holds for all the remaining receivers, leading to $\beta_{\text{UMCD}}(\mathcal{I}_{4})=3<\beta_{\text{R}}(\mathcal{I}_{4})=3.5$. Finally, the encoding matrix will be determined by the MCD algorithm over a field of size $q\geq q_{min}={5 \choose 2}=10$, which will be\footnote{For this case, however, we show in Example \ref{exmp:field&complexity-reduction} that the MCD algorithm over any field of size $q\geq 2$ will satisfy \eqref{eq:rank(H)=mcm(G)}.} 
\begin{equation} \nonumber
  \boldsymbol{H}= \mathrm{MCD}(\boldsymbol{G}_{[3]}, q)=
  \begin{bmatrix}
       0 & 1 & 1 & 0 & 0\\
      0 & 0 & 0 & 1 & 1\\
      1 & 1 & 0 & 0 & 1
      \end{bmatrix}.
\end{equation}
Note that for other possible random selections of $w$ at each transmission, it can be checked that the broadcast rate will be the same.
\end{exmp}

 \begin{figure}
       \centering
\subfloat[Index coding instance $\mathcal{I}_{4}$]
{
           \begin{tikzpicture}
                 \tikzset{vertex/.style = {shape=circle,draw,minimum size=1em}}
                 \tikzset{edge/.style = {->,> = latex'}}
                 \node[vertex,inner sep=1.5pt, minimum size=0.5pt] (a1) at  (0,1.1)   {\small 1};
                 \node[vertex,inner sep=1.5pt, minimum size=0.5pt] (a2) at  (1.5,0)     {\small 2};
                 \node[vertex,inner sep=1.5pt, minimum size=0.5pt] (a3) at  (1,-1.3) {\small 3};
                 \node[vertex,inner sep=1.5pt, minimum size=0.5pt] (a4) at  (-1,-1.3){\small 4};
                 \node[vertex,inner sep=1.5pt, minimum size=0.5pt] (a5) at  (-1.5,0)    {\small 5};
                 \draw[edge] (a1) to             (a2);
                 \draw[edge] (a1) to             (a5);
                 \draw[edge] (a2) to             (a1);
                 \draw[edge] (a5) to             (a1);

                 \draw[edge] (a2) [dashed]   to             (a3);
                 \draw[edge] (a3) [dashed]   to             (a5);
                 \draw[edge] (a4) [dashed]   to             (a2);
                 \draw[edge] (a5) [dashed]   to             (a4);
            \end{tikzpicture} \label{fig:I_3}
}
\ \ \ \ \ \ \ \ \ \ \ \
\subfloat[Index coding instance $\mathcal{I}_{5}$]
{
            \begin{tikzpicture}
                 \tikzset{vertex/.style = {shape=circle,draw,minimum size=1em}}
                 \tikzset{edge/.style = {->,> = latex'}}
                 \node[vertex,inner sep=1.5pt, minimum size=0.5pt] (a1) at  (0,1.1)     {\small 1};
                 \node[vertex,inner sep=1.5pt, minimum size=0.5pt] (a2) at  (1.5,0)     {\small 2};
                 \node[vertex,inner sep=1.5pt, minimum size=0.5pt] (a3) at  (1,-1.3)    {\small 3};
                 \node[vertex,inner sep=1.5pt, minimum size=0.5pt] (a4) at  (-1,-1.3)   {\small 4};
                 \node[vertex,inner sep=1.5pt, minimum size=0.5pt] (a5) at  (-1.5,0)    {\small 5};
                 \draw[edge] (a1) to   (a2);
                 \draw[edge] (a2) to   (a1);
                 \draw[edge] (a1) to   (a5);
                 \draw[edge] (a5) to   (a1);
                 \draw[edge] (a2) to   (a3);
                 \draw[edge] (a3) to   (a2);
                 \draw[edge] (a3) to   (a4);
                 \draw[edge] (a4) to   (a3);
                 \draw[edge] (a4) to   (a5);
                 \draw[edge] (a5) to   (a4);

                 \draw[edge] (a3) [dashed]   to    (a1);
                 \draw[edge] (a2) [dashed]   to    (a4);
                 \draw[edge] (a5) [dashed]   to    (a2);
                 \draw[edge] (a5) [dashed]   to    (a3);
                 
            \end{tikzpicture} \label{fig:I_4}
}
       \caption{(a) The broadcast rate of the UMCD coding scheme is 3, while it is 3.5 for the recursive coding scheme. (b) The broadcast rate of the UMCD coding scheme is 2, while it is 2.5 for the ICC coding scheme. For both instances, the UMCD coding scheme is optimal.}
       \label{fig:ins_3,4}
   \end{figure}
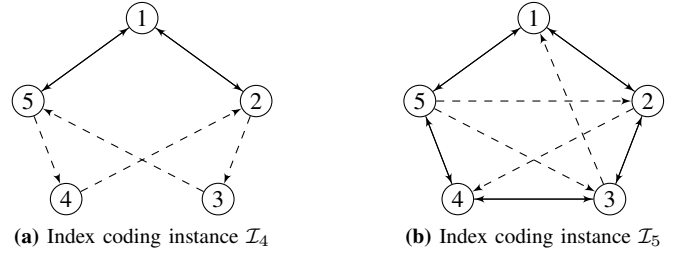

\begin{exmp} \label{exmp:ICC-5}
Consider the index coding instance $\mathcal{I}_{5}=\{(1|2,3,5)$, $(2|1,3,5)$, $(3|2,4,5)$, $(4|2,3,5)$, $(5|1,4)\}$, depicted in Figure \ref{fig:I_4}. In the first round, UMCD begins with the receiver indexed by $W=\{5\}$ which has the minimum size of side information. So, the UMCD coding scheme sets $w=5$, $G_1=\{5\}\cup A_5=\{1,4,5\}$, satisfying receiver $u_5$. So, $N=[5]\backslash \{5\}$.
Using the Hopcroft-Karp algorithm, \eqref{eq:dec-mcm} does not hold for the remaining receivers in $N$.
\\
Now, for the second round $k=2$, we have $W=\{1,2,3,4\}$, $w=2$, $G_{2}=\{2\}\cup A_2=\{1,2,3,5\}$, and $N=[5]\backslash \{2,5\}$. Thus, the binary matrix $\boldsymbol{G}_{[2]}$ will be
\begin{equation} \label{exmp:G-ICC}
      \boldsymbol{G}_{[2]}=
      \begin{bmatrix}
      1 & 0 & 0 & 1 & 1\\
      1 & 1 & 1 & 0 & 1
      \end{bmatrix}.
\end{equation}
Using the Hopcroft-Karp algorithm, \eqref{eq:dec-mcm} holds for all the remaining receivers, which leads to $\beta_{\text{UMCD}}(\mathcal{I}_{5})=2<\beta_{\text{ICC}}(\mathcal{I}_{5})=2.5$. Finally, the encoding matrix will be determined by the MCD algorithm over a field of size $q\geq q_{min}={5 \choose 2}=10$, which will be\footnote{For this case, however, we show in Example \ref{exmp:field&complexity-reduction} that the MCD algorithm over any field of size $q\geq 3$ will satisfy \eqref{eq:rank(H)=mcm(G)}.}
\begin{equation} \nonumber
      \boldsymbol{H}= \mathrm{MCD}(\boldsymbol{G}_{[2]}, q)=
      \begin{bmatrix}
      1 & 0 & 0 & 1 & 1\\
      1 & 1 & 1 & 0 & 2
      \end{bmatrix}.
\end{equation}
Note that for the other possible random selections of $w$ in the second transmission, it can be checked that the broadcast rate will be the same.
\end{exmp}
\vspace{-3ex}

\subsection{Intuition behind the selection of receivers with the minimum size of side information}
Here we discuss the intuition for satisfying receivers with the minimum size of side information. Assume that in the UMCD coding scheme, in each transmission, instead of receiver $u_{w_1}$ with the minimum size of side information, receiver $u_{w_2}$ with a larger size of side information is chosen. Now, since $|A_{w_2}|\geq |A_{w_1}|$, then $|G_k|=|\{w_2\}\cup A_{w_2}|$ will be greater than $|\{w_1\}\cup A_{w_1}|$, which places more ones in the $k$-th row of $\boldsymbol{G}$. 
This, in turn, is expected to increase the value of $\mathrm{mcm}(\boldsymbol{G}_{[k]}^{B_i})$ for receivers with the smaller size of side information (or larger size of interfering message set). In other words, satisfying a receiver with the larger size of side information can increase the dimension of the space spanned by the messages inside the interference message set $B_i$ for receivers with the smaller size of side information.
Hence, it is expected that the decoding condition \eqref{eq:dec-mcm} is not met for receivers with the smaller size of side information. Thus, separate transmissions can still be required to satisfy receivers with the smaller size of side information, which increases the broadcast rate. In the following example, it is shown that selecting receivers with the non-minimum size of side information will lead to suboptimal codes for $\mathcal{I}_{4}$ and $\mathcal{I}_{5}$.

\begin{exmp} \label{rem:satisfying-minimum}
Assume that in the UMCD coding scheme for $\mathcal{I}_{5}$, first we satisfy receiver $u_1$. So, $G_1=\{1\}\cup A_1$. Then, using the Hopcroft-Karp algorithm, \eqref{eq:dec-mcm} holds for receivers $u_1$ and $u_2$. Then, assume we target receiver $u_3$ by setting $G_2=\{3\}\cup A_3$. Then, using the Hopcroft-Karp algorithm, \eqref{eq:dec-mcm} holds for receivers $u_3$ and $u_4$. Thus, we need one more transmission for satisfying receiver $u_5$, which will result in a suboptimal code. Similarly, for the index coding instance $\mathcal{I}_{4}$, if one of the receivers $u_3$ and $u_4$ is not selected for the first three transmissions, then the broadcast rate will be 4, leading to a suboptimal code. 
\end{exmp}

\section{The Maximum Cardinality Matching (MCM) Problem and the Maximum Column Distance (MCD) Algorithm} \label{sec:MCM-MCD}
In this section, first we prove that the solution of the optimization problem in \eqref{eq:opt} will be equal to the MCM value of the bipartite graph associated with binary matrix $\boldsymbol{G}$. This shows that the satisfied receivers in the UMCD coding scheme can be identified by finding $\mathrm{mcm}(\boldsymbol{G})$. Second, we propose the MCD algorithm to design the encoding matrix $\boldsymbol{H}$ such that it fits $\boldsymbol{G}$ and its submatrices will reach their maximum possible rank as desired in \eqref{eq:rank(H)=mcm(G)}.
\vspace{-3ex}
\subsection{The Equivalence of Identifying Satisfied Receivers and Maximum Cardinality Matching} \label{sec:05}

In this subsection, we formalize the relation between finding the satisfied receivers and the MCM problem which will be based on Theorem \ref{thm:01}. 
To prove Theorem \ref{thm:01}, first, we provide Lemma \ref{lem:01} below.
\vspace{-0.5ex}
\begin{lem} \label{lem:01}
For a given matrix $\boldsymbol{G}\in \mathbb{F}_{2}^{r\times r}$, if $g_{k_1,r}=1$ for some $k_1\in [r]$, and there exists a permutation of columns in $\boldsymbol{F}=\boldsymbol{G}_{[r]\backslash \{k_1\}}^{[r-1]}$ which results in $f_{k,k}=1, \forall k\in[r-1]$, then there will also exist a permutation of columns in $\boldsymbol{G}$ which will lead to $g_{k,k}=1, \forall k\in [r]$.
\end{lem}

\begin{proof}
It can be seen the relation between the elements of $\boldsymbol{F}$ and $\boldsymbol{G}$ is as follows
\begin{equation} \nonumber
    \left\{\begin{array}{lccc}
         f_{k,j}=g_{k,j}, \ \ \ \ \ \ \ \ \ \ \ \ \forall\ \ k<k_1, \ \ \forall j\in[r-1], 
         \\
         f_{k,j}=g_{k,j+1}, \ \ \ \ \ \ \ \ \ \ \forall \ \ k\geq k_1, \ \ \forall j\in[r-1], 
    \end{array}\right.
\end{equation}
So, if there exists a permutation of columns in $\boldsymbol{F}$ to make $f_{k,k}=1, \forall k\in[r-1]$, the same permutation of columns in $\boldsymbol{G}$ will result in $g_{k,k}=1, \forall k<k_1$ and $g_{k+1,k}=1, \forall k>k_1$. Since $g_{k_1,r}=1$, then we make a left circular shift of the columns $\boldsymbol{G}^{\{k_1\}}, \dots, \boldsymbol{G}^{\{r-1\}}, \boldsymbol{G}^{\{r\}}$ by one position so that they will be placed in a new position as columns $\boldsymbol{G}^{\{k_1+1\}}, \dots, \boldsymbol{G}^{\{r\}}, \boldsymbol{G}^{\{k_1\}}$, respectively. Then, it can be observed that we will have $g_{k,k}=1, \forall k\in [r]$.
\end{proof}

\vspace{2mm}

\begin{thm} \label{thm:01}
    Assume matrix $\boldsymbol{H}$ fits the binary matrix $\boldsymbol{G}\in \mathbb{F}_{2}^{r\times r}$. Now, $\boldsymbol{H}$ can be designed to be full-rank over any field size $q\geq 2$, iff there exists a permutation of its columns which results in $g_{k,k}=1, \forall k\in[r]$.
\end{thm}


\begin{proof}

For the \textit{if} condition, assume that there is a permutation of columns in $\boldsymbol{G}$ which results in $g_{k,k}=1$ for all $k\in[r]$, then we do such a permutation for $\boldsymbol{H}$, and set $h_{k,k}=1$ for all $k\in[r]$, and its other elements to zero. Since rearranging the columns does not alter the matrix rank, we have $\mathrm{rank}(\boldsymbol{H})=r$ over any field of size $q\geq 2$.
\\
\textit{Conversely}, having Lemma \ref{lem:01}, we use induction for proving the converse. For $r=1$, it is obvious that $\mathrm{rank}(\boldsymbol{H}=[h_{1,1}])=1$, only if $h_{1,1}\neq 0$, which gives $g_{1,1}=1$. For the induction hypothesis, we assume that the necessary condition holds for $\boldsymbol{H}$ of size ${(r-1)\times (r-1)}$. Now, we need to prove that the necessary condition must also hold for $\boldsymbol{H}$ of size ${r\times r}$. Let $\mathrm{rank}(\boldsymbol{H})=r$. The determinant of $\boldsymbol{H}$ can be obtained using the Laplace expansion along the last column $r$, as follows
\begin{equation} \label{eq:04}
     |\boldsymbol{H}|=\sum_{k\in[r]} (-1)^{k+r}\ h_{k,r} \ \left |\boldsymbol{H}_{[r]\backslash\{k\}}^{[r-1]}\right |.
\end{equation}
Thus, for having $|\boldsymbol{H}|\neq 0$, there must exist at least one $k_1\in [r]$ such that $h_{k_1,r} \ |\boldsymbol{H}_{[r]\backslash\{k_1\}}^{[r-1]}|\neq 0$. So, we must have $h_{k_1,r}\neq 0$, which requires $g_{k_1,r}=1$, and $|\boldsymbol{H}_{[r]\backslash\{k_1\}}^{[r-1]}|\neq 0$. This is guaranteed by the induction hypothesis that for $\boldsymbol{F}=\boldsymbol{G}_{[r]\backslash\{k_1\}}^{[r-1]}$ there exists a permutation of its columns which leads to $f_{k,k}=1, \forall k\in[r]$, which completes the proof according to Lemma \ref{lem:01}.
\end{proof}

\begin{rem}
   It can be easily concluded from Theorem \ref{thm:01} that the maximum rank of
   $\boldsymbol{H}\in \mathbb{F}_{q}^{r\times m}$ which fits $\boldsymbol{G}\in \mathbb{F}_{2}^{r\times m}$
   in \eqref{eq:opt} will be equal to the maximum number of ones that appear on the main diagonal of $\boldsymbol{G}$ after performing an appropriate permutation of its columns. Thus, this will be equal to the maximum number of elements with value of 1, which are positioned in distinct rows and columns of $\boldsymbol{G}$.
\end{rem}

\subsubsection{The MCM Problem}

\begin{figure*}
\centering
\subfloat[Binary matrix $\boldsymbol{G}_{6\times 6}$]
{
\begin{blockarray}{ccccccc}
              & \textcolor{brown}{{\Large \textcircled{\small $c_1$}}} & \textcolor{brown}{{\Large \textcircled{\small $c_2$}}} & \textcolor{brown}{{\Large \textcircled{\small $c_3$}}} & \textcolor{brown}{{\Large \textcircled{\small $c_4$}}} & \textcolor{brown}{{\Large \textcircled{\small $c_5$}}} & \textcolor{brown}{{\Large \textcircled{\small $c_6$}}} \\
             \begin{block}{c[cccccc]}
               \textcolor{blue}{{\large \textcircled{\small $p_1$}}} & 0 & 1 & 0 & 1 & 0 & 0 \\
               \textcolor{blue}{{\large \textcircled{\small $p_2$}}} & 1 & 1 & 1 & 1 & 1 & 1 \\
               \textcolor{blue}{{\large \textcircled{\small $p_3$}}} & 0 & 1 & 0 & 0 & 0 & 0\\
               \textcolor{blue}{{\large \textcircled{\small $p_4$}}} & 0 & 1 & 0 & 1 & 0 & 0\\
               \textcolor{blue}{{\large \textcircled{\small $p_5$}}} & 1 & 0 & 0 & 0 & 1 & 0\\
               \textcolor{blue}{{\large \textcircled{\small $p_6$}}} & 0 & 0 & 1 & 1 & 1 & 1\\
             \end{block}
\end{blockarray}
}
 \
\subfloat[Bipartite graph $\mathcal{G}(P,C,E)$]
{
                 \begin{tikzpicture}[baseline=(4.base)]
                 \tikzset{vertex/.style = {shape=circle,draw, minimum size=1em}}
                 \tikzset{edge/.style = {->,> = latex'}}
                 \node[vertex,inner sep=1.5pt, color=blue, minimum size=0.5pt]  (1)   at  (0,3.75) {\small $p_1$};
                 \node[vertex,inner sep=1.5pt, color=blue, minimum size=0.5pt]  (2)   at  (0,3)    {\small $p_2$};
                 \node[vertex,inner sep=1.5pt, color=blue, minimum size=0.5pt]  (3)   at  (0,2.25) {\small $p_3$};
                 \node[vertex,inner sep=1.5pt, color=blue, minimum size=0.5pt]  (4)   at  (0,1.5)  {\small $p_4$};
                 \node[vertex,inner sep=1.5pt, color=blue, minimum size=0.5pt]  (5)   at  (0,.75)  {\small $p_5$};
                 \node[vertex,inner sep=1.5pt, color=blue, minimum size=0.5pt]  (6)   at  (0,0)    {\small $p_6$};
                 \node[vertex,inner sep=1.5pt, color=brown, minimum size=0.5pt] (7)   at  (3,3.75) {\small $c_1$};
                 \node[vertex,inner sep=1.5pt, color=brown, minimum size=0.5pt] (8)   at  (3,3)    {\small $c_2$};
                 \node[vertex,inner sep=1.5pt, color=brown, minimum size=0.5pt] (9)   at  (3,2.25) {\small $c_3$};
                 \node[vertex,inner sep=1.5pt, color=brown, minimum size=0.5pt] (10)  at  (3,1.5)  {\small $c_4$};
                 \node[vertex,inner sep=1.5pt, color=brown, minimum size=0.5pt] (11)  at  (3,.75)  {\small $c_5$};
                 \node[vertex,inner sep=1.5pt, color=brown, minimum size=0.5pt] (12)  at  (3,0)    {\small $c_6$};
                 \draw[edge] (1) to             (8);
                 \textcolor{red}{\draw[edge] (1) to             (10);}
                 
                 \textcolor{red}{\draw[edge] (2) to             (7);}
                 \draw[edge] (2) to             (8);
                 \draw[edge] (2) to             (9);
                 \draw[edge] (2) to             (10);
                 \draw[edge] (2) to             (11);
                 \draw[edge] (2) to             (12);
                 
                 \textcolor{red}{\draw[edge] (3) to             (8);}
                 
                 \draw[edge] (4) to             (8);
                 \draw[edge] (4) to             (10);
                 
                 \draw[edge] (5) to             (7);
                 \textcolor{red}{\draw[edge] (5) to             (11);}
                 
                 \draw[edge] (6) to             (9);
                 \draw[edge] (6) to             (10);
                 \draw[edge] (6) to             (11);
                 \textcolor{red}{\draw[edge] (6) to             (12);}
           \end{tikzpicture} 
           \label{fig:ex:mcm:g}
}
\
\subfloat[Binary matrix $\boldsymbol{G}_{6\times 6}^{\prime}$]
{
\begin{blockarray}{ccccccc}
               & \textcolor{brown}{{\Large \textcircled{\small $c_1$}}} & \textcolor{brown}{{\Large \textcircled{\small $c_2$}}} & \textcolor{brown}{{\Large \textcircled{\small $c_3$}}} & \textcolor{brown}{{\Large \textcircled{\small $c_4$}}} & \textcolor{brown}{{\Large \textcircled{\small $c_5$}}} & \textcolor{brown}{{\Large \textcircled{\small $c_6$}}}\\
             \begin{block}{c[cccccc]}
               \textcolor{blue}{{\large \textcircled{\small $p_1$}}} & 0 & 0 & 0 & \textcolor{red}{{\large \textcircled{\small 1}}} & 0 & 0 \\
               \textcolor{blue}{{\large \textcircled{\small $p_2$}}} & \textcolor{red}{{\large \textcircled{\small 1}}} & 0 & 0 & 0 & 0 & 0 \\
               \textcolor{blue}{{\large \textcircled{\small $p_3$}}} & 0 & \textcolor{red}{{\large \textcircled{\small 1}}} & 0 & 0 & 0 & 0\\
               \textcolor{blue}{{\large \textcircled{\small $p_4$}}} & 0 & 0 & 0 & 0 & 0 & 0\\
               \textcolor{blue}{{\large \textcircled{\small $p_5$}}} & 0 & 0 & 0 & 0 & \textcolor{red}{{\large \textcircled{\small 1}}} & 0\\
               \textcolor{blue}{{\large \textcircled{\small $p_6$}}} & 0 & 0 & 0 & 0 & 0 & \textcolor{red}{{\large \textcircled{\small 1}}}\\
             \end{block}
\end{blockarray}\label{fig:ex:mcm:c}
}
\caption{Consider the binary matrix $\boldsymbol{G}_{6\times 6}$ in (a) which is represented as a bipartite graph $\mathcal{G}(P,C,E)$ in (b), where $P=\{p_k, k\in[6]\}$, $C=\{c_i, i\in[6]\}$, and $E=\{(k,i)\in [6]\times [6]: g_{k,i}=1\}$. It can be verified that $E^{\prime}=\{(p_1,c_4), (p_2,c_1), (p_3,c_2), (p_5,c_5), (p_6,c_6)\}$ (in red) is an MCM of graph $\mathcal{G}$. So, we have $\mathrm{mcm}(\boldsymbol{G})=|E^{\prime}|=5$. The binary matrix $\boldsymbol{G}^{\prime}$ in (c), associated with the subgraph $\mathcal{G}(P,C,E^{\prime})$, illustrates how the elements with value of one are placed in distinct rows and columns, so by a proper permutation of its columns, all the ones can be positioned on the main diagonal.}
\label{fig:mcm}
\end{figure*}
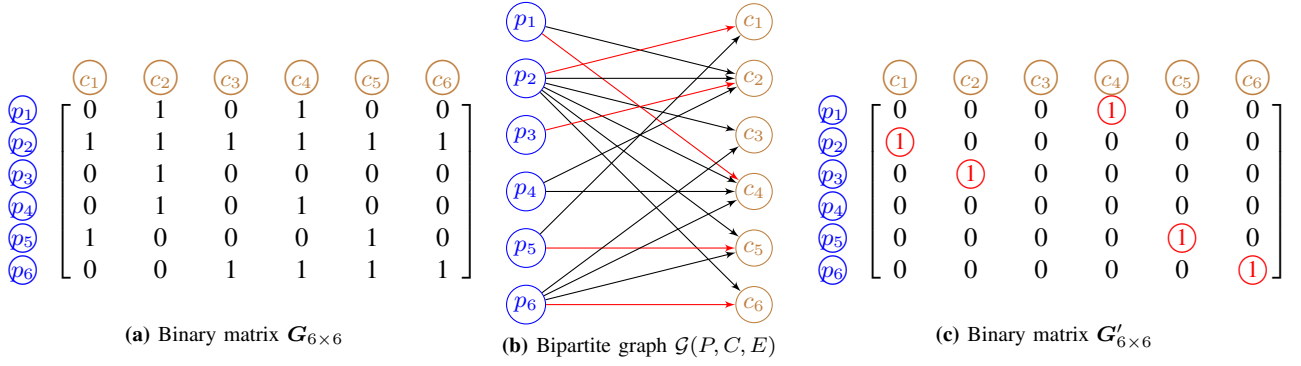

A binary matrix $\boldsymbol{G}\in \mathbb{F}_{2}^{r \times m}$ can be represented as a bipartite graph $\mathcal{G}(P,C,E)$, where $P=\{p_k, k\in[r]\}$ and $C=\{c_{i}, i\in[m]\}$ are two disjoint sets of vertices, called, the row and column vertex set, respectively. $E\subseteq P\times C$ denotes the set of edges such that $(p_k,c_i)\in E, (k,i)\in [r]\times [m]$, connects $p_k$ to $c_i$ only if $g_{k,i}=1$.

\begin{defn}[Maximum Cardinality Matching of $\mathcal{G}$ \cite{Hopcroft1971}]
For the bipartite graph $\mathcal{G}(P,C,E)$, let $E^{\prime}\subseteq E$ denote the set of edges such that no two edges share a common vertex (i.e., if $(p_{k_1},c_{i_1})\in E^{\prime}$ and $(p_{k_2},c_{i_2})\in E^{\prime}$, then $k_1\neq k_2$ and $i_1\neq i_2$). Such a set $E^{\prime}$ with maximum size is called the maximum cardinality matching (MCM) of $\mathcal{G}$.
\end{defn}

\begin{defn}[$\mathrm{mcm}(\boldsymbol{G})$]
Consider a binary matrix $\boldsymbol{G}$ represented as a graph $\mathcal{G}$. Let $E^{\prime}$ be an MCM of $\mathcal{G}$. Then, the maximum cardinality matching of $\boldsymbol{G}$ is defined as $\mathrm{mcm}(\boldsymbol{G})\triangleq|E^{\prime}|$.
\end{defn}

\begin{rem}
As mentioned above, for each two edges $(p_{k_1},c_{i_1})$ and $(p_{k_2},c_{i_2})$ in $E^{\prime}$, we must have $k_1\neq k_2$ and $i_1\neq i_2$. Now, consider the binary matrix $\boldsymbol{G}^{\prime}$ associated with subgraph $\mathcal{G}(P,C,E^{\prime})$. It can be easily observed that all the elements with value of one are positioned in distinct rows and columns. This implies that $\mathrm{mcm}(\boldsymbol{G})$ gives the maximum number of elements with value of one, which can be placed on the main diagonal after permuting the columns in a specific way. Thus, if $\boldsymbol{H}$ fits $\boldsymbol{G}\in \mathbb{F}_{2}^{r\times m}$, then for any field size $q\geq 2$, we have 

\begin{equation} \nonumber
\mathrm{mcm}(\boldsymbol{G}) = \max_{\substack{h_{k,i}\in \mathbb{F}_q\\  (k,i)\in G}} \ \mathrm{rank} \ \boldsymbol{H}.
\end{equation}

\end{rem}

\begin{exmp}
An example of the MCM problem is illustrated in Figure \ref{fig:mcm}.
\end{exmp}

\subsection{The Maximum Column Distance (MCD) Algorithm} \label{sec:06}
In this subsection, given a binary matrix $\boldsymbol{G}\in \mathbb{F}_{2}^{r\times m}$ with set $G=\{(k,i)\in[r]\times [m], g_{k,i}=1 \}$, the MCD algorithm is proposed to design a matrix $\boldsymbol{H}\in \mathbb{F}_{q}^{r\times m}$ such that $\boldsymbol{H}$ fits $\boldsymbol{G}$, and each submatrix $\boldsymbol{H}_{[k]}^{L}, k\in [r], L\subseteq[m]$ will achieve its maximum possible rank.
\\
First, it is trivial that if we set $\boldsymbol{H}_{[1]}=\boldsymbol{G}_{[1]}$, then each $\boldsymbol{H}_{[1]}^{L}, L\subseteq[m]$ will reach its maximum possible rank. Thus, the MCD algorithm fixes the first row as $\boldsymbol{H}_{[1]}=\boldsymbol{G}_{[1]}$, and in the following, we discuss designing the other rows of $\boldsymbol{H}_{[k]}, k\in [r]\backslash \{1\}$.
First, we begin with the concepts of basis and circuit set in the matroid theory \cite{Rouayheb2010}. Then, for each element $h_{k,i}, (k,i)\in G$, we define a veto set to characterize the values which must be vetoed by the MCD algorithm so that each submatrix $\boldsymbol{H}_{[k]}^{L}$ will achieve its maximum possible rank. This will be followed by a description of the MCD algorithm which is presented as Algorithm \ref{alg:MCD}.

\subsubsection{Prerequisite Material for the MCD Algorithm}
In this part, it is assumed that $k\in [r]\backslash \{1\}$, $i\in L\subseteq [m]$, and all elements of $\boldsymbol{H}_{[k]}^{L}$ are prefixed except element $h_{k,i}$. We also assume that $\mathrm{rank}(\boldsymbol{H}_{[k]}^{\emptyset})=\mathrm{rank}([\ \ ])=0$.
\\
\textbf{Goal}: It is aimed to assign a proper value from $\mathbb{F}_{q}$ to element $h_{k,i}$ such that column $\boldsymbol{H}_{[k]}^{\{i\}}$ will become as much linearly independent of the space spanned by any other columns in $\boldsymbol{H}_{[k]}^{L\backslash \{i\}}$ as possible.

\begin{defn}[Basis and Circuit Set of $\boldsymbol{H}$ \cite{Rouayheb2010}]
We say that set $L$ is an independent set of $\boldsymbol{H}$ if $\mathrm{rank} \ \boldsymbol{H}^{L}=|L|$. Otherwise, $L$ is considered to be a dependent set. A maximal independent set $L$ is referred to as a basis set. A minimal dependent set $L$ is referred to as a circuit set. Let sets $\mathcal{B}_{k}$ and $\mathcal{C}_{k}$, respectively, denote the set of all basis and circuit sets of $\boldsymbol{H}_{[k]}$. It can be shown that
\vspace{-2ex}
\begin{align}
    \mathrm{rank}(\boldsymbol{H}_{[k]})&=\mathrm{rank}(\boldsymbol{H}_{[k]}^{L})=|L|, \ \ \ \ \ \ \ \ \ \ \ \forall L\in \mathcal{B}_{k},
    \nonumber
    \\
    \mathrm{rank}(\boldsymbol{H}_{[k]}^{L})&=\mathrm{rank}(\boldsymbol{H}_{[k]}^{L\backslash \{i\}})=|L|-1, \ \ \forall i\in L,\ \forall L\in \mathcal{C}_{k}. \label{eq:cicuitset}
\end{align}
\end{defn}

\vspace{-2ex}
\begin{prop} \label{prop:vet-value}
Let $L\in \mathcal{C}_{k-1}$. Then, there exists only one value for $h_{k,i}$, denoted by $h_{k,i}(L)$, such that if $h_{k,i}\in \mathbb{F}_{q}\backslash \{h_{k,i}(L)\}$, then $L$ will become an independent set of $\boldsymbol{H}_{[k]}$.
\end{prop}
\begin{proof}
Refer to Appendix \ref{app:prop3-MCD}.
\end{proof}

\begin{defn}[Veto Value] \label{veto-value-def}
Let $L\in \mathcal{C}_{k-1}$. Since the unique value $h_{k,i}(L)$ in Proposition \ref{prop:vet-value} is undesirable, it is referred to as the veto value of set $L$ for element $h_{k,i}$.
\end{defn}

\begin{rem} \label{rem:L-empty-h-zero}
Let $L=\{i\}$. Then, $\{i\}\in \mathcal{C}_{k-1}$ means that column $\boldsymbol{H}_{[k-1]}^{\{i\}}$ is a full-zero column. Thus, by setting $h_{k,i}\in \mathbb{F}_{q}\backslash \{0\}$, set $\{i\}$ will be an independent set of $\boldsymbol{H}_{[k]}$. This means that if $L=\{i\}$, then $h_{k,i}(\{i\})=0$.
\end{rem}
\begin{defn}[Veto Set] \label{def:veto-set}
Let $\mathcal{C}_{k-1,i}^{L}$ denote the set of all $L^{\prime}\subseteq L\backslash \{i\}$ such that $L^{\prime}\cup \{i\}\in \mathcal{C}_{k-1}$. We refer to the set $Z_{k,i}^{L}=\{h_{k,i}(L^{\prime}\cup \{i\})\in \mathbb{F}_q: \ L^{\prime}\in \mathcal{C}_{k-1,i}^{L}\}$ as the veto set of $L$ for $h_{k,i}$. 
\end{defn}
\begin{exmp} \label{exmp:veto-set}
Consider the following matrix $\boldsymbol{H}\in \mathbb{F}_{5}^{3\times 6}$, where $\mathbb{F}_{5}=GF(5)$,
 \begin{equation} \label{eq:veto-set-exmp}
     \begin{bmatrix}
     1  &  0  &  1  &  1  &  0  &  0  \\
     0  &  1  &  0  &  1  &  1  &  1  \\
     0  &  1  &  1  &  2  &  0  & h_{3,6} 
     \end{bmatrix}.
 \end{equation}
We want to find the veto set of $L=[6]$ for $h_{3,6}$. We set $k=3$, $i=6$. Now, it can be seen that $\mathcal{C}_{2,6}^{[6]}=\{L_{j}^{\prime}\subseteq L\backslash \{i\}, j\in [4]\}$, where $L_{1}^{\prime}=\{1,4\}$, $L_{2}^{\prime}=\{2\}$, $L_{3}^{\prime}=\{3,4\}$, and $L_{4}^{\prime}=\{5\}$, which means that each $L_{j}^{\prime}\cup \{6\}, j\in[4]$ is a circuit set of $\boldsymbol{H}_{[2]}$. It can also be checked that the corresponding veto value of each $L_{j}^{\prime}\cup \{6\}$ for $h_{3,6}$ is as follows:
$h_{3,6}(L_{1}^{\prime}\cup \{6\})=2$, $h_{3,6}(L_{2}^{\prime}\cup \{6\})=h_{3,6}(L_{3}^{\prime}\cup \{6\})=1$, and $h_{3,6}(L_{4}^{\prime}\cup \{6\})=0$. Thus, $Z_{3,6}^{[6]}=\{0,1,2\}$.
\end{exmp}
\begin{prop} \label{prop:L-ind-Li-dep}
Assume that $L\backslash \{i\}\in \mathcal{B}_{k-1}$ and $L$ is a dependent set of $\boldsymbol{H}_{[k-1]}$. Now, if $h_{k,i}\in \mathbb{F}_{q}\backslash Z_{k,i}^{L}$, then $L\in \mathcal{B}_{k}$ (we also have $|Z_{k,i}^L|=1$).
\end{prop}
\begin{proof}
Refer to Appendix \ref{app:prop4-MCD}.
\end{proof}
\begin{thm} \label{thm:02}
If $h_{k,i}\in \mathbb{F}_{q}\backslash Z_{k,i}^{L}$, then column $\boldsymbol{H}_{[k]}^{\{i\}}$ will be linearly independent of the space spanned by any other columns in $\boldsymbol{H}_{[k]}^{L\backslash \{i\}}$ to the extent possible.
\end{thm}
\begin{proof}
We assume that the column $\boldsymbol{H}_{[k-1]}^{\{i\}}$ is linearly dependent on columns in $\boldsymbol{H}_{[k-1]}^{L\backslash \{i\}}$, since otherwise the statement always holds and $|Z_{k,i}^{L}|=0$.
\\
Now, let $\mathcal{B}_{k-1,i}^{L}$ denote the set of all $L^{\prime}\subseteq L\backslash \{i\}$ such that $L^{\prime}\in \mathcal{B}_{k-1}$. Since column $\boldsymbol{H}_{[k-1]}^{\{i\}}$ is linearly dependent on columns in $\boldsymbol{H}_{[k-1]}^{L\backslash \{i\}}$, for each $L^{\prime}\in \mathcal{B}_{k-1,i}^{L}$, set $L^{\prime}\cup \{i\}$ will be a dependent set of $\boldsymbol{H}_{[k-1]}$. Now, according to Proposition \ref{prop:L-ind-Li-dep}, if $h_{k,i}\in \mathbb{F}_{q}\backslash Z_{k,i}^{L^{\prime}\cup \{i\}}$ (where $|Z_{k,i}^{L^{\prime}\cup \{i\}}|=1$) for each $L^{\prime}\in \mathcal{B}_{k-1,i}^{L}$, then we will have $L^{\prime}\cup\{i\} \in \mathcal{B}_{k}$.
\\
Thus, by setting $h_{k,i}\in \mathbb{F}_{q}\backslash Z_{k,i}^{L}$, where 

\begin{equation} \label{Z-size}
    Z_{k,i}^{L}=\cup_{L^{\prime} \in \mathcal{B}_{k-1,i}^{L}} Z_{k,i}^{L^{\prime}\cup \{i\}},
\end{equation}
column $\boldsymbol{H}_{[k]}^{\{i\}}$ will be linearly independent of the space spanned by any other columns in $\boldsymbol{H}_{[k]}^{L\backslash \{i\}}$ to the extent possible.
\end{proof}

\subsubsection{Description of the MCD Algorithm} \label{subsec:MCD}
Assume a binary matrix $\boldsymbol{G}\in \mathbb{F}_{2}^{r\times m}$ with support sets $G_k=\{i\in[m], g_{k,i}=1\}, k\in[r]$ and a field of size $q\geq q_{min}$ are given, where $q_{min}={m \choose p}, \ p=\min \{\floor{\frac{m}{2}},r \}$.
The MCD algorithm, presented as Algorithm \ref{alg:MCD}, in transmission $k$, assigns a proper value to the elements $h_{k,i}, i\in G_{k}$ such that each submatrix $\boldsymbol{H}_{[k]}^{L}, k\in [r], L\subseteq [m]$ will reach its maximum possible rank.
\\
First, to guarantee that $\boldsymbol{H}$ will always fit $\boldsymbol{G}$, the MCD algorithm sets $h_{k,i}=0, \forall (k,i)\in ([r]\times [m])\backslash G$ where $G=\{(k,i)\in [r]\times [m], g_{k,i}=1\}$. Then, for the first row $k=1$, we fix $h_{1,i}=1, \forall i\in G_1$, leading to $\boldsymbol{H}_{[1]}=\boldsymbol{G}_{[1]}$. 
\\
In the $k$-th row, the MCD algorithm determines a value for the element $h_{k,i}$ where $i$ is chosen randomly from the support set $G_{k}$. Then, element $i$ is removed from set $G_k$ as $h_{k,i}$ is going to be fixed in this round. Now, all elements of $\boldsymbol{H}_{[k]}^{L}$ where $L=[m]\backslash G_{k}$ are fixed except $h_{k,i}$, which is aimed to be designed so that column $\boldsymbol{H}_{[k]}^{\{i\}}$ will be as much linearly independent of any other columns in $\boldsymbol{H}_{[k]}^{L\backslash \{i\}}$ as possible.
Thus, the MCD algorithm computes the veto set $Z_{k,i}^{L}$, and assigns a randomly chosen value from $\mathbb{F}_q\backslash Z_{k,i}^{L}$ to $h_{k,i}$. We repeat this process for the remaining elements in $G_k$ (this process is repeated for all rows of $\boldsymbol{G}$) until all the elements $h_{k,i}, (k,i)\in G$ are assigned a value outside of their veto set.

\begin{algorithm}
\SetAlgoLined
\KwInput{$\boldsymbol{G}\in \mathbb{F}_{2}^{r\times m}$ and $q\geq q_{min}$}
  \KwOutput{$\boldsymbol{H}=\mathrm{MCD}(\boldsymbol{G}, q)$}
 initialization\;
 set $h_{k,i}=0,\forall (k,i)\in ([r]\times [m])\backslash G$\;
 $h_{1,i}=1, \forall i\in G_1$\;
 $k=2$\;
 \While{$k\leq r$}{
  \While{${G}_{k}\neq \emptyset$}{
    choose an element $i\in {G}_{k}$ at random\; 
    $G_k=G_k\backslash \{i\}$\;
    $L=[m]\backslash G_k$\;
    find the veto set $Z_{k,i}^{L}$\;
    assign a value to $h_{k,i}$ from $\mathbb{F}_q\backslash Z_{k,i}^{L}$ at random\;
    }
    $k\leftarrow k+1$
    }
 \caption{MCD Algorithm}
 \label{alg:MCD}
\end{algorithm}
\begin{prop} \label{thm:MCD}
Let $\boldsymbol{H}=\mathrm{MCD}(\boldsymbol{G}, q)$ where $\boldsymbol{G}\in \mathbb{F}_{2}^{r\times m}$ and $q\geq q_{min}$. Then, each submatrix $\boldsymbol{H}_{[k]}^{L}, k\in [r], L\subseteq[m]$ will reach its maximum possible rank.
\end{prop}
\begin{proof}
We note that the statement holds for the first row as $\boldsymbol{H}_{[1]}=\boldsymbol{G}_{[1]}$. Now, we prove the statement for all $k\in [r]\backslash \{1\}, L\subseteq [m]$ by induction. First, for $k=[2]\backslash \{1\}$ and $L=\{i\}\subseteq [m]$, the statement is correct according to Remark \ref{rem:L-empty-h-zero}. Let $k_1\in [r], i_1\in L_1\subseteq [m]$.
Now, we assume that except element $h_{k_1,i_1}$ all the elements of $\boldsymbol{H}_{[k_1]}^{L_1}$ are already fixed such that each submatrix $\boldsymbol{H}_{[k]}^{L}, k\in [k_1]\backslash \{1\}, L\subseteq L_1$ except $\boldsymbol{H}_{[k_1]}^{L_1}$ reaches its maximum possible rank. Now, based on Theorem \ref{thm:02}, if $h_{k_1,i_1}\in \mathbb{F}_{q}\backslash Z_{k_1,i_1}^{L_1}$, then the column $\boldsymbol{H}_{[k_1]}^{\{i_1\}}$ will be as much linearly independent of the space spanned by any other columns in $\boldsymbol{H}_{[k_1]}^{L_1\backslash \{i_1\}}$ as possible. Thus, now all the submatrices $\boldsymbol{H}_{[k]}^{L}, k\in [k_1]\backslash \{1\}, L\subseteq L_1$ including $\boldsymbol{H}_{[k_1]}^{L_1}$ reach their maximum possible rank, which completes the proof.
\end{proof}

\subsubsection{On the Required Field Size and Complexity of the MCD Algorithm}

\begin{rem} \label{rem:circuit-base-size}
Based on Definition \ref{def:veto-set}, we have $|Z_{k,i}^{L}|\leq |\mathcal{C}_{k-1,i}^{L}|$, where the equality holds if all the veto values inside $Z_{k,i}^{L}$ are distinct. Now, if column $\boldsymbol{H}_{[k-1]}^{\{i\}}$ is already linearly independent of columns in $\boldsymbol{H}_{[k-1]}^{L\backslash \{i\}}$, then $|\mathcal{C}_{k-1,i}^{L}|=0, i\in L$. Otherwise, based on \eqref{Z-size}, we have $|\mathcal{C}_{k-1,i}^{L}|=|\mathcal{B}_{k-1,i}^{L}|$. For this case, for each $k\in [r]\backslash \{1\}, i\in L\subseteq[m]$, we have
\begin{equation} \label{eq:field-complexity-general}
    |Z_{k,i}^{L}|\leq |\mathcal{B}_{k-1,i}^{L}|\leq \max_{l\in \min \{k-1, |L|-1\}} {|L|-1 \choose l}.
\end{equation}
\end{rem}
\begin{prop} \label{prop:field-size}
Given the binary matrix $\boldsymbol{G}\in \mathbb{F}_{2}^{r\times m}$, there always exists a value outside of the veto set for all the elements $h_{k,i}, i\in G_k, k\in [r]$ if the field size is chosen as $q\geq q_{min}={m \choose p}$, where $p=\min \{\floor{\frac{m}{2}},r \}$.
\end{prop}
\begin{proof}
Based on \eqref{eq:field-complexity-general}, for the maximum size of $Z_{k,i}^{L}, k\in [r]\backslash \{1\}, i\in L\subseteq [m]$, we have
\begin{align}
    \max_{\substack{k\in [r]\backslash \{1\}\\ i\in L\subseteq [m] }} |Z_{k,i}^{L}|
    &\leq \max_{\substack{k\in [r]\backslash \{1\}\\ L\subseteq [m], L\neq \emptyset}} \max_{l\in \min \{k-1, |L|-1\}} {|L|-1 \choose l}
    \nonumber
    \\
    &< {m \choose p}= q_{min}
    \nonumber
\end{align}
Thus, if $q\geq q_{min}$, we have $|\mathbb{F}_{q}\backslash Z_{k,i}^{L}|>0$ for all $k\in [r]\backslash \{1\}, i\in L\subseteq[m]$.
\end{proof}
\begin{prop}
Given $\boldsymbol{G}\in \mathbb{F}_{2}^{r\times m}$, the complexity of the MCD algorithm in the worst case scales as $\mathcal{O}(2^{m})$.
\end{prop}
\begin{proof}
First note that for each circuit set we need to run the $\mathrm{rref}$ (reduced row echelon form) to achieve the corresponding veto value for each nonzero element $h_{k,i}, (k,i)\in G$. Now, the worst case scenario for the complexity occurs when all the elements in $\boldsymbol{G}$ are nonzero. In this case, each submatrix $\boldsymbol{H}_{[k]}^{L}$ will become full-rank designed by the MCD algorithm. Thus, if $k> |L|$, then $|\mathcal{C}_{k-1,i}^{L}|=0$ for each $i\in L$. If $k\leq |L|$, we have $|\mathcal{C}_{k-1,i}^{L}|=|\mathcal{B}_{k-1,i}^{L}|= {|L|-1 \choose k-1}$. Thus, the maximum number of times the $\mathrm{rref}$ needs to run will be equal to 
\begin{align}
    \sum_{\substack{k\in [r]\backslash \{1\}\\ i\in L\subseteq [m]}}  |\mathcal{C}_{k-1,i}^{L}|&=\sum_{\substack{k\in [n]\backslash \{1\}\\ n\in [m]}} {n-1 \choose k-1}=\sum_{\substack{n\in [m]}} (2^{n-1}-1)
    \nonumber
    \\
    &= 2^{m}-(m+1),
\end{align}
where $n=|L|\in [m]$. Thus, the complexity scales as $\mathcal{O}(2^{m})$.
\end{proof}
Appendix \ref{app:discuss:field-complexity} provides more discussion on the required field size and complexity of the MCD algorithm for the index coding problem.

\begin{exmp}
Consider the following binary matrix
\begin{equation} \nonumber
      \boldsymbol{G}=
      \begin{bmatrix}
      0 & 1 & 0 & 0 & 0\\
      0 & 1 & 1 & 1 & 1\\
      1 & 0 & 1 & 1 & 1
      \end{bmatrix}.
\end{equation}
Now, given a finite field of size $q\geq q_{min}=10$, we run the MCD algorithm to generate the encoding matrix $\boldsymbol{H}$. Let $q=11$ and $\mathbb{F}_{11}=GF(11)$. 
\\
First, the MCD algorithm sets $h_{k,i}=0$ whenever $g_{k,i}=0, (k,i)\in [3]\times [5]$, and fixes $h_{1,i}=1, \forall i\in G_1=\{2\}$. Now,
we move to the second row $k=2$, where $G_2=\{2,3,4,5\}$. Let $i=4\in G_2$ be chosen randomly. Then, $G_2\leftarrow G_2\backslash \{4\}=\{2,3,5\}$, $L=([5]\backslash G_2)=\{1,4\}$ and $Z_{2,4}^{L}=\{0\}$. Let $h_{2,4}=1\in \mathbb{F}_{11}\backslash \{0\}$. Now, let $i=5\in G_2$ be chosen randomly, then $G_2\leftarrow G_2\backslash \{5\}=\{2,3\}$, $L=\{1,4,5\}$, and $Z_{2,5}^{L}=\{0\}$. Let $h_{2,5}=2\in \mathbb{F}_{11}\backslash \{0\}$. Now, let $i=2\in G_2$. Then, $G_2\leftarrow G_2\backslash \{2\}=\{3\}$, $L=\{1,2,4,5\}$, and $Z_{2,2}^{L}=\{0\}$. Let $h_{2,2}=1\in \mathbb{F}_{11}\backslash \{0\}$. Now, we have $i=3\in G_2$, $G_2\leftarrow G_2\backslash \{3\}=\emptyset$, $L=[5]$, $Z_{2,3}^{L}=\{0\}$. Let $h_{2,3}=3\in \mathbb{F}_{11}\backslash \{0\}$.
\\
Now, since $G_2=\emptyset$, we move to the third row $k=3$, where $G_3=\{1,3,4,5\}$. Let $i=3\in G_3$ be chosen randomly. Then, $G_3\leftarrow G_3\backslash \{3\}=\{1,4,5\}$, $L=\{2,3\}$, $Z_{3,3}^{L}=\{0\}$. Let $h_{3,3}=4\in \mathbb{F}_{11}\backslash \{0\}$. Let $i=4\in G_3$ be chosen at random. Then, $G_3\leftarrow G_3\backslash\{4\}=\{1,5\}$,  $L=\{2,3,4\}$, $Z_{3,4}^{L}=\{5\}$. Let $h_{3,4}=3\in \mathbb{F}_{11}\backslash \{5\}$. Now, let $i=5\in G_3$ be chosen randomly. Then, $G_3\leftarrow G_3\backslash\{5\}=\{1\}$, $L=\{2,3,4,5\}$, $Z_{3,5}^{L}=\{6,10\}$. Let $h_{3,5}=2\in \mathbb{F}_{11}\backslash \{6,10\}$. Now, $i=1\in G_3$, $G_3\leftarrow G_3\backslash \{1\}=\emptyset$, $L=[5]$, and $Z_{3,1}^{L}=\{0\}$. Let $h_{3,1}=5\in \mathbb{F}_{11}\backslash \{0\}$.
\\
Thus, the encoding matrix $\boldsymbol{H}\in \mathbb{F}_{11}^{3\times 5}$ will be as follows
\begin{equation} \nonumber
      \boldsymbol{H}=\mathrm{MCD}(\boldsymbol{G}, q)=
      \begin{bmatrix}
      0 & 1 & 0 & 0 & 0\\
      0 & 1 & 3 & 1 & 2\\
      5 & 0 & 4 & 3 & 2
      \end{bmatrix}.
\end{equation}
It can be verified each submatrix $\boldsymbol{H}_{[k]}^{L}, k\in [3], L\subseteq[5]$ reaches its maximum possible rank.
\end{exmp}

\section{The Proposed UMCD versus the MDS, Recursive, and ICC Coding Schemes} \label{sec:07}
In this section, first we prove that the broadcast rate of the UMCD scheme is always at least as low as the broadcast rate of the MDS scheme. Then, two classes of index coding instances are characterized to show that the gap between the broadcast rates of the recursive and ICC coding schemes and the proposed UMCD coding scheme can grow linearly with the number of messages.

\subsection{The UMCD versus the MDS Coding Scheme}
First, we show that the binary matrix $\boldsymbol{G}$ of the UMCD scheme meets the linear code condition with a specific distance. Then, we prove that for any index coding instance $\mathcal{I}$, we have $\beta_{\text{UMCD}}(\mathcal{I})\leq \beta_{\text{MDS}}(\mathcal{I})$. 
\begin{defn}[Linear Code Condition \cite{DauMDS}]
Assume the binary matrix $\boldsymbol{G}\in \mathbb{F}_{2}^{r\times m}$ with support sets $G_{k}, k\in [r]$ represents the generator matrix of an $(m,r,d)$ linear code, where $d\leq m-r+1$. Then, the linear code condition is defined as follows
\vspace{-1ex}
\begin{equation} \label{eq:MDS-con}
    |F_{K}| \geq d-1+|K|, \ \ \ \forall K\subseteq [r],
\end{equation}
where $F_{K}=\cup_{k}\in K G_{k}$. Note, $F_{K}$ characterizes the indices of nonzero columns in matrix $\boldsymbol{G}_{K}$.
\end{defn}

\begin{lem} \label{lem:MDS-con-mcm}
Let $\boldsymbol{G}\in \mathbb{F}_{2}^{r\times m}$. Suppose that the binary matrix $\boldsymbol{G}_{[k]}, k\in [r]$ meets the condition in \eqref{eq:MDS-con}. Now, if for $i\in [m]$ the following three conditions are satisfied:
\begin{equation} \nonumber
    (1)\ |F_{[k]}|=d-1+k, \ \ \ (2) \ \{i\}\cup A_i\subseteq F_{[k]}, \ \  (3) \ |A_i|\geq d-1,
\end{equation}
then, we will have $\mathrm{mcm}(\boldsymbol{G}_{[k]}^{\{i\}\cup B_i})=1+\mathrm{mcm}(\boldsymbol{G}_{[k]}^{B_i})$.
\end{lem}
\begin{proof}
Refer to Appendix \ref{App:lem2}.
\end{proof}

\begin{lem} \label{lem:UMCD=MDS}
The binary matrix $\boldsymbol{G}\in \mathbb{F}_{2}^{r\times m}$, obtained at the end of the UMCD coding scheme, satisfies the linear code condition in \eqref{eq:MDS-con} with $d=|A|_{\text{min}}+1$.
\end{lem}

\begin{proof}
Refer to Appendix \ref{App:lem3}.
\end{proof}

\begin{thm} \label{thm:UMCD<=MDS}
For an index coding instance $\mathcal{I}$, we have $\beta_{\text{UMCD}}(\mathcal{I})\leq \beta_{\text{MDS}}(\mathcal{I})$.
\end{thm}

\begin{proof}
Let matrix $\boldsymbol{G}\in \mathbb{F}_{2}^{r\times m}$ be the binary matrix of the UMCD scheme. Then, based on Lemma \ref{lem:UMCD=MDS}, matrix $\boldsymbol{G}$ meets the linear code condition in \eqref{eq:MDS-con} with $d=|A|_{\text{min}}+1$. Now, we show that the UMCD scheme guarantees that all the receivers will be satisfied at transmission $r\leq m-|A|_{\text{min}}$.
If we set $r=m-|A|_{\text{min}}$, then due to \eqref{eq:MDS-con}, we have $|F_{[r]}|=m$. So, $F_{[r]}=[m]$. \\
Now, since $\{i\}\cup A_i\subseteq F_{[r]}=[m]$ and $|A_i|\geq |A|_{\text{min}}=d-1, \forall i\in[m]$, the three conditions of Lemma \ref{lem:MDS-con-mcm} are met for $r=m-|A|_{\text{min}}$ and $d=|A|_{\text{min}}+1$. So, we have
\begin{equation}
    \mathrm{mcm}(\boldsymbol{G}_{[r]}^{\{i\}\cup B_i})=1+ \mathrm{mcm}(\boldsymbol{G}_{[r]}^{B_i}), \ \ \forall i\in [m],
\end{equation}
which means that the decoding condition in \eqref{eq:dec-mcm} is met for each receiver $u_i, i\in[m]$. 
So, we have $r\leq m-|A|_{\text{min}}$, which gives $\beta_{\text{UMCD}}(\mathcal{I})\leq \beta_{\text{MDS}}(\mathcal{I})$.
\end{proof}

\subsection{The UMCD versus the Recursive Coding Scheme}
In this subsection, a class of index coding instances are characterized to show that the gap between the broadcast rates of the recursive and the proposed UMCD coding schemes can grow linearly with the number of messages.

\begin{defn}[$\mathcal{I}_{6}(l)$: Class-$\mathcal{I}_{6}$ Index Coding Instances]
   Class-$\mathcal{I}_{6}$ index coding instances are defined as $\mathcal{I}_{6}(l)=\{(i|A_i), i\in[4l+1]\}$, where the side information sets are as follows
    \begin{equation} \nonumber
   \left\{\begin{array}{lc}
       A_{2i-1}&=\{2j: j\in [2l]\backslash\{i\}\}\cup \{4l+1\}, \ \ \ \ \ \ \forall i\in[2l], \nonumber
       \\
       A_{2i}&=\{2i-1\}, \ \ \ \ \ \ \ \ \ \ \ \ \ \ \ \ \ \ \ \ \ \ \ \ \ \ \ \ \ \ \ \forall i\in[2l], \nonumber \\
       A_{4l+1}&=\{2i-1: i\in[2l]\}.  \ \ \ \ \ \ \ \ \ \ \ \ \ \ \ \ \ \ \ \ \ \ \ \ \ \ \ \ \ \ \ \ \nonumber
    \end{array}
    \right.
\end{equation}
Note, the instance $\mathcal{I}_{4}$ is a special case of the class-$\mathcal{I}_6$ instances with $l=1$, i.e., $\mathcal{I}_{6}(1)=\mathcal{I}_{4}$.
\end{defn}

\begin{thm}
For the class-$\mathcal{I}_{6}$ index coding instances, we have
\vspace{-1ex}
   \begin{equation}
       \beta_{\text{R}}(\mathcal{I}_{6}(l))-\beta_{\text{UMCD}}(\mathcal{I}_{6}(l))=l-\frac{1}{2},
   \end{equation}
which scales as $\mathcal{O}(m)$. This means that the gap between the broadcast rates of the recursive coding scheme and our proposed UMCD grows linearly with the number of messages.
\end{thm}
Proof can directly be derived from Propositions \ref{prop:A-1} and \ref{prop:A-2}.

\begin{prop} \label{prop:A-1}
The broadcast rate of the proposed UMCD coding scheme for the class-$\mathcal{I}_{6}$ index coding instances is $\beta_{\text{UMCD}}(\mathcal{I}_{6}(l))=2l+1$.
\end{prop}

\begin{proof}
First, note that $|A_{2i}|=1, \ |A_{2i-1}|=2l, \ \forall i\in[2l],$ and $|A_{4l+1}|=2l$. Now, the UMCD coding scheme begins with one of the receivers indexed by $W=\{2i, i\in [2l]\}$ $u_{2i}, \forall i\in[2l]$. It can checked that for the first $2l$ transmissions, we have $w=\{2i\}$, $G_i=\{2i\}\cup A_{2i}=\{2i-1, 2i\}, \forall i\in [2l]$. So, the first $2l$ rows of the binary matrix will be set as follows
   \begin{equation}
      \boldsymbol{G}_{[2l]}=
      \small
      \begin{bmatrix}
      1 & 1 & 0 & 0 &\hdots & 0 & 0 & 0\\
      0 & 0 & 1 & 1 &\hdots & 0 & 0 & 0\\
      \vdots & \vdots & \vdots & \vdots & \ddots & \vdots & \vdots \\
      0 & 0 & 0 & 0 &\hdots & 1 & 1 & 0
      \end{bmatrix},
   \end{equation}
which satisfies all the receivers $u_{2i}, \forall i\in[2l]$. Now, for the round $2l+1$, we have $W=\{2i-1, i\in [2l+1]\}$. Let $w=2i-1,$ for some randomly chosen $i\in[2l]$. To determine whether other receivers $u_{2j-1}, j\neq i\in[2l]$ can also decode their message, we study the mcm of the following matrix:
\begin{equation}
\boldsymbol{G}_{[2l+1]}^{\{2j-1\}\cup B_{2j-1}}=
\small
      \begin{bmatrix} 
      1 & 0  & \hdots & 0 & 1\\
      0 & 1  &\hdots & 0 & 0\\
      \vdots & \vdots & \ddots & \vdots & \vdots \\
      0 & 0 &\hdots & 1 & 0\\
      0 & 0 &\hdots & 0 & 1
      \end{bmatrix}.
 \end{equation}
It can be seen that the elements of the main diagonal are all equal to 1. Thus, the decoding condition in \eqref{eq:dec-cond} holds for all receiver $u_{2i-1}, i\in[2l]$. Now, for receiver $u_{4l+1}$, we have
\begin{equation}
\boldsymbol{G}_{[2l+1]}^{\{4l+1\}\cup B_{4l+1}}=
\small
      \begin{bmatrix}
      1 & 0  & \hdots & 0\\
      0 & 1  & \hdots & 0\\
      \vdots & \vdots & \ddots & \vdots \\
      1 & 1 &\hdots &  1
      \end{bmatrix}.
 \end{equation}
Similarly, it can be seen that the diagonal elements are all set to 1. Thus, the decoding condition in \eqref{eq:dec-cond} holds for receiver $u_{4l+1}$, which completes the proof.
\end{proof}

\begin{prop} \label{prop:A-2}
For the class-$\mathcal{I}_{6}$ index coding instances, we have $\beta_{\text{R}}(\mathcal{I}_{6}(l))=3l+\frac{1}{2}$.
\end{prop}

\begin{proof}
The proof of Proposition \ref{prop:A-2} appears in Appendix \ref{proof:A-2} of supplemental material section.
\end{proof}

\begin{prop}
For the class-$\mathcal{I}_{6}$ index coding instances, the UMCD coding scheme is optimal.
\end{prop}
\begin{proof}
It can be easily seen that subgraph $M=\{2i: i\in[2l]\}\cup \{4l+1\}$ is acyclic. Therefore, $\beta_{\text{MAIS}}(\mathcal{I}_{6}(l))\geq 2l+1$, which is met by the UMCD coding scheme. This proves the optimality of the UMCD coding for $\mathcal{I}_{6}(l)$.
\end{proof}
\vspace{-3ex}
\subsection{The UMCD versus the ICC Coding Scheme}
In this subsection, a class of index coding instances are characterized to show that the gap between the broadcast rates of the ICC and the proposed UMCD coding schemes can grow linearly with the number of messages.

\begin{defn}[$\mathcal{I}_{7}(l)$: Class-$\mathcal{I}_{7}$ Index Coding Instances]
   Class-$\mathcal{I}_{7}$ index coding instances are defined as $\mathcal{I}_{7}(l)=\{(i|A_i), i\in[5l+3]\}$, where the side information sets are as follows
  \begin{equation} \label{eq:Class-B:side}
  A_i=
   \left\{\begin{array}{lc}
       E_{L_1}\cup (L_2\backslash\{i+(2l+1)\})\cup L_3, \ \ \ \ \forall i\in O_{L_1}, 
       \\
       O_{L_2}\cup (L_1\backslash\{i\})\cup L_3, \ \ \ \ \ \ \ \ \ \ \ \ \ \ \ \ \ \forall i\in E_{L_1}, 
       \\ 
       O_{L_2}\cup (L_1\backslash\{i-(2l+1)\})\cup L_3, \ \ \ \ \ \forall i\in E_{L_2}, 
       \\ 
       E_{L_1}\cup (L_2\backslash\{i\})\cup L_3, \ \ \ \ \ \ \ \ \ \ \ \ \ \ \ \ \ \ \forall i\in O_{L_2}, 
       \\ 
       \{2(i-4l)+3, 2(i-3l-2)\}, \ \ \ \ \ \ \ \ \forall i\in L_3, 
    \end{array}
    \right.
\end{equation}
where, $L_1=[2l+1]$, $L_2=[2l+2:4l+2]$, and $L_3=[4l+3:5l+3]$.
Here, $O_{L_i}$, $E_{L_i}$, respectively, denote a set whose elements are the odd and even numbers inside $L_i, \ i=1,2$. 
\end{defn}

\begin{exmp}
For the $\mathcal{I}_{7}(l=1)$, we have $m=8$, and 
\begin{align}
    A_1&=\{2, 5, 6, 7, 8\},\ A_2=\{1, 3, 5, 7, 8\}, \
    A_3=\{2, 4, 5, 7, 8\},
    \nonumber
    \\
    A_4&=\{2, 3, 5, 7, 8\}, \
    A_5=\{2, 3, 6, 7, 8\}, \ A_6=\{1, 2, 5, 7, 8\},
    \nonumber
    \\
    A_7&=\{1, 4\},\ \ \ \ \ \ \ \ \ \ \ A_8=\{3, 6\}.
   \nonumber
\end{align}
\end{exmp}
\begin{thm} \label{thm:class-ICC-UMCD}
For the class-$\mathcal{I}_{7}$ index coding instances, we have
   \begin{equation}
       \beta_{\text{ICC}}(\mathcal{I}_{7}(l))-\beta_{\text{UMCD}}(\mathcal{I}_{7}(l))\geq 2l-\frac{1}{2},
   \end{equation}
which scales as $\mathcal{O}(m)$. This means the gap between the broadcast rates of the ICC coding scheme and our proposed UMCD coding scheme grows linearly with the number of messages. 
\end{thm}
\begin{proof}
The proof of Theorem \ref{thm:class-ICC-UMCD} appears in Appendix \ref{proof:thm-class-ICC} of supplemental material section.
\end{proof}

\begin{prop}
For the class-$\mathcal{I}_{7}$ index coding instances, the UMCD coding scheme is optimal.
\end{prop}

\begin{proof}
It can be easily seen that subgraph $M=L_3\cup\{2\}$ is acyclic. Therefore, $\beta_{\text{MAIS}}(\mathcal{I}_{7}(l))\geq l+2$, which is met by the broadcast rate of the UMCD coding scheme. This proves the optimality of the UMCD coding for $\mathcal{I}_{7}(l)$.
\end{proof}

\section{Extensions of the UMCD Coding Scheme} \label{sec:09}
Inspired by the PPC and FPCC coding schemes, in this section, we extend the UMCD coding scheme to its partial and fractional versions, where each is, respectively, motivated through Examples \ref{exm:partial} and \ref{exm:frac}.

\subsection{Partial UMCD (P-UMCD) Coding Scheme}

\begin{exmp} \label{exm:partial}
Consider the index coding instance $\mathcal{I}_{8}=\{(1|2,5), (2|3,4), (3|2,4), (4|2,3)$, $(5|1,4)\}$.
For this instance $\beta_{\text{UMCD}}(\mathcal{I}_{8})=3$, while $\beta({\mathcal{I}_{8}})=2$, indicating that the UMCD is suboptimal. However, we first partition $\mathcal{I}_{8}$ into subinstances $M_1=\{1,5\}$ and $M_2=\{2,3,4\}$. Then we apply the UMCD coding scheme for each subinstance separately, which results in $\beta_{\text{UMCD}}(M_1)=\beta_{\text{UMCD}}(M_2)=1$. This achieves the broadcast rate.
\end{exmp}

\begin{defn}[P-UMCD Coding Scheme]
Given an index coding instance $\mathcal{I}=\{(i|A_i), i\in[m]\}$, the P-UMCD coding scheme first partitions $\mathcal{I}$ into subinstances $M_1,\dots, M_n, n\leq m$, satisfying the condition in \eqref{eq:partition-con}. Then, the UMCD coding scheme is used for solving each subinstance individually. So, the broadcast rate of the P-UMCD will be achieved as below
\begin{equation}
    \beta_{\text{P-UMCD}}(M_j, j\in[n])=\sum_{j\in [n]} \beta_{\text{UMCD}}(M_j).
\end{equation}
\end{defn}
\vspace{-1ex}

\begin{thm}
Given an index coding instance $\mathcal{I}$, the broadcast rate $\beta(\mathcal{I})$ is upper bounded by $\beta_{\text{P-UMCD}}(\mathcal{I})$, which is the solution to the following optimization problem
\vspace{-1ex}
\begin{equation} \label{eq:opt-F-UMCD}
    \min_{M_j, j\in [n]} \beta_{\text{P-UMCD}}(M_j, j\in[n]),
\end{equation}
subject to the constraint in \eqref{eq:partition-con}.
\end{thm}

\begin{prop}
For the index coding instance $\mathcal{I}$, we have $\beta_{\text{P-UMCD}}(\mathcal{I})\leq \beta_{\text{PCC}}(\mathcal{I})$.
\end{prop}

\begin{proof}
This can directly be concluded from Theorem \ref{thm:UMCD<=MDS}, which states that for each subinstance $M_j, j\in[n]$, we always have $\beta_{\text{UMCD}}(M_j)\leq \beta_{\text{MDS}}(M_j)$.
\end{proof}

\subsection{Fractional Partial UMCD (FP-UMCD) Coding Scheme}

\begin{exmp} \label{exm:frac}
Consider the index coding instance $\mathcal{I}_{9}=\{(1|2,5), (2|1,3), (3|2,4), (4|3,5)$, $(5|1,4)\}$.
For this instance, we have $\beta_{\text{UMCD}}(\mathcal{I}_{9})=\beta_{\text{P-UMCD}}(\mathcal{I}_{9})=3$, while $\beta(\mathcal{I}_{9})=2.5$, which requires a vector index code to be optimal. This can be achieved by applying time sharing over the subinstances $M_1=\{1,2\}$, $M_2=\{2,3\}$, $M_3=\{3,4\}$, $M_4=\{4,5\}$, and $M_5=\{1,5\}$, each solved by the UMCD coding scheme.
\end{exmp}

\begin{defn}[FP-UMCD Coding Scheme]
Given an index coding instance $\mathcal{I}=\{(i|A_i), i\in[m]\}$, the FP-UMCD coding scheme applies time sharing over the UMSD coding solution of the subinstances $M_j, j\in [n]$ for some $n\in 2^{[m]}$, as follows
\begin{equation}
    \beta_{\text{FP-UMCD}}(M_j,\gamma_j, j\in [n])=\sum_{j\in [n]} \gamma_j\beta_{\text{UMCD}}(M_j),
\end{equation}
\vspace{-4ex}
such that 
\begin{align} 
    \gamma_j &\in [0,1],\ \ \ \ \sum_{j\in P_i} \gamma_j \geq 1,  \ \ \ \ \ \forall i\in[m], 
    \label{eq:partition-frac1}
\end{align}
where $P_i=\{j\in [n], i\in M_j\}$. Note that subsets $M_j, j\in[n]$ may have overlap each other.
\end{defn}

\begin{thm}
Given an index coding instance $\mathcal{I}$, the broadcast rate $\beta(\mathcal{I})$ is upper bounded by $\beta_{\text{FP-UMCD}}(\mathcal{I})$, which is the solution of the following optimization problem
\begin{equation} \label{eq:opt-FP-UMCD}
    \min_{M_j,\gamma_j, j\in [n]} \beta_{\text{FP-UMCD}}(M_j,\gamma_j, j\in [n]),
\end{equation}
subject to the constraint in \eqref{eq:partition-frac1}.
\end{thm}

\begin{prop}
For an index coding instance $\mathcal{I}$, we have $\beta_{\text{FP-UMCD}}(\mathcal{I})\leq \beta_{\text{FPCC}}(\mathcal{I})$.
\end{prop}

\begin{proof}
This can directly be concluded from Theorem \ref{thm:UMCD<=MDS}, which states that for each subinstance $M_j, j\in[n]$, we always have $\beta_{\text{UMCD}}(M_j)\leq \beta_{\text{MDS}}(M_j)$.
\end{proof}

\begin{rem} \label{rem:complexity-extention}
Note that the broadcast rate of the UMCD algorithm for the subinstance $M_i$, i.e., $\beta_{UMCD}(M_i), i\in [n]$, is achieved independently of using the MCD algorithm. Thus, finding the optimal solution (optimal subinstances) of the P-UMCD and FP-UMCD schemes, respectively in \eqref{eq:opt-F-UMCD} and \eqref{eq:opt-FP-UMCD} can be achieved only using the MCM algorithm. Since the MCM is a polynomial-time algorithm, it can be shown that the computational complexity of finding the optimal solution of \eqref{eq:opt-F-UMCD} and \eqref{eq:opt-FP-UMCD} can be achieved, respectively, from the optimal solution of the PCC and FPCC schemes, respectively in \eqref{eq:opt-PCC} and \eqref{eq:FPCC-opt}, by a polynomial reduction.
Once the optimal solutions \eqref{eq:opt-F-UMCD} and \eqref{eq:opt-FP-UMCD} are found using the MCM algorithm, then the encoding matrix for each optimal subinstance can be generated using the MCD algorithm. This highlights the importance of separating the MCM and the MCD algorithms in the UMCD coding scheme to reduce the computational complexity of the P-UMCD and FP-UMCD coding schemes.
\end{rem}

\begin{rem}
Among all the existing coding schemes, only minrank and composite coding schemes are optimal for index coding instances of up to and including five receivers. Now, we find that the proposed FP-UMCD coding scheme can also achieve the broadcast rate of the instances with up to and including five receivers (9846 non-isomorphic instances). However, for six receivers, there are some index coding instances for which the FP-UMCD coding scheme is suboptimal. One of these instances is illustrated in the following example.
\end{rem}

\begin{exmp}
For the index coding instance $\mathcal{I}_{10}=\{(1|4,5)$, $(2|1,6), (3|1,2,4,5,6), (4|1,2,3)$, $(5|2,3), (6|3,4)\}$,
we have $\beta(\mathcal{I}_{10})=3$, which can be achieved by index code $\{y_1=x_1+x_4+x_5, y_2=x_1+x_2+x_6, y_3=x_2+x_3+x_5\}$. However, it can be verified that $\beta_{\text{FP-UMCD}}(\mathcal{I}_{10})=\frac{10}{3}<\beta_{\text{ICC}}(\mathcal{I}_{10})=\beta_{\text{R}}(\mathcal{I}_{10})=3.5$, implying that while the proposed FP-UMCD coding scheme outperforms both the recursive and ICC coding schemes, it is still suboptimal for $\mathcal{I}_{10}$.
\end{exmp}

\section{Conclusion} \label{sec:10}
In this paper, a new index coding scheme, referred to as update-based maximum column distance (UMCD) was proposed in which for each step of transmission, a linear coded message is designed with the aim of satisfying at least one receiver with the minimum size of side information. The problem is updated after each transmission using the polynomial-time Hopcroft-Karp algorithm, which is able to identify the satisfied receivers in each transmission. The maximum column distance (MCD) algorithm was proposed to generate the encoding matrix of the UMCD coding scheme such that each subset of its columns achieves its maximum possible rank.
Several index coding instances were provided to show that the UMCD can outperform the ICC \cite{Thapa2017} and recursive \cite{Arbabjolfaei2014} schemes. Moreover, we proved that the broadcast rate of the proposed UMCD scheme is never larger than the MDS coding scheme.
Then, we characterized two classes of index coding instances for which the gap between the broadcast rates of the recursive and ICC coding schemes and the UMCD coding scheme grows linearly with the number of messages. The UMCD coding scheme was extended to its fractional version by dividing the instance into subinstances, and then applying the time sharing over their UMCD solutions. The fractional UMCD is optimal for all index coding instances with up to and including five receivers. Extending the UMCD coding scheme for the more general index coding scenarios, including the groupcast index coding \cite{Tehrani2012, Shanmugam2014, Unal2016, Sharififar-groupcast}, the secure index coding \cite{Narayanan-secure,Dau,Ong-secure,Mojahedian-secure,Liu-secure}, and the distributed index coding \cite{Liu-distri, Liu2020} would be directions for future studies.

\appendices

\section{Proof of Proposition \ref{prop-lineardecoding-condition}} \label{proof:prop-lineardecoding-condition}
First note that we must have $\mathrm{rank}\ \boldsymbol{H}^{\{i\}}=t,\forall i\in[m]$. For the \textit{if} condition, we suppose that \eqref{eq:dec-cond} holds. Now, we show that there is a decoder function $\psi_{\mathcal{I}}^{i}(\boldsymbol{y}, S_i)=\boldsymbol{x}_i$, which can correctly decode $\boldsymbol{x}_i$ as follows
\begin{align} 
       \boldsymbol{y}-\boldsymbol{H}\boldsymbol{x}&= \sum_{j\in [m]} \boldsymbol{H}^{\{j\}} \boldsymbol{x}_j - \sum_{j\in A_i} \boldsymbol{H}^{\{j\}} \boldsymbol{x}_j \nonumber
       \\ 
       &=\boldsymbol{H}^{\{i\}} \boldsymbol{x}_i + \sum_{j\in B_i} \boldsymbol{H}^{\{j\}} \boldsymbol{x}_j. \label{eq:1:2}
\end{align}
Let $\mathrm{rank}\ \boldsymbol{H}^{B_{i}}=c_it$, then we partition $B_i=C_i\cup (B_i\backslash C_i)$ such that $|C_i|=c_i$ and $\{\boldsymbol{H}^{\{j\}}: j\in C_i\}$ represents the set of local encoding matrices whose columns are linearly independent, i.e., $\mathrm{rank}\ \boldsymbol{H}^{C_{i}}=c_it$. Then, the remaining local encoding matrices $\boldsymbol{H}^{\{l\}}, \forall l\in (B_i\backslash C_i)$ can be expressed as
\begin{equation} \nonumber 
       \boldsymbol{H}^{\{l\}}= \sum_{j\in C_i} \boldsymbol{H}^{\{j\}}\boldsymbol{p}_{l,j},  
\end{equation}
where, $\boldsymbol{p}_{l,j}\in \mathbb{F}_q^{t\times 1}, \forall j\in C_i$. So, \eqref{eq:1:2} will be equal to
\begin{equation} \label{eq:1:03}
       \boldsymbol{H}^{\{i\}} \boldsymbol{x}_i + \sum_{j\in C_i} \boldsymbol{H}^{\{j\}} (\boldsymbol{x}_j+ \sum_{l\in B_i\backslash C_i} \boldsymbol{p}_{l,j} \boldsymbol{x}_l).
\end{equation}
Since \eqref{eq:dec-cond} holds, then all the columns of $\boldsymbol{H}^{\{i\}\cup C_i}$ are linearly independent. Hence, the set of messages $\{\boldsymbol{x}_i\}\cup \{\boldsymbol{x}_j+ \sum_{l\in B_i\backslash C_i} \boldsymbol{p}_{l,j} \boldsymbol{x}_l, j\in C_i\}$ can be decoded by receiver $u_i$, which completes the proof for the \textit{if} condition. 
\\
\textit{Conversely}, according to the polymatroidal bound \cite{Arbabjolfaei2018}, in order to decode all the messages correctly for the described system model, the following constraint must be met by any polymatroidal function $f: 2^{[m]} \rightarrow r$
\begin{equation} \nonumber
       f(\{i\}\cup B_i)\geq f(\{i\})+f(B_i),  \ \ \ \forall i\in[m].
\end{equation}
Since the $\mathrm{rank}$ function is polymatroidal, by setting $f(L)=\mathrm{rank}(\boldsymbol{H}^{L})$, we must have,
\begin{align} 
    \mathrm{rank}\ \boldsymbol{H}^{\{i\}\cup B_i}&\geq \mathrm{rank} \ \boldsymbol{H}^{\{i\}}+\mathrm{rank}\ \boldsymbol{H}^{B_{i}}
    \nonumber
    \\
    &=t+\mathrm{rank} \ \boldsymbol{H}^{B_i}. \label{eq:conv1}
\end{align}
On the other hand, it can be easily observed that
\begin{align}
    \mathrm{rank} \ \boldsymbol{H}^{\{i\}\cup B_i} &\leq \mathrm{rank}\ \boldsymbol{H}^{\{i\}} +  \mathrm{rank} \ \boldsymbol{H}^{B_i}  
    \nonumber
    \\
    &=t+ \mathrm{rank} \ \boldsymbol{H}^{B_i}. \label{eq:conv2}
\end{align}
Now, combining \eqref{eq:conv1} and \eqref{eq:conv2} will complete the proof.

\section{A Brief Overview of the PCC, FPCC, Recursive and ICC Coding Schemes}  \label{FirstAppendix}

\begin{defn}[PCC Coding Scheme]
In the PCC scheme, the index coding instance $\mathcal{I}$ is partitioned into $n\leq m$ subinstances $M_i, i\in[n]$, satisfying the condition in \eqref{eq:partition-con}. Then, the PCC coding scheme solves each subinstace using the MDS code. So, the broadcast rate of the PCC scheme, $\beta_{\text{PCC}}(\mathcal{I})$, will be achieved by solving the following optimization problem
\begin{equation} \label{eq:opt-PCC}
    \min_{M_j, j\in[n]} \sum_{j\in [n]} \beta_{\text{MDS}}(M_j),
\end{equation}
subject to the constraint in \eqref{eq:partition-con}. 
Each subset $M_i, j\in[n]$ is called a partial clique set.
\end{defn}

\begin{rem}
For the broadcast rate of the MDS code for each subset $M_j$, we have
\begin{equation} \label{eq:MDS-local}
    \beta_{\text{MDS}}(M_j)=|M_j|- \min_{i\in M_j} |M_j\cap A_i|,
\end{equation}
where $M_j\cap A_i$ is the local side information set of receiver $u_i$ with regards to subset $M_j$.
\end{rem}

\begin{defn}[FPCC Coding Scheme]
The FPCC scheme is an extension of the PPC scheme to the vector version through time-sharing over the MDS solutions of the subinstances $M_j, j\in[n]$ for some $n\in 2^m$. Thus, the broadcast rate of the FPCC scheme, $\beta_{\text{FPCC}}(\mathcal{I})$, will be equal to the solution of the following optimization problem
\begin{equation} \label{eq:FPCC-opt}
    \min_{M_j, \gamma_j, j\in[n]} \sum_{j\in [n]} \gamma_j\beta_{\text{MDS}}(M_j),
\end{equation}
subject to the constraints
\begin{align}
    \sum_{j\in P_i} \gamma_j &\geq 1,  \ \ \ \ \ \forall i\in [m], \label{eq:FPCC:cond1}
    \\
    \gamma_j &\in [0,1],
    \label{eq:FPCC:cond2}
\end{align}
where $P_i=\{j\in[n], i\in M_j\}$.
\end{defn}

\begin{defn}[Recursive Coding Scheme]
The recursive scheme begins with fixing an initial broadcast rate $\beta_{\text{R}}(\{i\})=1, \forall i\in[m]$. Now, let $M_j\subsetneq M, j\in[n]$ for some $n\in 2^{|M|}$.
Then, the broadcast rate for subinstance $\beta_{\text{R}}(M\subseteq[m])$ is recursively achieved as follows
\begin{equation} \label{eq:recursive:broadcastrate}
   \beta_{\text{R}}(M)=\min_{M_j, \gamma_j, j\in[n]} \max_{i\in M} \sum_{j\in[n], M_j\not\subseteq A_i} \gamma_{j}\beta_{\text{R}}(M_j) ,
\end{equation}
subject to
\begin{align} \label{eq:partition-frac-rec}
    \sum_{j\in P_i} \gamma_j &\geq 1,  \ \ \ \ \ \forall i\in M, \\
    \gamma_j &\in [0,1], \label{eq:partition-frac-rec1}
\end{align}
where $P_i=\{j\in[n], i\in M_j\}$.
\end{defn}

\begin{prop}
Given an index coding instance $\mathcal{I}$, we have \cite{Arbabjolfaei2018}
\begin{equation}
    \beta_{\text{R}}(\mathcal{I})\leq \beta_{\text{FPCC}}(\mathcal{I})\leq \beta_{\text{PCC}}(\mathcal{I})\leq \beta_{\text{MDS}}(\mathcal{I}).
\end{equation}
\end{prop}

\begin{defn}[ICC Coding Scheme]
Given an index coding instance $\mathcal{I}=\{(i|A_i), i\in[m]\}$, the ICC coding scheme first identifies all the ICC-structured subgraphs such as $M\subseteq[m]$, as follows:

\begin{itemize}[leftmargin=*]
    \item Inner vertex set $J\subseteq M$, which consists of vertices, where there exists a path between each ordered pair of vertices $(i,j)\in J, i\neq j$ such that the path does not include any other vertex inside $J\backslash\{i,j\}$.
    \item $J$-path condition: There is only one $J$-path between any pair of vertices inside inner vertex set $J$, where the $J$-path is defined as follows: A path in which only the first and the last vertices (distinct vertices) belong to $J$.
    \item $J$-cycle condition: There is no $J$-cycle, where $J$-cycle is defined as follows: If in the $J$-path, the first and the last vertices are the same, it is considered as a $J$-cycle.
\end{itemize}
It has been shown that $\beta_{\text{ICC}}(M)=|M|-|J|+1$, where all the savings in the transmissions are due to the inner vertex set. Then, the broadcast rate of the ICC coding scheme for $\mathcal{I}$ is achieved as follows
\begin{equation}
    \min_{M_j, \gamma_j, j\in[n]} \sum_{j\in[n]} \gamma_{j} \beta_{\text{ICC}}(M_j),
\end{equation}
subject to the constraints in \eqref{eq:FPCC:cond1} and \eqref{eq:FPCC:cond2}.
\end{defn}

\begin{conj} \label{conj:1}
Given an index coding instance $\mathcal{I}$, it is conjectured in \cite{Thapa2017} that
\begin{equation}
    \beta_{\text{ICC}}(\mathcal{I})\leq\beta_{\text{FPCC}}(\mathcal{I}).
\end{equation}
\end{conj}

\begin{rem}
Neither the recursive nor the ICC is always outperformed by the other scheme. Consider the index coding instances $\mathcal{I}_{4}$ and $\mathcal{I}_{5}$, where $3=\beta_{\text{ICC}}(\mathcal{I}_4)<\beta_{\text{R}}(\mathcal{I}_4)=3.5$ and $2=\beta_{\text{R}}(\mathcal{I}_5)<\beta_{\text{ICC}}(\mathcal{I}_5)=2.5$.
\end{rem}

Now, we provide Propositions \ref{prop: FPCC-modification} and \ref{prop:minimal}, which will be used in the proof of Lemmas \ref{lem1:prop4} and \ref{lem1:prop7}.

\begin{prop}  \label{prop: FPCC-modification}
If for the constraint $\sum_{j\in P_i} \gamma_j \geq 1$ in the optimization problem of the FPCC, recursive, ICC, and the proposed FP-UMCD schemes, we consider only its equality case, i.e., $\sum_{j\in P_i} \gamma_j = 1, \forall i\in [m]$, the solution does not change.
\end{prop}

The proof can be achieved using standard techniques in the linear programming (LP).

\begin{defn} [Minimal Partial Clique Set]
The partial clique set $M\subseteq[m]$ is said to be minimal if its broadcast rate $\beta_{\text{MDS}}(M)$ cannot be reduced by being further partitioned into any subsets $M_j\subseteq M, j\in[n]$. 
This means
\begin{equation}
    \beta_{\text{MDS}}(M)\leq \sum_{j\in [n]} \beta_{\text{MDS}}(M_j).
\end{equation}
For example, all the cycles and cliques are minimal partial clique sets. 
Note that, any partial clique can be partitioned into minimal partial cliques.
\end{defn}

\begin{prop}  \label{prop:minimal}
The optimal solution of the FPCC scheme in \eqref{eq:FPCC-opt} can always be expressed in minimal partial clique sets.
\end{prop}

\begin{proof}
First, we suppose that subsets $M_j, j\in [n]$ give the optimal solution in \eqref{eq:FPCC-opt} such that subsets $M_j, j\in[r-1]$ are minimal partial clique set, while $M_j, j\in[r:n]$ are not minimal. Thus,
\begin{align}
       \beta_{\text{FPCC}}(\mathcal{I})&= \sum_{j\in[n]} \gamma_{j}\beta_{\text{MDS}}(M_j) 
       \label{eq:opt:minimal-pc0}
       \\
       &=\sum_{j\in[r-1]} \gamma_{j}\beta_{\text{MDS}}(M_j) + \sum_{j\in[r:n]} \gamma_{j}\beta_{\text{MDS}}(M_j). \label{eq:opt:minimal-pc}
\end{align}
Now, assume that each non-minimal partial clique set $M_j, j\in [r:n]$ is further partitioned into $p_j+1$ minimal partial clique sets $M_{j,z}, z\in[0:p_j]$. Then, for the second term in \eqref{eq:opt:minimal-pc}, we have
\begin{align}
         \sum_{j\in[r:n]} \gamma_{j}\beta_{\text{MDS}}(M_j) &\geq \sum_{j\in[r:n]}\gamma_{j} \sum_{z\in[0:p_j]} \beta_{\text{MDS}}(M_{j,z})
         \label{eq:113}
\end{align}
Now, if for all $z\in [0: p_j], \forall j\in [r:n]$, we set 
\begin{align}
    M_{j+z+\sum\limits_{z\in [j-1]} p_{z}}^{\prime} &= M_{j,z},
    \\
    \gamma_{j+z+\sum\limits_{z\in [j-1]} p_{z}}^{\prime}&=\gamma_{j},
\end{align}
then the right hand side of \eqref{eq:113} will be equal to
\begin{align}
    \sum_{j\in[r:r+b]} \gamma_{j}^{\prime}\beta_{\text{MDS}}(M_j^{\prime}), \label{eq:114}
\end{align}
where $b=n+\sum\limits_{j\in[r:n]} p_j$. Since subsets $M_{j,z}$ are minimal partial clique set, so are subsets $M_{j}^{\prime}$. This means the second term of \eqref{eq:opt:minimal-pc} can be expressed in minimal partial clique sets. 

Now, from \eqref{eq:opt:minimal-pc0}, \dots, \eqref{eq:114}, we have
\begin{equation}
    \sum_{j\in[n]} \gamma_{j}\beta_{\text{MDS}}(M_j)\geq \sum_{j\in[r+b]} \gamma_{j}^{\prime}\beta_{\text{MDS}}(M_j^{\prime}),
\end{equation}
where for $j\in [r-1]$, we set $M_j^{\prime}=M_j$ and $\gamma_j^{\prime}=\gamma_j$.
This means that if subsets $M_j, j\in [n]$ with coefficients $\gamma_{j}$ give the optimal solution in \eqref{eq:FPCC-opt}, so do subsets $M_{j}^{\prime}, j\in [r+b]$ with coefficients $\gamma_{j}^{\prime}$. Thus, the optimal solution of \eqref{eq:FPCC-opt} can always be expressed in minimal partial cliques.
\end{proof}

\begin{rem}
In \cite{Thapa2017}, the authors proved that for any minimal partial clique set, the ICC scheme performs at least as well as the FPCC scheme. Then, they provided Conjecture \ref{conj:1}. Now, having proved Proposition \ref{prop:minimal}, this conjecture is settled in the affirmative.
\end{rem}

\begin{defn} [Minimal Recursive Set]
The recursive set $M\subseteq[m]$ is said to be minimal if its broadcast rate $\beta_{\text{R}}(M)$ cannot be reduced by being further partitioned into any subsets $M_j\subseteq M, j\in[n]$. 
This means
\begin{equation}
    \beta_{\text{R}}(M)\leq \sum_{j\in [n]} \beta_{\text{R}}(M_j).
\end{equation}
For instance, all the cliques and minimal cycles are minimal recursive sets. 
\end{defn}

\begin{prop}  \label{prop:minimal-recursive}
The optimal solution of the recursive scheme in \eqref{eq:recursive:broadcastrate} can always be expressed in minimal recursive sets.
\end{prop}

\begin{proof}
The proof can be easily achieved by replacing $\beta_{\text{MDS}}$ with $\beta_{\text{R}}$ in the proof of Proposition \ref{prop:minimal}.
\end{proof}

\section{Example \ref{exm:motive:3}} \label{app:minrank~UMCD}
Since the proposed UMCD scheme is a scalar linear code, we make a comparison between the broadcast rates of the UMCD and the scalar binary minrank for the index coding instance $\mathcal{I}_{3}=\{(1|3,5,6), (2|1,4,5,6), (3|2,4,6),(4|1,2,3,5), (5|1,2,3,4)$, $(6|1,2,4,5)\}$, depicted in Figure \ref{fig:motiv:3}. One can verify that the broadcast rate of the scalar binary minrank is 3. Now, assume that the server first targets receivers with the minimum size of side information. It begins with receiver $u_1$ and transmits a linear combination of messages $x_1$ (requested message) and $x_i, i\in A_1$ as follows
\begin{equation}
    y_1=[h_{1,1}\ 0\ h_{1,3}\ 0\ h_{1,5}\ h_{1,6}]\ \boldsymbol{x},
\end{equation}
where $h_{1,j}\in \mathbb{F}_{q}$ and $\boldsymbol{x}=[x_1\ x_2\ x_3\ x_4\ x_5\ x_6]$. This transmission satisfies only receiver $u_1$. Then, the server targets receiver $u_3$ by sending a linear combination of messages $x_3$ (requested) and $x_i, i\in A_3$ as follows
\begin{equation}
    y_2=[0\ h_{2,2}\ h_{2,3}\ h_{2,4}\ 0\ h_{2,6}]\ \boldsymbol{x}.
\end{equation}
Now, we set $q=3$, and fix $h_{1,1}=h_{1,3}=h_{1,5}=h_{1,6}=h_{2,2}=h_{2,3}=h_{2,4}=1$ and $h_{2,6}=2$. The reason of fixing $h_{2,6}=2$ is to make the two column vectors $\boldsymbol{H}^{\{3\}}=[h_{1,3}\ h_{2,3}]^T$ and $\boldsymbol{H}^{\{6\}}=[h_{1,6}\ h_{2,6}]^T$ linearly independent so that receiver $u_6$ is able to to decode its requested message $x_6$. Now, it can be checked that the remaining receivers are able to decode their requested message upon receiving the coded messages $y_1$ and $ y_2$. This index code achieves the MAIS bound and, so it is optimal for $\mathcal{I}_{3}$. This simple instance illustrates how the proposed UMCD coding scheme can outperform the scalar binary minrank.

\section{Proof of Propositions  \ref{prop:vet-value} and \ref{prop:L-ind-Li-dep}}  \label{app:prop-MCD}
\subsection{Proof of Proposition \ref{prop:vet-value}}\label{app:prop3-MCD}

First, since $L$ is a circuit set of $\boldsymbol{H}_{[k-1]}$ and any circuit set is a minimal dependent set, we have
\begin{equation} \nonumber
    \boldsymbol{H}_{[k-1]}^{\{i\}}=\boldsymbol{H}_{[k-1]}^{L\backslash\{i\}}\boldsymbol{f},
\end{equation}
where $\boldsymbol{f}=[f_{1},\dots,f_{|L|-1}]^T$ is a unique vector and $f_{i}$ must be nonzero for all $i\in{[|L|-1]}$, since otherwise it contradicts \eqref{eq:cicuitset}. Note that vector $\boldsymbol{f}$ can be achieved using reduced row echelon form ($\mathrm{rref}$) as follows
\vspace{.1mm}
\begin{align}
    \mathrm{rref} \ 
    \left [\begin{array}{c|c}
       \boldsymbol{H}_{[k-1]}^{{L\backslash\{i\}}}   &  \boldsymbol{H}_{[k-1]}^{\{i\}}
    \end{array}
    \right ]
    &=
    \left [\begin{array}{c|c}
       \boldsymbol{I}_{|L|-1}                 &  \boldsymbol{f} \\
       \hline
       \boldsymbol{0}_{|L|-1}                 & 0
     \end{array}
    \right ] \nonumber
    \\
    &\overset{r}{=}
    \left [\begin{array}{c|c}
       \boldsymbol{I}_{|L|-1}                     &  \boldsymbol{f}
       \end{array}
    \right ],
    \label{eq:07}
\end{align}
where $\boldsymbol{H}_1\overset{r}{=}\boldsymbol{H}_2$ means that matrices $\boldsymbol{H}_1$ and $\boldsymbol{H}_2$ have an equal rank.
Now, we have
\begin{align}
    \mathrm{rref} \ 
     \boldsymbol{H}_{[k]}^{{L}}
     &=
    \mathrm{rref} \ 
    \left [\begin{array}{c|c}
       \boldsymbol{H}_{[k-1]}^{{L\backslash \{i\}}}       &  \boldsymbol{H}_{[k-1]}^{\{i\}} \\
       \hline
       \boldsymbol{H}_{\{k\}}^{L\backslash{\{i\}}}    &      h_{k,i}
    \end{array}
    \right ]
    \nonumber
    \\
    &\overset{r}{=}
    \mathrm{rref} \
    \left [\begin{array}{c|c}
       \boldsymbol{I}_{|L|-1}                         &  \boldsymbol{f}  \\
       \hline
       \boldsymbol{H}_{\{k\}}^{L\backslash{\{i\}}}  &   h_{k,i}
    \end{array}
    \right ]
    \label{eq:08}
    \\
    &=
    \mathrm{rref} \
    \left [\begin{array}{c|c}
        \boldsymbol{I}_{|L|-1}                &     \boldsymbol{f}  \\
       \hline
       \boldsymbol{0}_{|L|-1}                 &      h_{k,i}^{\ast}
    \end{array}
    \right ],
    \label{eq:09}
\end{align}
where $\eqref{eq:08}$ is due to $\eqref{eq:07}$ and $\eqref{eq:09}$ is achieved by running $\mathrm{rref}$ over the last row such that
\vspace{.1mm}
\begin{equation} \nonumber
    h_{k,i}^{\ast}=h_{k,i}-\boldsymbol{H}_{\{k\}}^{L\backslash\{i\}}\boldsymbol{f}.
\end{equation}
Thus, only the value $h_{k,i}=\boldsymbol{H}_{\{k\}}^{L\backslash \{i\}}\boldsymbol{f}$, denoted by $h_{k,i}(L)$ will cause $h_{k,i}^{\ast}=0$, which keeps the rank unchanged. 
In other words, choosing any value $h_{k,i}\in \mathbb{F}_{q}\backslash \{h_{k,i}(L)\}$ will lead to 
\begin{equation} \nonumber
\mathrm{rank} \ \boldsymbol{H}_{[k]}^{L} = \mathrm{rank} \ \boldsymbol{H}_{[k-1]}^{L}+1=|L|,
\end{equation}
which means while $L$ is a circuit set of $\boldsymbol{H}_{[k-1]}$, it will be an independent set of $\boldsymbol{H}_{[k]}$. Note that any field of size $q\geq 2$ guarantees that we can always find a value inside $\mathbb{F}_{q}\backslash \{h_{k}^{L}( \{i\})\}$.

\subsection{Proof of Proposition \ref{prop:L-ind-Li-dep}} \label{app:prop4-MCD}
Since sets $L\backslash \{i\}$ and $L$, respectively, are a basis and a dependent set of $\boldsymbol{H}_{[k-1]}$, column $\boldsymbol{H}_{[k-1]}^{\{i\}}$ is linearly dependent on columns in $\boldsymbol{H}_{[k-1]}^{L\backslash \{i\}}$. So, set $L\backslash \{i\}$ can be partitioned into two subsets $L^{\prime}$ and $L^{\prime\prime}$ such that sets $L^{\prime}\cup \{i\}$ and $L^{\prime\prime}\cup \{i\}$, respectively, will be a circuit and an independent set of $\boldsymbol{H}_{[k-1]}$.
\\
For the circuit set $L^{\prime}\cup \{i\}$, based on Proposition \ref{prop:vet-value}, if $h_{k,i}\in \mathbb{F}_{q}\backslash \{h_{k,i}(L^{\prime}\cup \{i\})\}$, then $L^{\prime}\cup \{i\}$ will be an independent set of $\boldsymbol{H}_{[k]}$.
\\
It is obvious that any independent set of $\boldsymbol{H}_{[k-1]}$ will be an independent set of $\boldsymbol{H}_{[k]}$ as well. Thus, set $L^{\prime\prime}\cup \{i\}$ is also an independent set of $\boldsymbol{H}_{[k]}$.
\\
Thus, by setting $h_{k,i}\in \mathbb{F}_{q}\backslash Z_{k,i}^L$, where $Z_{k,i}^L=\{h_{k,i}(L^{\prime}\cup \{i\})\}$, column $\boldsymbol{H}_{[k]}^{\{i\}}$ will be linearly independent of the columns in $\boldsymbol{H}_{[k]}^{L\backslash \{i\}}$. Now, since $L\backslash \{i\}\in \mathcal{B}_{k-1}$, by setting $h_{k,i}\in \mathbb{F}_{q}\backslash Z_{k,i}^L$, we have $L\in \mathcal{B}_{k}$, which completes the proof.

\section{More Discussion on the Field size and Complexity of the MCD Algorithm} 
\label{app:discuss:field-complexity}
\begin{rem}
According to the UMCD algorithm, the nonzero elements of matrix $\boldsymbol{G}$ are determined by the receivers with the minimum size of side information. This implies that many elements of $\boldsymbol{G}$ can be zero, which may reduce both the computational complexity and required filed size for the MCD algorithm.
\end{rem}

\begin{exmp} \label{exmp:field&complexity-reduction}
Consider the binary matrix $\boldsymbol{G}$ in
\begin{itemize}[leftmargin=*]
    \item Example \ref{exmp:Recursive-5}, equation \eqref{exmp:G-recu}. First, it can be verified that for any $q\geq 2$, the size of the veto set for all the nonzero elements is always one. Thus, fixing the field size as $q\geq 2$ guarantees the existence of at least one element outside of the veto set for each nonzero element. Second, determining the veto value for only two elements $h_{3,2}$ and $h_{3,5}$ requires the $\mathrm{rref}$ operation over two submatrices of size $3\times 2$, implying the low-complexity of the MCD algorithm for this case.
    \item Example \ref{exmp:ICC-5}, equation \eqref{exmp:G-ICC}. First, it can be verified that for any $q\geq 3$, the maximum size of the veto set for the nonzero elements is always two. Thus, fixing the field size as $q\geq 3$ guarantees the existence of at least one element outside of the veto set for each nonzero element. Second, determining the veto value for only two elements $h_{2,1}$ and $h_{2,5}$ requires the $\mathrm{rref}$ operation over two submatrices of size $2\times 2$, implying the low-complexity of the MCD algorithm for this case.
\end{itemize}
\end{exmp}

\begin{rem} \label{rem:square-m}
In the MCD algorithm, the aim is to design the elements of the encoding matrix $\boldsymbol{H}$ such that each of its submatrices achieves its maximum possible rank. However, in the following, we show that, for the UMCD scheme with the binary matrix $\boldsymbol{G}\in \mathbb{F}_{2}^{r\times m}$, to satisfy all the receivers $u_i, i\in [m]$, we only need to make sure that $m$ specific square submatrices of $\boldsymbol{H}\in \mathbb{F}_{q}^{r\times m}$ are full-rank. This can significantly reduce the complexity of the MCD algorithm as well as the required field size.
Let $U_k=\{i\in [m], u_i\ \text{is satisfied at transmission}\ k\}$ denote the set of receivers who are satisfied at transmission $k\leq r$ by the UMCD algorithm. This means that
\begin{equation} \label{eq:rem:dec:con}
    \mathrm{mcm}(\boldsymbol{G}_{[k]}^{\{i\}\cup B_i})=1+\mathrm{mcm}(\boldsymbol{G}_{[k]}^{B_i}), \ \ i\in U_k.
\end{equation}
Now, assume that $\mathrm{mcm}(\boldsymbol{G}_{[k]}^{\{i\}\cup B_i})=l_i$. This implies that there exists a square submatrix $\boldsymbol{G}_{L_i}^{\{i\}\cup B_{i}^{\prime}}$, where $L_i\subseteq [k], |L_i|=l_i$ and $B_{i}^{\prime}\subseteq B_i, |B_{i}^{\prime}|=l_i-1$, such that
\begin{equation} \label{eq:square:mcm}
    \mathrm{mcm}(\boldsymbol{G}_{L_i}^{\{i\}\cup B_i^{\prime}})=1+\mathrm{mcm}(\boldsymbol{G}_{L_i}^{B_i^{\prime}}), \ \ i\in U_k.
\end{equation}
Thus, if the elements of the encoding matrix $\boldsymbol{H}$ are designed such that all the submatrices $\boldsymbol{H}_{L_i}^{\{i\}\cup B_{i}^{\prime}}$ are full-rank (which is possible due to \eqref{eq:square:mcm}), then all the receivers $u_i, i\in U_k$ are able to decode their requested message in the $k$-th transmission.
\end{rem}

\begin{exmp}
As seen in Example \ref{exmp:ICC-5}, in the first transmission, receiver $u_5$ and in the second transmission, other receivers are satisfied. Thus, $U_1=\{5\}$ and $U_2=\{1,2,3,4\}$. It can be verified that for any field size $q\geq 2$, the decoding condition in \eqref{eq:rem:dec:con} will be met by assigning value one to all nonzero elements of the encoding matrix. In fact, in the MCD algorithm, since columns 1 and 5 must be linearly independent, the size of the veto set for one of the elements $h_{2,1}$ and $h_{2,5}$ will be equal to two. Thus, $q\geq 3$ is required (and sufficient as said in Example \ref{exmp:field&complexity-reduction}). However, based on \eqref{eq:rem:dec:con}, the linear independence of columns 1 and 5 is not required (because $1\in A_5$ and $5\in A_1$). Thus, any field size $q\geq 2$ is sufficient for satisfying the decoding condition \eqref{eq:rem:dec:con} for all receivers.
\end{exmp}

\section{Proof of Lemmas \ref{lem:MDS-con-mcm} and \ref{lem:UMCD=MDS}} \label{App-Lemma2,3}
First, we begin with the following definitions and remarks.
\begin{defn} [Potentially Full-rank Binary Matrix]
Let $r\leq m$. We say that binary matrix $\boldsymbol{G}\in \mathbb{F}_{2}^{r\times m}$ is potentially full-rank, if $\mathrm{mcm}(\boldsymbol{G})=r$.
\end{defn}

\begin{defn}[Potential Pivot Column]
Let $i\in [m]$. We say that the $i$-th column of $\boldsymbol{G}\in \mathbb{F}_{2}^{r\times m}$ is a potential pivot column if $\mathrm{mcm}(\boldsymbol{G}^{[m]})=1+\mathrm{mcm}(\boldsymbol{G}^{[m]\backslash \{i\}})$. 
\end{defn}

\begin{rem} \label{rem:mcm-properties}
It can be easily shown that the $\mathrm{mcm}$ function captures some properties of the $\mathrm{rank}$ function such as follows:
\begin{enumerate}[label=(\roman*)]
    \item Assume that $\boldsymbol{G}^{\prime}$ is an $k\times r$ submatrix of the square matrix $\boldsymbol{G}_{r\times r}$, where $k\leq r$, and $\boldsymbol{G}$ is potentially full-rank, i.e., $\mathrm{mcm}(\boldsymbol{G})=r$. Then $\mathrm{mcm}(\boldsymbol{G}^{\prime})=k$. Moreover, every column of $\boldsymbol{G}$ is a potential pivot column. \label{rem:mcm-properties-i}
    \item Let $B\subseteq C$ and $i\not\in C$. If $\mathrm{mcm}(\boldsymbol{G}^{\{i\}\cup C})=1+\mathrm{mcm}(\boldsymbol{G}^{C})$, then $\mathrm{mcm}(\boldsymbol{G}^{\{i\}\cup B})=1+\mathrm{mcm}(\boldsymbol{G}^{B})$. \label{rem:mcm-properties-ii}
\end{enumerate}
\end{rem}

\begin{rem} \label{rem:F-Fprime}
Let set $C_K\subset F_{K}$ be obtained by removing $d-1$ elements from $F_{K}$ in \eqref{eq:MDS-con}. Then, it can be easily verified that (i) $C_K$ satisfies the linear code condition in \eqref{eq:MDS-con} with $d=1$, and (ii) for any set $C_K$ satisfying \eqref{eq:MDS-con} with $d=1$, we have $\mathrm{mcm}(\boldsymbol{G}_{[k]}^{C_{[k]}})=k, \forall k\in [r]$. This means, $\boldsymbol{G}_{[k]}^{C_{[k]}}$ is potentially full-rank.
\end{rem}

\subsection{Proof of Lemma \ref{lem:MDS-con-mcm}} \label{App:lem2}
Let $B_{i}^{\prime}=F_{[k]}\cap B_i$ denote the indices of nonzero columns related to $B_i$. So, 
\begin{equation}
    \mathrm{mcm}(\boldsymbol{G}_{[k]}^{\{i\}\cup B_i})=\mathrm{mcm}(\boldsymbol{G}_{[k]}^{\{i\}\cup B_i^{\prime}}).
\end{equation}
Since $|F_{[k]}|=d-1+k$ and $|A_i|\geq d-1$, then
$|F_{[k]}\backslash A_i|\leq k$, and so, we have
\begin{equation}
    |\{i\}\cup B_{i}^{\prime}|\leq k.
\end{equation}
This means that there exists a set $C_K\subseteq F_{K}$ such that $i\in C_{[k]}$, $B_{i}^{\prime}\subseteq C_{[k]}$, and $|C_{[k]}|=k$. This set $C_K$ can be achieved by removing $d-1$ elements from $F_K$. Based on Remark \ref{rem:F-Fprime}, we have
\begin{equation}
    \mathrm{mcm}(\boldsymbol{G}_{[k]}^{C_{[k]}})=k.
\end{equation}
And because $|C_{[k]}|=k$, then $\boldsymbol{G}_{[k]}^{C_{[k]}}$ is a square matrix, and based on Remark \ref{rem:mcm-properties}-\ref{rem:mcm-properties-i}, each of its columns is a potential pivot column. 
Since $i\in C_{[k]}$, then its corresponding column will be a potential pivot column. Note, since $B_{i}^{\prime}\subseteq C_{[k]}$, based on Remark \ref{rem:mcm-properties}-\ref{rem:mcm-properties-ii}, the corresponding column of the $i$-th element in matrix $\boldsymbol{G}^{\{i\}\cup B_i^{\prime}}$ is a potential pivot column.

\subsection{Proof of Lemma \ref{lem:UMCD=MDS}}\label{App:lem3}
The proof is achieved by induction.
\begin{itemize}[leftmargin=*]
    \item  Consider the condition in \eqref{eq:MDS-con} for $r=1$. Since $|A_i|\geq |A|_{\text{min}}, \forall i\in[m]$, we have $|F_{\{k\}}|=|G_k|=|\{i\}\cup A_i|\geq |A|_{\text{min}}+1$.
    \item As the induction hypothesis, we suppose that the condition in \eqref{eq:MDS-con} holds for $r=k-1$. Thus, \eqref{eq:MDS-con} holds for $K=[k-1]$, which means
    \begin{equation} \label{proof:MDS-1}
        |F_{[k-1]}|\geq |A|_{\text{min}}+ k-1.
    \end{equation}
    \item Now, we need to prove that the condition in \eqref{eq:MDS-con} will also hold for $r=k$. Now, assume that the linear code condition does not hold for $K=[k]$. So,
    \begin{equation} \label{proof:MDS-2}
        |F_{[k]}|\leq |A|_{\text{min}}+ k-1.
    \end{equation}
    Based on \eqref{proof:MDS-1} and \eqref{proof:MDS-2}, we have
    \begin{equation} \label{proof:MDS-3}
        F_{[k]}=F_{[k-1]}=|A|_{\text{min}}+ k-1.
    \end{equation}
    Let $w$ denote the index of receiver $u_{w}$, which is selected by the UMCD coding scheme for the $k$-th transmission. Now, from \eqref{proof:MDS-3}, we must have 
    \begin{equation} \label{proof:MDS-4}
        G_k=\{w\}\cup A_{w}\subseteq F_{[k-1]}.
    \end{equation} 
    Now, since $|A_w|\geq |A|_{\text{min}}=d-1$, the three conditions of Lemma  \ref{lem:MDS-con-mcm} are met for $r=k-1$ and $d=|A|_{\text{min}}+1$. So, we will have
    \begin{equation}
         \mathrm{mcm}(\boldsymbol{G}_{[k-1]}^{\{w\}\cup B_w})=1+\mathrm{mcm}(\boldsymbol{G}_{[k-1]}^{B_w}).
    \end{equation}
    Hence, receiver $u_{w}$ is able to decode its requested messages from the first $k-1$ transmissions. However, this contradicts the UMCD coding scheme's logic, where at each step of transmission, it picks a receiver which has not been satisfied by the previous transmissions (as it can be seen in Algorithm \ref{alg:UMCD} that in each transmission, we remove $i$ from $N$ as $N\leftarrow N\backslash\{i\}$ when $u_i$ is satisfied). This completes the proof.
\end{itemize}

\bibliographystyle{IEEEtran}
	\bibliography{References1} 
	
\vskip -2\baselineskip plus -12fil

\begin{IEEEbiography}[{\includegraphics[width=1in,height=1.25in,clip,keepaspectratio]{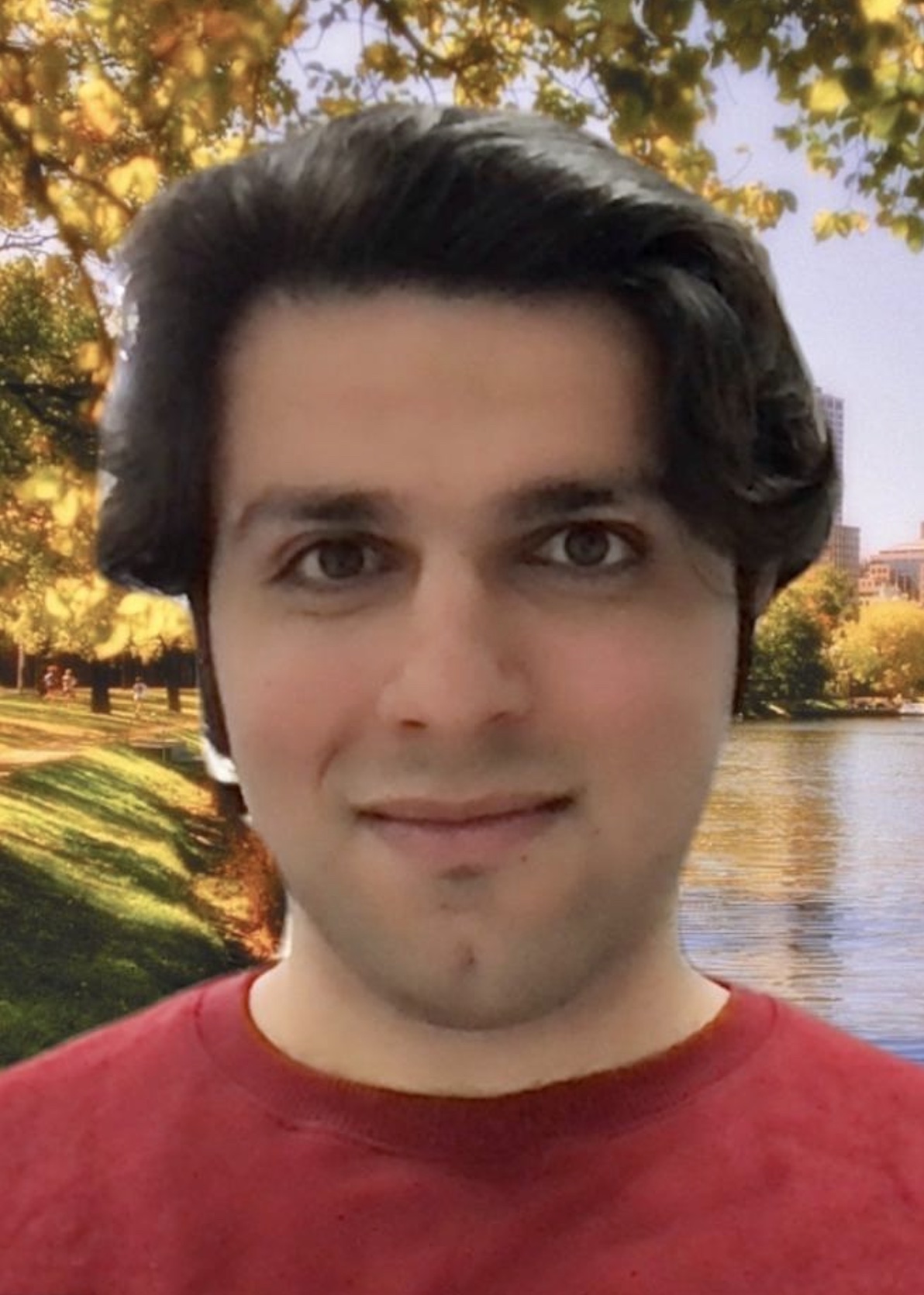}}]{Arman Sharififar}
received his B.Sc. degree in Electrical Engineering from Bahonar University, Iran. He completed his M.Sc. degree in the field of coding and communication systems at Shiraz University, Iran. Currently, he is pursuing his PhD degree at the School of Engineering and Information Technology, the University of New South Wales, Canberra, Australia. His research interests include index and network coding, private and secured index coding, coded caching, and space-time coding in the MIMO systems.
\end{IEEEbiography}

\vskip -2\baselineskip plus -1fil

\begin{IEEEbiography}[{\includegraphics[width=1in,height=1.25in,clip,keepaspectratio]{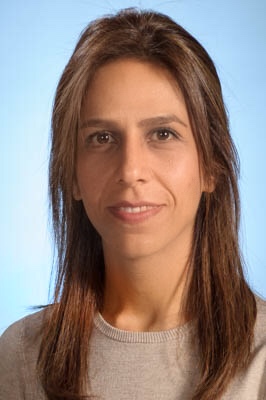}}]{Neda Aboutorab}
(S’09-M’12-SM’17) is currently a Senior Lecturer at the School of Engineering and Information Technology at the University of New South Wales, Canberra, Australia. She received her PhD in Electrical Engineering from the University of Sydney, Australia, in 2012. From 2012-2015 and before joining the University of New South Wales, she was a Postdoctoral Research Fellow at the Research School of Engineering, the Australian National University. Her research interests include index and network coding, applied information theory, big data caching and storage systems, wireless communications and signal processing.
\end{IEEEbiography}

\vskip -2\baselineskip plus -1fil

\begin{IEEEbiography}[{\includegraphics[width=1in,height=1.25in,clip,keepaspectratio]{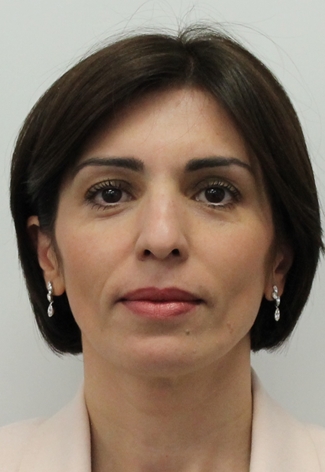}}]{Parastoo Sadeghi}
(Senior Member, IEEE) received the bachelor’s and master’s degrees in electrical engineering from the Sharif University of Technology, Tehran, Iran, in 1995 and 1997, respectively, and the Ph.D. degree in electrical engineering from the University of New South Wales, Sydney, NSW, Australia, in 2006. She is currently a Professor with the School of Engineering and Information Technology, University of New South Wales, Canberra, ACT, Australia. She has co-authored the book Hilbert Space Methods in Signal Processing (Cambridge University Press, 2013) and around 190 refereed journal articles and conference papers. Her research interests include information theory, data privacy, index coding, and network coding. She was a recipient of the 2019 Future Fellowship from the Australian Research Council. From 2016 to 2019, she served as an Associate Editor for the IEEE
TRANSACTIONS ON INFORMATION THEORY. From 2019 to 2020, she served as a member on the Board of Governors of the IEEE Information Theory Society. She was the General Co-chair of the 2021 IEEE International Symposium on Information Theory.

\end{IEEEbiography}

\pagebreak

\section{Proof of Proposition \ref{prop:A-2}}   \label{proof:A-2}
The proof of Proposition \ref{prop:A-2} can be directly concluded from the following Lemmas \ref{lem1:prop4} and \ref{lem2:prop4}.  In Lemma \ref{lem1:prop4}, we prove that $\beta_{\text{FPCC}}(\mathcal{I}_{6}(l))=3l+\frac{1}{2}$. Then, in Lemma \ref{lem2:prop4}, it will be shown that the recursive scheme cannot outperform the FPCC scheme for the class-$\mathcal{I}_{6}$ index coding instances.
The associated graph $\mathcal{G}_{\mathcal{I}_{6}}$ is depicted in Figure \ref{fig:class-A}.

\begin{lem} \label{lem1:prop4}
For the class-$\mathcal{I}_{6}$ index coding instances, we have $\beta_{\text{FPCC}}(\mathcal{I}_{6}(l))=3l+\frac{1}{2}$.
\end{lem}

\begin{proof}
Regarding Proposition \ref{prop:minimal}, we only consider the minimal partial clique sets in class-$\mathcal{I}_{6}$ index coding instances. It can be verified that the minimal partial clique sets can be all characterized as follows
\begin{equation}
   \left\{\begin{array}{lc}
     M_{i,j}&=\{2i-1, 2i, 2j-1, 2j\}, \ \ \ \forall i\neq j\in[2l],  \label{eq:minimalpartialclique-class-A}
     \\ \\
     M_{i,2l+1}&=\{2i-1, 4l+1\}, \ \ \ \ \ \ \ \ \ \ \   \forall i\in[2l], \ \ \ \ \ \ \ 
     \\ \\
     M_{i}&=\{i\}, \ \ \  \ \ \ \ \ \ \ \ \ \ \  \ \ \ \ \ \ \ \ \ \ \ \ \ \forall i\in[4l+1].
    \end{array}
    \right.
\end{equation}
Moreover, it can also be easily seen that
\begin{equation} \nonumber
   \left\{\begin{array}{lc}
     \beta_{\text{MDS}}(M_{i,j})&=3, \ \ \ \forall i\neq j\in[2l],  \nonumber\\ \\
     \beta_{\text{MDS}}(M_{i,2l+1})&=1, \ \ \  \forall i\in[2l], \ \ \ \ \ \ \ \nonumber\\ \\
     \beta_{\text{MDS}}(M_{i})&=1,  \ \ \ \forall i\in[4l+1]. \nonumber
    \end{array}
    \right.
\end{equation}
Thus, the optimal solution for the FPCC scheme in \eqref{eq:FPCC-opt} will be equal to (note, since $M_{i,j}=M_{j,i}, \forall i\neq j\in[2l]$, to avoid repetition, we only consider $i<j\in[2l]$)
\begin{align}
    &\beta_{\text{FPCC}}(\mathcal{I}_{6}(l))= \min_{\substack{\gamma_{i,j}, \ \ i<j\in [2l] \\ \gamma_{i,2l+1}, \ \ i\in[2l] \\ \gamma_{i},\ \ i\in[4l+1]}} 
   \Big ( \sum_{i<j \in[2l]} \gamma_{i,j} \ \beta_{\text{MDS}}(M_{i,j}) 
   \nonumber
   \\
    & 
    +\sum_{i\in[2l]} \gamma_{i,2l+1} \ \beta_{\text{MDS}}(M_{i,2l+1})
    + \sum_{i\in[4l+1]} \gamma_{i} \ \beta_{\text{MDS}}(M_i) \nonumber
   \ \Big )\\
   =& 
    \min_{\substack{\gamma_{i,j}, \ \ i<j\in [2l] \\ \gamma_{i,2l+1}, \ \ i\in[2l] \\ \gamma_{i},\ \ i\in[4l+1]}}  
   \Big ( \sum_{i<j \in[2l]} 3\gamma_{i,j} +\sum_{i\in[2l]} \gamma_{i,2l+1}  + \sum_{i\in[4l+1]} \gamma_{i} 
   \ \Big ), \label{eq:pr:prop4:1}
\end{align}
subject to the constraints in \eqref{eq:FPCC:cond1}, which will be expressed as follows (according to Proposition \ref{prop: FPCC-modification}, we just consider the equality case and also we let $\gamma_{i,j}=\gamma_{j,i}, i<j\in[2l]$)
\begin{align}
        &\Big (\sum\limits_{\substack{j\in[2l]\\ j\neq i}} \gamma_{\frac{i+1}{2},j}\Big ) + \gamma_{i,2l+1}+\gamma_i=1, \ \forall i=2i_1-1, i_1\in[2l],  
        \label{eq:pr:prop4:con1} 
        \\ 
        &\Big (\sum\limits_{\substack{j\in[2l]\\ j\neq i}} \gamma_{\frac{i}{2},j}\Big ) +\gamma_i=1, \ \ \ \ \ \ \ \ \ \ \ \ \ \forall i=2i_1, i_1\in[2l], \label{eq:pr:prop4:con2}
        \\
        &\Big (\sum\limits_{i\in[2l]} \gamma_{i,2l+1} \Big ) + \gamma_{4l+1}=1, \ \ \ \ \ \ \ \ \ \ \ \ i=4l+1. \label{eq:pr:prop4:con3}
\end{align}
Now, we add all the constraints in \eqref{eq:pr:prop4:con1}, \eqref{eq:pr:prop4:con2} and \eqref{eq:pr:prop4:con3} with each other, which will result in
\begin{equation} \label{eq:pr:prop4:2}
    \sum_{i<j\in[2l]} 4\gamma_{i,j} + \sum_{i\in[2l]} 2\gamma_{i,2l+1} + \sum_{i\in[4l+1]} \gamma_{i} =4l+1.
\end{equation}
Now, using \eqref{eq:pr:prop4:2}, we can express \eqref{eq:pr:prop4:1} as follows
\begin{align} 
    &\sum_{i<j \in[2l]} 3\gamma_{i,j} +\sum_{i\in[2l]} \gamma_{i,2l+1}  + \sum_{i\in[4l+1]} \gamma_{i} =
    4l+1
    \nonumber
    \\
    &
    -\Big (\sum_{i<j\in[2l]} \gamma_{i,j} + \sum_{i\in[2l]} \gamma_{i,2l+1} \Big ), \label{eq:pr:prop4:result-1}
\end{align}
Moreover, adding all the constraints in \eqref{eq:FPCC:cond1} and \eqref{eq:pr:prop4:con3} will result in
\begin{align} 
    \sum_{i<j\in[2l]} \gamma_{i,j} + \sum_{i\in[2l]} \gamma_{i,2l+1} + \sum_{i\in [2l+1]} \gamma_{2i-1}= \frac{2l+1}{2}, \label{eq:pr:prop4:result-2}
\end{align}
From \eqref{eq:pr:prop4:result-1} and \eqref{eq:pr:prop4:result-2}, we have
\begin{align}
    \beta_{\text{FPCC}}(\mathcal{I}_{6}(l))&= 3l+\frac{1}{2}-\min_{\gamma_{2i-1}, i\in[2l+1]} \sum_{i\in [2l+1]} \gamma_{2i-1} \nonumber
    \\
    &=3l+\frac{1}{2}, \label{eq:pr:prop4:final}
\end{align}
where \eqref{eq:pr:prop4:final} is achieved by setting  $\gamma_{2i-1}=0, \forall i\in [2l+1]$. This completes the proof.

\end{proof}

\begin{figure}
    \centering
     \begin{tikzpicture}
                 \tikzset{vertex/.style = {shape=circle,draw,minimum size=1em}}
                 \tikzset{edge/.style = {->,> = latex'}}
                 \node[vertex,inner sep=2pt, minimum size=0.5pt] (1)    at  (0,2)   {\small 1};
                 \node[vertex,inner sep=2pt, minimum size=0.5pt] (2)    at  (0,0)     {\small 2};
                 \node[vertex,inner sep=2pt, minimum size=0.5pt] (3)    at  (1.5,2) {\small 3};
                 \node[vertex,inner sep=2pt, minimum size=0.5pt] (4)    at  (1.5,0)   {\small 4};
                 \node[vertex,inner sep=2pt, minimum size=0.5pt] (4l-1) at  (5,2)   {\scriptsize 4$l$-1};
                 \node[vertex,inner sep=2pt, minimum size=0.5pt] (4l)   at  (5,0)     {\scriptsize 4$l$};
                 \node[vertex,inner sep=2pt, minimum size=0.5pt] (4l+1) at  (2.5,4)   {\scriptsize 4$l$+1};
                 
                 \draw[edge] (1)       [dashed]    to          (2);
                 \draw[edge] (3)       [dashed]    to          (4);
                 \draw[edge] (4l-1)    [dashed]    to          (4l);

                 \draw[edge] (2)     [dashed]    to            (3);
                 \draw[edge] (2)     [dashed]    to            (4l-1);
                 
                 \draw[edge] (4)     [dashed]    to            (1);
                 \draw[edge] (4)     [dashed]    to            (4l-1);
                 
                 \draw[edge] (4l)    [dashed]    to            (3);
                 \draw[edge] (4l)    [dashed]    to            (1);

                 \draw[edge] (1)         to            (4l+1);
                 \draw[edge] (3)         to            (4l+1);
                 \draw[edge] (4l-1)      to            (4l+1);
                 \draw[edge] (4l+1)      to            (1);
                 \draw[edge] (4l+1)      to            (3);
                 \draw[edge] (4l+1)      to            (4l-1);

                 \path (3)   [dashed]    to            node {\dots} (4l-1);
                 \path (4)   [dashed]    to            node {\dots} (4l);
            \end{tikzpicture} 
    \caption{The class-$\mathcal{I}_{6}$ index coding instances for which $\beta_{\text{R}}(\mathcal{I}_{6}(l))-\beta_{\text{UMCD}}(\mathcal{I}_{6}(l))=l-\frac{1}{2}$.}
    \label{fig:class-A}
\end{figure}
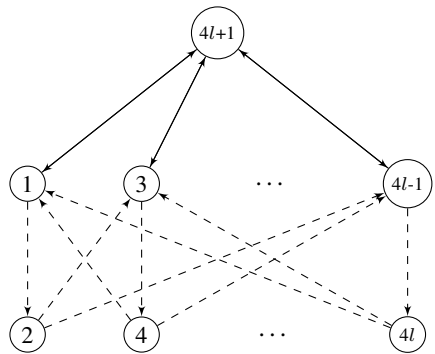
\begin{lem} \label{lem2:prop4}
For the class-$\mathcal{I}_{6}$ index coding instances, we have $\beta_{\text{R}}(\mathcal{I}_{6}(l))=\beta_{\text{FPCC}}(\mathcal{I}_{6}(l))$.
\end{lem}

\begin{proof}
First, it can be seen that the subsets in \eqref{eq:minimalpartialclique-class-A} are all minimal recursive sets for which the broadcast rate of the recursive and the MDS coding schemes are equal. Second, since none of the sets $M_{i,j}, i\neq j\in [2l]$ and $M_{i,2l+1}, i\in [2l]$ in \eqref{eq:minimalpartialclique-class-A} are subset of any side information set $A_i, i\in [m]$, then due to \eqref{eq:recursive:broadcastrate}, the broadcast rate of the recursive coding scheme will be equal to the time sharing over the subsets in \eqref{eq:minimalpartialclique-class-A}. Therefore, the recursive optimization problem will be reduced to the FPCC optimization problem, which completes the proof.
\end{proof}

\section{Proof of Theorem \ref{thm:class-ICC-UMCD}}   \label{proof:thm-class-ICC}

The proof of Theorem \ref{thm:class-ICC-UMCD} can directly be concluded from Propositions \ref{prop:B-1} and \ref{prop:B-2}. 

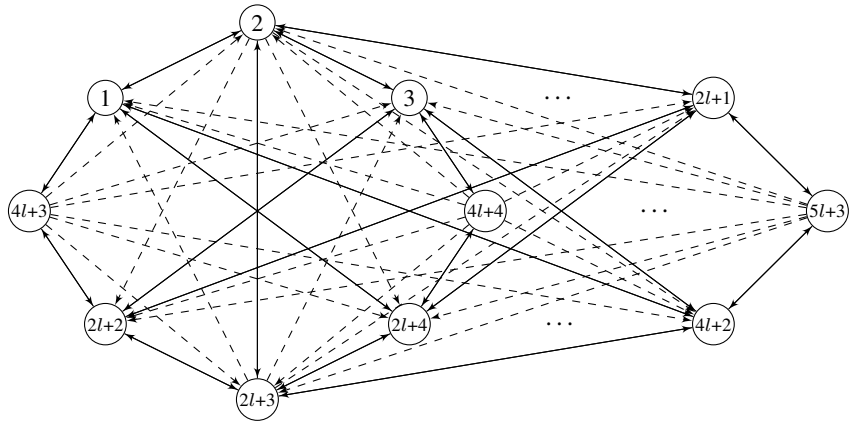
\begin{figure*}
    \centering
    \begin{tikzpicture}
                 \tikzset{vertex/.style = {shape=circle,draw,minimum size=1em}}
                 \tikzset{edge/.style = {->,> = latex'}}

                 \node[vertex,inner sep=2pt, minimum size=0.5pt] (1)       at  (0,3)      {\small 1};
                 \node[vertex,inner sep=2pt, minimum size=0.5pt] (2)       at  (2,4)      {\small 2};
                 \node[vertex,inner sep=2pt, minimum size=0.5pt] (3)       at  (4,3)      {\small 3};
                 \node[vertex,inner sep=.5pt, minimum size=0.5pt] (2k+1)    at  (8,3)      {\scriptsize 2$l$+1};
                 
                 \node[vertex,inner sep=.5pt, minimum size=0.5pt] (2k+2)    at  (0,0)      {\scriptsize 2$l$+2};
                 \node[vertex,inner sep=.5pt, minimum size=0.5pt] (2k+3)    at  (2,-1)      {\scriptsize 2$l$+3};
                 \node[vertex,inner sep=.5pt, minimum size=0.5pt] (2k+4)    at  (4,0)      {\scriptsize 2$l$+4};
                 \node[vertex,inner sep=.5pt, minimum size=0.5pt] (4k+2)    at  (8,0)      {\scriptsize 4$l$+2};
                 
                 \node[vertex,inner sep=.5pt, minimum size=0.5pt] (4k+3)    at  (-1,1.5)   {\scriptsize 4$l$+3};
                 \node[vertex,inner sep=.5pt, minimum size=0.5pt] (4k+4)    at  (5,1.5)    {\scriptsize 4$l$+4};
                 \node[vertex,inner sep=.5pt, minimum size=0.5pt] (5k+3)    at  (9.5,1.5)    {\scriptsize 5$l$+3};
                 

                 \draw[edge] (1)           to          (2);
                 \draw[edge] (2)           to          (1);
                 \draw[edge] (2)           to          (3);
                 \draw[edge] (3)           to          (2);
                 
                 \draw[edge] (2k+2)        to          (2k+3);
                 \draw[edge] (2k+3)        to          (2k+2);
                 \draw[edge] (2k+3)        to          (2k+4);
                 \draw[edge] (2k+4)        to          (2k+3);
                 
                 \draw[edge] (2)           to          (2k+1);
                 \draw[edge] (2k+1)        to          (2);
                 \draw[edge] (2k+3)        to          (4k+2);
                 \draw[edge] (4k+2)        to          (2k+3);
                 
                 \draw[edge] (1)           to          (2k+4);
                 \draw[edge] (2k+4)        to          (1);
                 \draw[edge] (3)           to          (2k+2);
                 \draw[edge] (2k+2)        to          (3);
                 
                 \draw[edge] (1)           to          (4k+2);
                 \draw[edge] (4k+2)        to          (1);
                 \draw[edge] (3)           to          (4k+2);
                 \draw[edge] (4k+2)        to          (3);
                 
                 \draw[edge] (2k+2)        to          (2k+1);
                 \draw[edge] (2k+1)        to          (2k+2);
                 \draw[edge] (2k+4)        to          (2k+1);
                 \draw[edge] (2k+1)        to          (2k+4);
                 
                 \draw[edge] (1)           to          (4k+3);
                 \draw[edge] (2k+2)        to          (4k+3);
                 \draw[edge] (4k+3)        to          (1);
                 \draw[edge] (4k+3)        to          (2k+2);
                 
                 \draw[edge] (3)           to          (4k+4);
                 \draw[edge] (2k+4)        to          (4k+4);
                 \draw[edge] (4k+4)        to          (3);
                 \draw[edge] (4k+4)        to          (2k+4);
                 
                 \draw[edge] (2k+1)        to          (5k+3);
                 \draw[edge] (4k+2)        to          (5k+3);
                 \draw[edge] (5k+3)        to          (2k+1);
                 \draw[edge] (5k+3)        to          (4k+2);
                 
                 \draw[edge] (2)           to          (2k+3);
                 \draw[edge] (2k+3)        to          (2);
               
                 \draw[edge] (2)     [dashed]       to          (2k+2);
                 \draw[edge] (2k+3)  [dashed]       to          (1);
                 \draw[edge] (2)     [dashed]       to          (2k+4);
                 \draw[edge] (2k+3)  [dashed]       to          (3);
                 \draw[edge] (2)     [dashed]       to          (4k+2);
                 \draw[edge] (2k+3)  [dashed]       to          (2k+1);
                 
                  \draw[edge] (4k+3)  [dashed]       to          (2);
                  \draw[edge] (4k+3)  [dashed]       to          (3);
                  \draw[edge] (4k+3)  [dashed]       to          (2k+3);
                  \draw[edge] (4k+3)  [dashed]       to          (2k+4);
                  \draw[edge] (4k+3)  [dashed]       to          (2k+1);
                  \draw[edge] (4k+3)  [dashed]       to          (4k+2);
                  
                  \draw[edge] (4k+4)  [dashed]       to          (2);
                  \draw[edge] (4k+4)  [dashed]       to          (1);
                  \draw[edge] (4k+4)  [dashed]       to          (2k+3);
                  \draw[edge] (4k+4)  [dashed]       to          (2k+2);
                  \draw[edge] (4k+4)  [dashed]       to          (2k+1);
                  \draw[edge] (4k+4)  [dashed]       to          (4k+2);
                  
                  \draw[edge] (5k+3)  [dashed]       to          (2);
                  \draw[edge] (5k+3)  [dashed]       to          (1);
                  \draw[edge] (5k+3)  [dashed]       to          (3);
                  \draw[edge] (5k+3)  [dashed]       to          (2k+2);
                  \draw[edge] (5k+3)  [dashed]       to          (2k+3);
                  \draw[edge] (5k+3)  [dashed]       to          (2k+4);

                 \path (3)                 to         node {\dots} (2k+1);
                 \path (2k+4)              to         node {\dots} (4k+2);
                 
                 \path (4k+4)              to         node {\dots} (5k+3);
                 
\end{tikzpicture}  
    \caption{The class-$\mathcal{I}_{7}$ index coding instances for which $\beta_{\text{ICC}}(\mathcal{I}_{7}(l))-\beta_{\text{UMCD}}(\mathcal{I}_{7}(l))\geq 2l-\frac{1}{2}$.}
  \label{fig:class-B-s}
\end{figure*}

\begin{prop} \label{prop:B-1}
The broadcast rate of the proposed UMCD coding scheme for the class-$\mathcal{I}_{7}$ index coding instances is 
\begin{equation} \nonumber
    \beta_{\text{UMCD}}(\mathcal{I}_{7}(l))=l+2.
\end{equation}
\end{prop}
\begin{proof}
First, without loss of generality, we rearrange the elements of the message vector as follows
\begin{equation}
    \boldsymbol{x}=[\boldsymbol{o}_{l_1}\ |\ \boldsymbol{e}_{l_2}\ |\ \boldsymbol{l}_3\ |\ \boldsymbol{e}_{l_1}\ |\ \boldsymbol{o}_{l_2}],
\end{equation}
where $\boldsymbol{o}_{l_1}$, $\boldsymbol{e}_{l_2}$, $\boldsymbol{l}_3$, $\boldsymbol{e}_{l_1}$, $\boldsymbol{o}_{l_2}$ represent the vectors whose elements are selected, respectively, from the sets $O_{L_1}$, $E_{L_2}$, $L_3$, $E_{L_1}$ and $O_{L_2}$ in an ascending order. Now, note that $|A_i|=4l+2, \ \forall i\in L_1\cup L_2,$ and $|A_{i}|=2, \ \forall i\in L_3$. Therefore, the UMCD coding scheme begins with the receivers indexed by $W=L_3$, having the minimum size of side information. By choosing $w=i+4l+1$, in the $i$-th transmission, $i\in [l+1]$, it can be seen that the binary matrix for the first $l+1$ transmissions will be equal to
\begin{equation} \nonumber
  \boldsymbol{G}_{[l+1]}=
      \left [\begin{array}{c|c|c|c|c}
        \boldsymbol{I}_{l+1}   &  \boldsymbol{I}_{l+1} & \boldsymbol{I}_{l+1} &  \boldsymbol{0}_{(l+1)\times l} & \boldsymbol{0}_{(l+1)\times l}
    \end{array}
    \right ],
\end{equation}
which satisfies all receivers $u_{i}, \forall i\in L_3$. Since the size of the side information of the remaining receivers is equal, one of them is selected randomly. Assume that $w=j$ for some $j\in O_{L_1}$. Then, the binary matrix $\boldsymbol{G}_{[l+2]}$ will be
 \begin{equation}
   \left [ 
      \begin{array}{c|c|c|c|c}
        \boldsymbol{I}_{l+1}   &  \boldsymbol{I}_{l+1} & \boldsymbol{I}_{l+1} &  \boldsymbol{0}_{(l+1)\times l} & \boldsymbol{0}_{(l+1)\times l} \\
       \hline
       \boldsymbol{0}_{l+1}(\frac{j+1}{2}) &  \boldsymbol{1}_{l+1}(\frac{j+1}{2}) & \boldsymbol{1}_{l+1} & \boldsymbol{1}_{l} & \boldsymbol{1}_{l}
    \end{array}
   \right ],
  \end{equation}
where, $\boldsymbol{1}_{l}$ and $\boldsymbol{0}_{l}$, respectively, denote a row vector of size $l$ whose elements are all set to one and zero. Moreover, $\boldsymbol{1}_{l}(j)$ and $\boldsymbol{0}_{l}(j)$, respectively, represent a row vector of size $l$ whose elements are all set to one and zero, except the $j$-th element. 
Now, we show that all the remaining unsatisfied receivers will be able to decode their requested message using $\boldsymbol{G}_{[l+2]}$. It can be verified that binary matrix $\boldsymbol{G}_{[l+2]}^{\{i\}\cup B_i}$ for the remaining receivers is as follows

\begin{equation} \nonumber
   \left\{\begin{array}{lc}

   \left [\begin{array}{c|c}
        \boldsymbol{I}_{l+1}   &  \boldsymbol{0}_{l+1}^{T}(\frac{i+1}{2}) \\
       \hline
       \boldsymbol{0}_{l+1}(\frac{j+1}{2}) &  1
    \end{array}
    \right ], \ \ \ \ \  \ \ \ \forall i\in O_{L_1}\backslash\{j\},
    \\
    \\
   \left [\begin{array}{c|c}
        \boldsymbol{I}_{l+1}   &  \boldsymbol{0}_{l+1}^{T}\\
       \hline
       \boldsymbol{1}_{l+1}(\frac{j+1}{2}) &  1
    \end{array}
    \right ], \ \ \ \ \ \ \ \ \ \ \ \ \ \ \ \forall i\in E_{L_1},
    \\ 
    \\
   \left [\begin{array}{c|c}
        \boldsymbol{0}_{l+1}^{T}(\frac{i-2k}{2})   &  \boldsymbol{I}_{l+1}\\
       \hline
       0 &  \boldsymbol{1}_{l+1}(\frac{j+1}{2})
    \end{array}
    \right ], \ \ \forall i\in E_{L_2}\backslash\{j+2k+1\}, \nonumber 
    \\
    \\
    \left [\begin{array}{c|c}
        \boldsymbol{0}_{l+1}^{T}(\frac{i-2k}{2})   &  \boldsymbol{I}_{l+1}\\
       \hline
       1 &  \boldsymbol{1}_{l+1}(\frac{j+1}{2})
    \end{array}
    \right ], \ \ \  i=j+2k+1\in E_{L_2}, \nonumber 
    \\ 
    \\
  \left [\begin{array}{c|c}
        \boldsymbol{I}_{l+1}   &  \boldsymbol{0}_{l+1}^{T} \\
       \hline
       \boldsymbol{0}_{l+1}(\frac{j+1}{2}) &  1
    \end{array}
    \right ], \ \ \ \ \ \ \ \ \  \ \ \ \ \ \ \forall i\in O_{L_2}. \nonumber
    \end{array}
    \right. 
\end{equation}
Now, it can be easily observed that for each $\boldsymbol{G}_{[l+2]}^{\{i\}\cup B_i}$, the elements of the main diagonal can be all fixed to one. Thus, the decoding condition \eqref{eq:dec-mcm} will hold for all the receivers, which completes the proof. One can check that for the other possible selection of $w$ in transmission $l+2$, the broadcast rate will be the same.
\end{proof}

\begin{prop} \label{prop:B-2}
For the broadcast rate of the ICC coding scheme for the class-$\mathcal{I}_{7}$ index coding instances, we have $\beta_{\text{ICC}}(\mathcal{I}_{7}(l))\geq \frac{3(l+1)}{2}$.

\end{prop}
The proof of Proposition \ref{prop:B-2} can be directly concluded from the following Lemmas \ref{lem1:prop7} and \ref{lem2:prop7}.  In Lemma \ref{lem1:prop7}, we prove that $\beta_{\text{FPCC}}(\mathcal{I}_{7}(l))\geq\frac{3(l+1)}{2}$. Then, in Lemma \ref{lem2:prop7}, it will be shown that the ICC scheme cannot outperform the FPCC scheme for the class-$\mathcal{I}_{7}$ index coding instances.

\begin{lem} \label{lem1:prop7}
For the class-$\mathcal{I}_{7}$ index coding instance, we have $\beta_{\text{FPCC}}(\mathcal{I}_{7}(l))\geq\frac{3(l+1)}{2}$.
\end{lem}

\begin{proof}
First, we note that any subset $M_j\subseteq L_1\cup L_2, j\in[n]$ can be expressed as follows
\begin{equation} \label{eq:class-B:M_j}
    M_j=O_{L_1,j}\cup E_{L_1,j} \cup E_{L_2,j} \cup O_{L_2,j},
\end{equation}
where, $O_{L_i,j}\subseteq O_{L_i}$ and $E_{L_i,j}\subseteq E_{L_i}$ for $i=1,2$ and $\forall j\in \mathcal{P}([5l+3])$. Now, based on \eqref{eq:class-B:M_j} and \eqref{eq:Class-B:side}, it can be seen that the local side information $M_j\cap A_i$ will be equal to
\begin{equation} \nonumber
   \left\{\begin{array}{lc}
   
       E_{L_1,j}\cup O_{L_2,j} \cup E_{L_2,j}\backslash \{i+(2l+1)\},  \ \ \ \forall i\in O_{L_1,j}, \nonumber
       \\
       \\
       E_{L_1,j}\backslash \{i\} \cup O_{L_2,j} \cup O_{L_1,j},  \ \ \ \ \ \ \ \ \ \ \ \ \ \ \ \ \forall i\in E_{L_1,j}, \nonumber
       \\
       \\
       E_{L_1,j} \cup O_{L_2,j} \cup O_{L_1,j}\backslash \{i-(2l+1)\},  \ \ \ \forall i\in E_{L_2,j}, \nonumber
       \\
       \\
       E_{L_1,j} \cup O_{L_2,j}\backslash \{i\} \cup E_{L_2,j},  \ \ \ \ \ \ \ \ \ \ \ \ \ \ \ \ \  \forall i\in O_{L_2,j}. \nonumber
    \end{array}
    \right.
\end{equation}
So, $|M_j\cap A_i|$ will be equal to
\begin{align} \nonumber
       &e_{L_1,j}+ o_{L_2,j}+ e_{L_2,j}-|E_{L_2,j}\cap \{i+(2l+1)\}|, \forall i\in O_{L_1,j}, 
       \nonumber
       \\
       &
       e_{L_1,j}+ o_{L_2,j}+ o_{L_1,j}-|E_{L_1,j}\cap \{i\}|,  \ \ \ \ \ \ \ \ \ \ \ \ \forall i\in E_{L_1,j}, 
       \nonumber
       \\
       &
       e_{L_1,j}+ o_{L_2,j}+ o_{L_1,j}- |O_{L_1,j}\cap \{i-(2l+1)\}|,\forall i\in E_{L_2,j}, 
       \nonumber
       \\
       &
       e_{L_1,j}+ o_{L_2,j}+ e_{L_2,j}-|O_{L_2,j}\backslash \{i\}|,\ \ \ \ \ \ \ \ \ \ \ \ \ \  \forall i\in O_{L_2,j}, \nonumber
\end{align}
   
       
where $e_{L_i,j}=|e_{L_i,j}|$, $o_{L_i,j}=|O_{L_i,j}|$ for $i=1,2$.
Now, for the minimum size of the local side information, we have
\begin{align}
    \min_{i\in M_j} &|M_j\cap A_i|= e_{L_1,j}+o_{L_2,j}+ \min \Big \{ e_{L_2,j}
    \nonumber
    \\
    &
    -|E_{L_2,j}\cap \{i+(2l+1)\}|, \ o_{L_1,j}-1, 
    \nonumber
    \\
    &
    \ o_{L_1,j}-|O_{L_1,j}\cap \{i-(2l+1)\}|, \ e_{L_2,j}-1\Big \} 
    \nonumber
    \\
    =
    &\ e_{L_1,j}+o_{L_2,j}-1+\min \Big \{o_{L_1,j}, e_{L_2,j}\Big \}. \label{eq:pr:class-B-1}
\end{align}
Now, for the broadcast rate of the MDS code \eqref{eq:MDS-local} for each subset $M_j$ in \eqref{eq:class-B:M_j}, we have
\begin{align}
    \beta_{\text{MDS}}(M_j)&=|M_j|- \min_{i\in M_j} |M_j\cap A_i| \nonumber
    \\
    & = e_{L_2,j}+o_{L_1,j}+1-\min \Big \{o_{L_1,j}, e_{L_2,j}\Big \} \nonumber
    \\
    & \geq e_{L_2,j}+o_{L_1,j}+1-\frac{e_{L_2,j}+o_{L_1,j}}{2} \nonumber
    \\
    &= \frac{e_{L_2,j}+o_{L_1,j}}{2}+1.
\end{align}
It can also be verified that all the minimal partial clique sets which include receivers indexed by $L_3$, are characterized as follows
\begin{equation}
 \left\{\begin{array}{lc}
        M_{j}^{\prime}=\{2j-1, 4l+2+j\},\ \ \ \ \ \ \forall j\in[l+1], \nonumber
       \\
       \\
        M_{j}^{\prime \prime}=\{2l+2j, 4l+2+j\},\ \  \ \ \forall j\in[l+1], \nonumber
       \\
       \\
        M_{j}^{\prime \prime \prime}=\{4l+2+j\}, \ \ \ \ \ \ \ \ \ \ \ \ \ \ \forall j\in[l+1],\nonumber
    \end{array}
    \right.    
\end{equation}
where their MDS broadcast rate is as below
\begin{equation} \label{eq:pr:class-B-rate-Mprime}
    \beta_{\text{MDS}}(M_{j}^{\prime})=\beta_{\text{MDS}}(M_{j}^{\prime\prime})=\beta_{\text{MDS}}(M_{j}^{\prime\prime\prime})=1, \forall j\in[l+1].
\end{equation}
Now, having \eqref{eq:pr:class-B-1} and  \eqref{eq:pr:class-B-rate-Mprime}, the broadcast rate of the FPCC scheme is equal to

\vspace{-4ex}

\begin{align}
    &\beta_{\text{FPCC}}(\mathcal{I}_{7}(l))= \min_{\substack{\gamma_j, M_j, j\in[n] \\ M_{j}^{\prime}, M_{j}^{\prime\prime}, M_{j}^{\prime\prime\prime}, \\ \gamma_{j}^{\prime}, \gamma_{j}^{\prime\prime}, \gamma_{j}^{\prime\prime\prime}, j\in [l+1]}}\Bigg [ \sum_{j\in[n]} \gamma_{j}\beta_{\text{MDS}}(M_j) 
    \nonumber
    \\
    &
    + 
    \sum_{j\in[l+1]} \Big ( \gamma_{j}^{\prime}\beta_{\text{MDS}}(M_j^{\prime})
    + \gamma_{j}^{\prime\prime}\beta_{\text{MDS}}(M_j^{\prime\prime}) 
    + \gamma_{j}^{\prime\prime\prime}\beta_{\text{MDS}}(M_j^{\prime\prime\prime}) \Big )  \Bigg ]
    \nonumber
    \\
    &\geq \min_{\substack{\gamma_j, e_{L_2,j}, o_{L_1,j}\\ j\in[n], \\ \gamma_{j}^{\prime}, \gamma_{j}^{\prime\prime}, \gamma_{j}^{\prime\prime\prime}\\ j\in[l+1]}} \Bigg [ \frac{1}{2} \sum_{j\in[n]} \gamma_{j} \Big ( e_{L_2,j}+o_{L_1,j} +1\Big )  
    \nonumber
    \\
    &
    +\sum_{j\in[l+1]} \Big ( \gamma_{j}^{\prime} + \gamma_{j}^{\prime\prime}  +\gamma_{j}^{\prime\prime\prime}  \Big )\Bigg ], \label{eq:pr:class-B-rate-1}
\end{align}
subject to the following constraints
\vspace{-1ex}
\begin{align} 
       \Big ( \sum\limits_{j\in P_i} \gamma_{j}   \Big ) + \gamma_{\frac{i+1}{2}}^{\prime}&=1,  \ \ \ \ \ \ \ \ \ \ \ \ \forall i\in O_{L_1}, \label{eq:class-B-FPCC-con-1}
       \\
       \Big ( \sum\limits_{j\in P_i} \gamma_{j}   \Big ) + \gamma_{\frac{i-2l}{2}}^{\prime\prime}&=1,  \ \ \ \ \ \ \ \ \ \ \ \ \forall i\in E_{L_2}, \label{eq:class-B-FPCC-con-2}
       \\
       \Big ( \sum\limits_{j\in P_i} \gamma_{j}   \Big )&=1,  \ \ \ \ \ \ \ \ \ \ \ \ \forall i\in E_{L_1}\cup O_{L_2}, \label{eq:class-B-FPCC-con-3}
       \\
       \gamma_{i}^{\prime} + \gamma_{i}^{\prime\prime} + \gamma_{i}^{\prime\prime\prime}&=1, \ \ \ \ \ \ \ \ \ \ \ \ \ \forall i\in L_3. \label{eq:class-B-FPCC-con-4}
\end{align}
Now, we add all the constraints in \eqref{eq:class-B-FPCC-con-1} and \eqref{eq:class-B-FPCC-con-2} with each other, which results in

\begin{align}
    2&(l+1)=
    \sum_{i\in {O_{L_1}}} \sum_{j\in P_i} \gamma_{j} + \sum_{i\in O_{L_1}} \gamma_{\frac{i+1}{2}}^{\prime} 
    \nonumber
    \\
    &
    + \sum_{i\in {E_{L_2}}} \sum_{j\in P_i} \gamma_{j}  +  \sum_{i\in E_{L_2}} \gamma_{\frac{i-2l}{2}}^{\prime \prime}  
    \nonumber
    \\
    &
    = \sum_{j\in[n]} \gamma_{j} o_{L_1,j} +  \sum_{i\in[l+1]} \gamma_{j}^{\prime} +
    \sum_{j\in[n]} \gamma_{j} e_{L_2,j}
    +\sum_{i\in[l+1]} \gamma_{j}^{\prime \prime} \nonumber
    \\
    &
    = \sum_{j\in[n]} \gamma_{j}\Big ( e_{L_2,j}+o_{L_1,j} \Big ) +  \sum_{i\in[l+1]} \Big ( \gamma_{j}^{\prime} + \gamma_{j}^{\prime \prime} \Big ). \nonumber
\end{align}
So, we have
\begin{equation}
    \frac{1}{2}\sum_{j\in[n]} \gamma_{j}\Big ( e_{L_2,j}+o_{L_1,j} \Big ) = (l+1) - \frac{1}{2} \sum_{i\in[l+1]} \Big ( \gamma_{j}^{\prime} + \gamma_{j}^{\prime \prime} \Big ). \label{eq:pr:class-B-sum}
\end{equation}
Now, using \eqref{eq:pr:class-B-sum}, for the broadcast rate of the FPCC scheme in \eqref{eq:pr:class-B-rate-1}, we have
\begin{align}
    \beta_{\text{FPCC}}&(\mathcal{I}_{7}(l))\geq (l+1) + \min_{\substack{\gamma_{j}^{\prime}, \gamma_{j}^{\prime\prime}, \gamma_{j}^{\prime\prime\prime}\\ j\in[n]}} \Bigg [\sum_{j\in[n]} \gamma_{j} 
    \nonumber
    \\
    &+ \frac{1}{2} \sum_{i\in[l+1]} \Big ( \gamma_{j}^{\prime} + \gamma_{j}^{\prime \prime}\Big )+ \sum_{i\in[l+1]} \gamma_{j}^{\prime\prime\prime} \Bigg ]  
    \nonumber
    \\
    &
    \geq (l+1)+ \frac{1}{2}  \min_{\substack{\gamma_{j}^{\prime}, \gamma_{j}^{\prime\prime}, \gamma_{j}^{\prime\prime\prime}\\ j\in[n]}} \sum_{i\in[l+1]} \Big ( \gamma_{j}^{\prime} + \gamma_{j}^{\prime \prime}+ \gamma_{j}^{\prime\prime\prime}\Big )  \nonumber
    \\
    &=(l+1)+\frac{1}{2}(l+1)
    \label{eq:pr:class-B-last}
    \\
    &=\frac{3}{2}(l+1), \nonumber
\end{align}
where \eqref{eq:pr:class-B-last} is due to \eqref{eq:class-B-FPCC-con-4}. This completes the proof.
\end{proof}

Now, we provide the following definitions, which will be used in the proof of Lemma \ref{lem2:prop7}. 

\begin{defn}[$V_{i}^{\text{out}}, V_{i}^{\text{in}}$: Outgoing and Incoming Neighbor Set of $i$]
In an arbitrary graph $\mathcal{G}_{\mathcal{I}}=([m], E)$, for any vertex $i\in [m]$, we define an outgoing neighbor set as $V_{i}^{\text{out}}=\{j\in [m]\backslash \{i\}, (i,j)\in E\}$ and an incoming neighbor set as $V_{i}^{\text{in}}=\{j\in [m]\backslash \{i\}, (j,i)\in E\}$.
\end{defn}
For any vertex $i$, its neighbor set can be partitioned into three subsets $V_{i,1}=V_{i}^{\text{out}}\cap V_{i}^{\text{in}}$, $V_{i,2}=V_{i}^{\text{out}}\backslash V_{i,1}$, and $V_{i,3}=V_{i}^{\text{in}}\backslash V_{i,1}$. It is obvious that vertex $i$ forms a pairwise clique with the vertices in $V_{i,1}$.

\begin{defn}[Clique-outgoing Vertex]
Vertex $i\in [m]$ is said to be a clique-outgoing vertex, if $V_{i,2}=\emptyset$. 
\end{defn}

\begin{defn}[Clique-incoming Vertex]
Vertex $i\in [m]$ is said to be a clique-incoming vertex, if $V_{i,3}=\emptyset$.
\end{defn}

\begin{lem} \label{lem2:prop7}
For the class-$\mathcal{I}_{7}$ index coding instances, we have $\beta_{\text{ICC}}(\mathcal{I}_{7}(l))=\beta_{\text{FPCC}}(\mathcal{I}_{7}(l))$.
\end{lem}

\begin{proof}
We show that for the class-$\mathcal{I}_{7}$ index coding instances, each ICC-structured subgraph is reduced to a clique, and since the broadcast rate of the ICC and FPCC schemes  for a clique is equal, this will complete the proof.\\
First, it can be verified that (i) the vertices $i\in O_{L_1}\cup E_{L_2}$ are all clique-outgoing vertices and (ii) the vertices $i\in K_3$ are all clique-incoming vertices. Now, suppose that subgraph $M\subseteq[m=5l+3]$ is an ICC-structured subgraph, subset $J\subseteq M$ is its inner vertex set and $i\in J$. Now, we note the followings.
\begin{itemize}[leftmargin=*]
    \item Let vertex $i$ be a clique-outgoing vertex. (i) Since vertex $i$ forms a clique with all of its outgoing neighbors $j_1\in V_{i}^{\text{out}}$, to meet the $J$-cycle condition, we must have $j_1\not\in M\backslash J$. (ii) Since every path from vertex $i$ will include at least one of its outgoing neighbors, then, to meet the inner vertex set definition, none of the vertices in $j_2\in M\backslash V_{i}^{\text{out}}$ can be inside the inner vertex set, i.e., $j_2\not\in J$. This means that if a clique-outgoing vertex is inside the vertex set $i\in J$, then only its outgoing neighbors, which forms a pairwise clique with them, are allowed to be inside the inner vertex set $J$. 
    \item Let vertex $i$ be a clique-incoming vertex. With the same arguments, (i) for all of its incoming neighbors $j_1\in V_{i}^{\text{in}}$, we must have $j_1\not\in M\backslash J$ and (ii) only its incoming neighbors, which forms a pairwise clique with them, are allowed to be inside the inner vertex set.
    \item The vertices inside $E_{L_1}\cup O_{L_2}$ form a clique with each other. Let $i_1, i_2\in E_{L_1}\cup O_{L_2}$. Now, if $i_1\in J$, then to meet the $J$-cycle condition, we must have $i_2\not\in M\backslash J$. Note, if two vertices $i_1\in E_{L_1}\cup O_{L_2}$ and $i_3\in O_{L_1}\cup E_{L_2}\cup K_3$ are inside the inner vertex set, then in the previous cases, we already proved that $i_1$ and $i_3$ must form a clique with each other.
\end{itemize}
Thus, this shows that any two vertices $j_1\in J$ and $j_2\in J$ must form a clique with each other to meet the definition of the ICC-structure. In other words,
we showed that all the vertices inside the inner vertex set must form a pairwise clique with each other. This means that any ICC-structured subgraph will be reduced to a clique. This completes the proof.
\end{proof}

\end{document}